\def\llncs{0}
\def\anonymous{0}
\def\comments{0}

\ifnum\llncs=1
\documentclass{llncs}
\usepackage{amsmath,amssymb}
\else
\documentclass[11pt]{article}
\usepackage{amsmath,amssymb,amsthm}
\fi

\usepackage{enumitem}
\ifnum\llncs=0
\usepackage[a4paper,margin=1in]{geometry}
\fi
\usepackage{braket}
\usepackage{xcolor}

\usepackage{header}

\def\bra#1{\mathinner{\langle{#1}|}}
\def\ket#1{\mathinner{|{#1}\rangle}}
\def\braket#1{\mathinner{\langle{#1}\rangle}}
\def\ketbra#1#2{\mathinner{|{#1}\rangle\!\langle{#2}|}}

\usepackage[dvipsnames]{xcolor}
\ifnum\comments=0
\newcommand{\Zoenote}[1]{%
  \textcolor{RoyalBlue}{\textbf{[Zoe:} #1\textbf{]}}%
}
\newcommand{\ynote}[1]{%
  \textcolor{purple}{\textbf{[Yiming:} #1\textbf{]}}%
}

\newcommand{\bnote}[1]{%
  \textcolor{brown}{\textbf{[Boyang:} #1\textbf{]}}%
}
\else
\newcommand{\bnote}[1]{}
\newcommand{\ynote}[1]{}
\newcommand{\Zoenote}[1]{}
\fi

\newcommand{\pspace}{\mathsf{PSPACE}}

\newcommand{\Gen}{\mathrm{Gen}}
\newcommand{\Sim}{\mathrm{Sim}}
\newcommand{\Var}{\operatorname*{Var}}

\newcommand{\calA}{\mathcal{A}}
\newcommand{\calS}{\mathcal{S}}
\newcommand{\calO}{\mathcal{O}}
\newcommand{\Dec}{\mathsf{Ext}}
\newcommand{\Decfull}{\mathsf{ExtFull}}
\newcommand{\Samp}{\mathsf{Samp}}
\newcommand{\Ver}{\mathsf{Ver}}
\newcommand{\Help}{\mathsf{Help}}

\usepackage{authblk}
\usepackage{orcidlink}

\ifnum\llncs=0
\usepackage{titling}
\predate{}
\postdate{}
\fi

\title{The Breakdown of Classical Minicrypt Equivalences in the Quantum-Computation Classical-Communication Model}
\ifnum\anonymous=0
\author[1]{Boyang Chen}
\author[2]{Yiming Wang}
\author[1]{Ziyi Xie}

\affil[1]{Department of Computer Science and Technology, Tsinghua University}
\affil[2]{Paul G. Allen School of Computer Science \& Engineering, University of Washington}

\ifnum\llncs=0
\date{}
\fi

\else
\author{}
\institute{}
\fi
\begin{document}
\pagestyle{plain}
\maketitle
\thispagestyle{plain}

\begin{abstract}
Classically, it is well-known that several fundamental cryptographic primitives, including one-way functions, pseudorandom number generators, commitments, and signatures, characterize the same cryptographic world, which is known as ``Minicrypt''. In this paper, we investigate to what extent this picture persists in the quantum-computation classical-communication (QCCC) setting, where parties may perform quantum local computation, but communicate through only classical messages. 

We demonstrate that the classical Minicrypt landscape breaks down in the QCCC model by oracle separations. We construct two oracle worlds where efficiently verifiable one-way puzzles (EV-OWPuzz) exist. By a known equivalence, QCCC one-time signatures also exist in both worlds. In the first, one-way functions do not exist, even if we allow quantum pseudodeterministic evaluation. In the second, QCCC bit commitments do not exist. Thus, the classical Minicrypt reductions from signatures to one-way functions and bit commitments have no fully black-box counterparts in the QCCC setting.
Our proofs develop techniques for analyzing the random phase states, including a concentration theorem and an LOCC decoupling theorem.

As a separate result, we establish quantum sum-binding for the NOVY protocol, settling the question of its post-quantum security. We also generalize it to a fully black-box construction of a QCCC commitment from EV-OWPuzz with nearly uniform puzzle distribution, complementing our oracle separation results.
\end{abstract}
\ifnum\llncs=0
\tableofcontents
\fi
\section{Introduction}
Understanding the minimal assumptions for cryptography to exist is one of the fundamental tasks in cryptography.
Classically, one-way functions (OWFs) are regarded as a foundational primitive.
In a sequence of celebrated works, one-way functions were shown to be equivalent to various other primitives, such as pseudorandom generators (PRGs)~\cite{HILL}, (interactive) commitment schemes~\cite{Naor-commitment,SHcommit}, digital signatures~\cite{Rompel-signature}, and secret-key encryption schemes (SKEs)~\cite{GGM84,IL89}.
Although these primitives have rather different appearances, their existence is based on the same computational assumption.
The cryptographic world collectively characterized by these primitives is known as ``Minicrypt''~\cite{IL89,Imp95,HILL}.

The landscape is different in quantum cryptography.
A sequence of works introduced several intrinsically quantum primitives, which includes PRS~\cite{JLS}, OWSG~\cite{MY22a}, PRFS~\cite{AGQY22}, and QEFID~\cite{CGG-EV-OWPuzz}.
Subsequent oracle separations showed that some quantum primitives could exist even if $\mathbf{P}=\mathbf{NP}$ (where, in particular, post-quantum OWFs do not exist)~\cite{KQST,KQT}, suggesting that these primitives may rely on assumptions weaker than the OWF assumption underpinning Minicrypt.

However, the minimal assumption for quantum cryptography remains unsettled.
There has been strong evidence that EFI pair might be the minimal primitive in quantum cryptography~\cite{BCQ23,KT23}.
EFI pairs have been shown to be equivalent to various other quantum cryptographic primitives, including but not limited to quantum bit commitments and secure multiparty computation~\cite{AQY22,BCQ23,BJ24,CCC+26}.
Nevertheless, implementing protocols that rely on EFI pairs requires quantum communication, and thus seems difficult in the next few years.

Motivated by the difficulty of sending quantum messages, recent works have studied primitives in the quantum-computation classical-communication (QCCC) model~\cite{ACC+22,CLM23,BBSS23,ALY24,KT23,CGG-EV-OWPuzz,BMM+25,CountCrypt}. In this model, parties perform quantum computation locally but exchange only classical messages, avoiding the need for quantum communication infrastructures.

Unsurprisingly, QCCC primitives might still be weaker than their classical counterparts, and QCCC cryptography might have a different landscape as well.
The reason for this might be related to the difference between classical and quantum sampling.
A classical sampler can be written as a deterministic function of an explicit random seed, but there is not an analogous way to specify the randomness for quantum samplers.
In other words, it is hard to ``pull out'' the randomness in a quantum computer.
Consequently, one-wayness arising from a quantum sampler need not directly yield an OWF.

In order to capture this sampler-based form of one-wayness, the one-way puzzle (OWPuzz) and its efficiently verifiable variant (EV-OWPuzz) were introduced~\cite{KNY23,KT23,CGG-EV-OWPuzz,KT25}.
Both notions are implied by many natural QCCC primitives and admit amplification, combiners, and universal constructions~\cite{KT23,CGG-EV-OWPuzz}.
In particular, EV-OWPuzzs are proposed as a leading candidate for the central primitive in QCCC cryptography, and they are shown to be equivalent to QCCC one-time signatures~\cite{CGG-EV-OWPuzz}.

However, the power of EV-OWPuzzs in QCCC cryptography remains unclear, and it motivates the following open question raised in~\cite{CGG-EV-OWPuzz}:
\begin{center}
    \emph{Can we build other QCCC primitives besides one-time signatures from EV-OWPuzzs?}
\end{center}
Considering the elegant picture in classical Minicrypt, where OWFs, PRGs, commitments, signatures, and SKEs are all equivalent by fully black-box reductions, it is also natural to ask:
\begin{center}
    \emph{How much of the classical Minicrypt landscape survives in the QCCC setting?}
\end{center}

\paragraph{Our results.}
We present negative evidence for both questions.
Focusing on quantum-computable OWFs and QCCC bit commitments, whose classical counterparts are central in the Minicrypt landscape, we prove that it is impossible to construct them from EV-OWPuzz in a fully black-box manner.
In fact, our first separation applies even to pseudodeterministic OWFs (pd-OWFs), a relaxation of OWFs in which evaluation may be randomized but yields a canonical output with high probability on all but a inverse-polynomial fraction of inputs.

\begin{theorem}[Informal]\label{thm:informal-1}
    There is a unitary oracle relative to which an EV-OWPuzz exists, but pseudodeterministic one-way functions (and hence one-way functions) do not.
\end{theorem}

\begin{theorem}[Informal]\label{thm:informal-2}
    There is a unitary oracle relative to which an EV-OWPuzz exists, but QCCC bit commitment schemes do not.
\end{theorem}

We would like to emphasize that our oracles allow controlled access and access to the inverses, conjugates, and transposes.
Thus, our results hold in the unitary-oracle access model represented by~\cite{Zha25-model-unitary-oracles}.

Since fully black-box constructions relativize~\cite{RTV04-fully-bb,CCS25}, our results rule out such constructions of either primitive from EV-OWPuzz (and thus OWPuzz).

The main technical ingredient of Theorems~\ref{thm:informal-1}~and~\ref{thm:informal-2} is a concentration lemma and a LOCC indistinguishability lemma regarding the superposition of states with random phases, which is introduced in Appendix~\ref{sec:phase-concentration} and may be of independent interest.

While exploring the possibility of a fully black-box construction of a QCCC commitment scheme from EV-OWPuzz, we considered extending the commitment protocol of~\cite{SHcommit} to the quantum setting.
The NOVY commitment protocol~\cite{NOVY} is a construction of a statistically hiding commitment from one-way permutation, which constitutes the conceptual cornerstone of~\cite{SHcommit}.
As a complementary result, in Section~\ref{sec:appendix-uniform}, we prove the quantum security of the NOVY protocol, which naturally extends to a construction of a QCCC commitment scheme from EV-OWPuzz with a (nearly) uniform puzzle distribution.
\begin{theorem}[Informal]\label{thm:informal-NOVY-security}
    The NOVY commitment protocol~\cite{NOVY} constructed from a post-quantum one-way permutation is post-quantum secure, in the sense that it is statistically hiding and computationally quantum sum-binding. Moreover, the NOVY protocol extends to a construction of a QCCC commitment based on EV-OWPuzz with (nearly) uniform puzzle distribution. 
\end{theorem}
This identifies a natural boundary of our oracle separation and highlights the following obstruction in constructing QCCC commitments from EV-OWPuzz: not all puzzles accepted by the verifier are honestly generated by the puzzle--key sampler.

The post-quantum security of the NOVY commitment in itself also has many important implications, among which is a construction of oblivious state preparation from trapdoor permutation, which in turn implies a series of applications including proof of quantumness, qubit test, certified randomness, blind classical delegation of quantum computation, classical verification of quantum computation~\cite{MY23,BK25}.

\paragraph{Related work.}

\begin{figure}[!t]
  \centering
  \resizebox{\linewidth}{!}{

\definecolor{priorred}{RGB}{190,35,35}
\definecolor{ourblue}{RGB}{25,90,185}

\pgfdeclarelayer{edges}
\pgfsetlayers{edges,main}

\tikzset{
  primitive/.style={
    draw=black!45,
    fill=white,
    rounded corners=2pt,
    line width=.55pt,
    minimum width=28mm,
    minimum height=8mm,
    inner sep=2.5pt,
    align=center,
    font=\footnotesize
  },
  reduction/.style={
    -{Latex[length=2.4mm,width=1.7mm]},
    draw=black,
    line width=.9pt
  },
  equivalence/.style={
    <->,
    >={Latex[length=2.4mm,width=1.7mm]},
    draw=black,
    line width=.9pt
  },
  separation/.style={
    -{Latex[length=2.4mm,width=1.7mm]},
    line width=.95pt,
    dash pattern=on 3pt off 2pt
  },
  priorsep/.style={separation,draw=priorred},
  oursep/.style={
    separation,
    draw=ourblue,
    line width=1.2pt,
    dash pattern=on 6pt off 2.5pt
  },
  edge label/.style={
    fill=white,
    inner sep=1.2pt,
    font=\scriptsize,
    text=black!75
  },
  prior label/.style={edge label,text=priorred},
  our label/.style={edge label,text=ourblue,font=\scriptsize\bfseries}
}


\begin{tikzpicture}
  \node[primitive] (qprg) at (0,4.8)
    {$\mathsf{pd\text{-}PRG}_{1/\mathrm{poly}}$};
  \node[primitive] (nicom) at (5.8,3.8)
    {$\mathsf{NI\text{-}QCCC\ COM}$};

  \node[primitive] (qpoly) at (-2.9,2.45)
    {$\mathsf{pd\text{-}OWF}_{1/\mathrm{poly}}$};

  \node[primitive] (ots) at (-5.8,0)
    {$\mathsf{QCCC\ OTS}$};
  \node[primitive] (ev) at (0,0)
    {$\mathsf{EV\text{-}OWPuzz}$};
  \node[primitive] (qccc) at (5.8,0.75)
    {$\mathsf{QCCC\ COM}$};

  \node[primitive] (owp) at (0,-2.8)
    {$\mathsf{OWPuzz}$};
  \node[primitive] (efi) at (0,-5.4)
    {$\mathsf{EFI}$};

    \node[primitive] (pru) at (5.8, -2.3) {$\mathsf{PRU}$};

  \begin{pgfonlayer}{edges}

    \draw[reduction]
      (qprg.south west) --
      node[edge label,above right,pos=.50] {\cite{BBO+bot-OWF}}
      (qpoly.north);

    \draw[reduction]
      (qprg.south) --
      node[edge label,right,pos=.43] {\cite{CGG-EV-OWPuzz}}
      (ev.north);

    \draw[reduction]
      (pru.west) --
      node[edge label,below,pos=.50]
      {\cite{KT23}}
      (owp.east);

    \draw[reduction,rounded corners=7pt]
      (qprg.east) --
      node[edge label,above,pos=.52] {\cite{ALY24}}
      (7.9,4.8) -- (7.9,0.75) -- (qccc.east);

    \draw[reduction]
      (nicom.south west) --
      node[edge label,above right,pos=.48] {\cite{CGG-EV-OWPuzz}}
      (ev.north east);

    \draw[reduction]
      (nicom.south) --
      node[edge label,right] {def.}
      (qccc.north);

    \draw[equivalence]
      (ots.east) --
      node[edge label,above] {\cite{KT23}}
      (ev.west);

    \draw[reduction]
      (ev.south) --
      node[edge label,right] {def.}
      (owp.north);

    \draw[reduction]
      (qccc.south west) --
      node[edge label,above right,pos=.55] {\cite{CountCrypt}}
      (owp.north east);

    \draw[reduction]
      (owp.south) --
      node[edge label,right] {\cite{KT23,CGG-EV-OWPuzz}}
      (efi.north);

    \draw[priorsep]
      (qccc.south west)
      to[out=220,in=-40,looseness=1.12]
      node[prior label,above,pos=.52] {\cite{CountCrypt}}
      (ev.south east);

    \draw[priorsep]
      (owp.west)
      to[out=160,in=-160,looseness=1.25]
      node[prior label,left,pos=.55] {\cite{CGG-EV-OWPuzz}}
      (ev.west);

    \draw[priorsep]
      (pru.north)
      to[out=90,in=-90]
      node[prior label,left,pos=.50]{\cite{AGL25}}(qccc.south);

    \draw[oursep]
      (ev.north west)
      to[out=125,in=-55]
      node[our label,below left,pos=.49] {this work}
      (qpoly.south);

    \draw[oursep]
      (ev.north east)
      to[out=35,in=145,looseness=1.12]
      node[our label,below,pos=.50] {this work}
      (qccc.north west);

    \draw[oursep]
      (ev.north east)
      to[out=35,in=-55,looseness=0.5]
      node[our label,above,pos=.50] {this work}
      (qprg.south east);
  \end{pgfonlayer}

  \begin{scope}[every node/.style={anchor=west,font=\scriptsize}]
    \draw[reduction] (2.4,-4.4) -- (3.2,-4.4);
    \node at (3.4,-4.4) {reductions};

    \draw[priorsep] (2.4,-4.9) -- (3.2,-4.9);
    \node[text=priorred] at (3.4,-4.9) {known separations};

    \draw[oursep] (2.4,-5.4) -- (3.2,-5.4);
    \node[text=ourblue] at (3.4,-5.4) {separations proved in this work};
  \end{scope}

\end{tikzpicture}

  }
  \caption{Known black-box reductions and oracle separations among QCCC primitives.}
  \label{fig:primitive-relations}
\end{figure}

Figure~\ref{fig:primitive-relations} provides an intuitive visualization of this work and related results.

Besides the first open question mentioned above, which asks whether other QCCC primitives can be constructed from EV-OWPuzz, \cite{CGG-EV-OWPuzz} also posed the reverse question: can one construct EV-OWPuzz from QCCC commitments?
While our result (Theorem~\ref{thm:informal-2}) gives a negative answer to the former, \cite{CountCrypt} establishes the reverse separation by constructing an oracle relative to which $\mathbf{BQP}=\mathbf{QCMA}$ and hence EV-OWPuzz does not exist, while QCCC primitives, including a QCCC commitment scheme, exist.

Building on their earlier work~\cite{AGL24}, Ananth, Gulati, and Lin~\cite{AGL25} developed a framework for separating PRUs from QCCC primitives.
In particular, they used an LOCC indistinguishability technique similar to one of our techniques (Theorem~\ref{thm:LOCC-decoupling}) to construct an oracle relative to which a PRU exists but QCCC commitments and QCCC key agreement do not.
Although they use a similar technique, their separation does not subsume ours, because EV-OWPuzz cannot be constructed from PRUs in a fully black-box manner~\cite{Kre21,CGG-EV-OWPuzz}.
Moreover, \cite{AGL25} allows only unitary access, while our result adopts the more general oracle access model that is favored in~\cite{Zha25-model-unitary-oracles}.
Therefore, their results, including their LOCC indistinguishability lemmas on Haar-random unitaries, are incomparable to ours.

It is worth emphasizing that the error regime matters for PRGs and OWFs, as shown by~\cite{BHMV25,ALY24}.
They consider the separation of \emph{negligible}-error quantum PRGs and pseudodeterministic quantum PRGs, showing that quantum pseudodeterminism is powerful.
Our results are incomparable to theirs: we separate EV-OWPuzz and pseudodeterministic quantum PRGs, which implies that quantum sampling might be harder than a pseudodeterministic evaluation.

\cite{BNY25} studied quantum input sampling and $\bot$-pseudodeterminism and established several related black-box separations.
They ruled out a fully black-box construction of a quantum PRG that is pseudodeterministic on every input from EV-OWPuzz, even with inverse access.
This separation is strengthened by~\cite{BHMV25} and Theorem~\ref{thm:informal-1} under progressively weaker pseudodeterminism guarantees.
They also ruled out a fully black-box construction of a $\bot$-PRG from EV-OWPuzz under CPTP access.
This separation is complementary to ours because it concerns a different relaxation of PRGs.
Their commitment-related separations run in the opposite direction from Theorem~\ref{thm:informal-2}: their oracles contain QCCC commitments but not PRGs or $\bot$-PRGs, whereas our oracle contains an EV-OWPuzz but no QCCC commitment.

In addition to those mentioned above, there are various other black-box oracle separations.
\cite{Kre21} demonstrated a quantum oracle relative to which $\mathbf{BQP}=\mathbf{QMA}$, but PRSs exist.
This technique was used to show the impossibility of constructing post-quantum OWFs~\cite{Kre21}, EV-OWPuzzs~\cite{CGG-EV-OWPuzz}, digital signatures~\cite{CM24}, and logarithmic-output PRSs~\cite{BM24} from PRSs in a fully black-box way.
The common Haar-random state oracles were used to show the impossibility of a fully black-box construction of (multi-copy) PRSs and OWSGs from single-copy PRSs~\cite{CCS25,BCN25,AGL24}, and also the impossibility of stretching single-copy PRSs in a fully black-box way~\cite{1PRS-stretching}.
A quantum oracle that shares some similarity with ours was used in~\cite{BMM+25} to show the black-box impossibility of constructing OWSG and quantum money from QEFID.
\footnote{For a comprehensive list of black-box constructions and separations, see \href{https://sattath.github.io/microcrypt-zoo/}{Microcrypt-zoo}.}

\paragraph{AI disclaimer.}
The authors originated the central ideas and proof outlines, while ChatGPT 5.6 Sol Ultra supported the exploration of these directions.
The authors verified and formalized all arguments independently and are responsible for them.
The final manuscript is organized by the authors, with ChatGPT 5.6 Sol's and ChatGPT 6 Astra's help in polishing.

\section{Technical Overview}
Our two main theorems, Theorems~\ref{thm:informal-1} and~\ref{thm:informal-2}, require different techniques, so we divide the technical overview into two subsections. Section~\ref{sec:tech-overview-owf} describes the proof idea for Theorem~\ref{thm:informal-1} (separating pd-OWFs from EV-OWPuzz).
Section~\ref{sec:tech-overview-novy} sketches the proof for Theorem~\ref{} (the construction of a QCCC commitment scheme from EV-OWPuzz with nearly uniform puzzle distribution)
Section~\ref{sec:tech-overview-commitment} describes the proof idea for Theorem~\ref{thm:informal-2} (separating QCCC commitments from EV-OWPuzz).

\subsection{The separation of pd-OWF from EV-OWPuzz}\label{sec:tech-overview-owf}

\paragraph{What changes in the QCCC model.}
The classical equivalence between one-way functions, signatures, and commitments relies on more than the ability to generate hard instances. Recall that classically, we can always ``pull out'' the randomness: give the algorithm a random tape $r$, and its output becomes a deterministic function of $r$. This allows switching between classical sampling and function evaluation naturally.

For a quantum sampler, however, this connection is less straightforward: there is no general efficient analogue of fixing a classical random tape. In particular, there is no efficient way to force a quantum sampler to produce a specific desired outcome.

To capture this sampling-based one-wayness, \cite{CGG-EV-OWPuzz} proposed EV-OWPuzz as a central primitive for an analog of Minicrypt in the QCCC setting. Informally, an EV-OWPuzz consists of an efficient quantum sampler $\mathsf{Samp}$ and an efficient quantum verifier $\mathsf{Ver}$ satisfying the following properties:
\begin{itemize}
    \item The sampler $\mathsf{Samp}$ outputs a puzzle--key pair $(s,k)$ that passes verification by $\mathsf{Ver}$ with overwhelming probability.\footnote{Throughout this paper, the order of this puzzle--key pair is different with~\cite{KT23,CGG-EV-OWPuzz}.}
    \item Given a puzzle $s$ sampled by $\mathsf{Samp}$, no efficient quantum algorithm can find a key $k'$ that passes verification with non-negligible probability.
\end{itemize}

To separate EV-OWPuzz from other primitives, our oracles essentially provide such an EV-OWPuzz: one part allows efficient quantum sampling of puzzle--key pairs, while another checks whether a pair is valid. Intuitively, these operations let us obtain the desired samples without giving us reproducible access to the information hidden in the oracle.

\paragraph{The oracle construction.}
The basic construction starts with a uniformly random function $g:\{0,1\}^n\to\{0,1\}^{3n}$. Rather than providing direct evaluation access to $g$, we provide the corresponding graph state:
\[
    \frac{1}{\sqrt{2^n}}
    \sum_{x\in\{0,1\}^n}\ket{x}\ket{g(x)}\,.
\]
Measuring it in the computational basis produces a pair $(x,g(x))$. Intuitively, this lets us extract random pieces of information about $g$, but gives us little control over which piece we obtain: on each measurement, the probability of hitting a specified pair $(x,g(x))$ is merely $2^{-n}$. 
We also provide a verifier $f$ that, given input $(x, y)$, checks whether the proposed key $y$ equals $ g(x)$. This verifier does not directly reveal $g(x)$ either: on its own, finding $g(x)$ for a specified $x$ amounts to searching for the accepting input of $f(x,\cdot)$. With only polynomially many queries, any algorithm succeeds with only negligible probability, as captured by a BBBV-type argument.

\paragraph{Breaking pseudodeterministic OWFs.}
To rule out pd-OWFs, the goal is to show that this oracle setup does not provide a pseudodeterministic way of extracting the hidden information about $g$. Intuitively, although an algorithm can easily obtain random pairs $(x,g(x))$, a pseudodeterministic algorithm must repeatedly produce the same output on a given input, and therefore cannot rely on these random samples. An independent oracle, such as a $\mathbf{PSPACE}$-complete language, can then be used to invert the candidate pd-OWF, while it does not affect the security of EV-OWPuzz.

Heuristically, it seems that the only way of extracting information from the graph state is by performing computational-basis measurements, producing no more than random pairs $(x,g(x))$ that should not be able to be utilized by a pseudodeterministic algorithm. The main question is therefore whether a quantum algorithm can exploit the coherence of the graph state to obtain information that ordinary computational-basis measurements would not provide. This is where an observation from~\cite{AG25} comes into play: superpositions of states with random phases can be simulated using a sampler of the components.

To leverage this observation, we modify the oracle so that it provides the state
\[
    \ket{\psi}
    =
    \frac{1}{\sqrt{2^n}}
    \sum_{x\in\{0,1\}^n}
    \alpha(x)\ket{x}\ket{g(x)}\,,
\]
where each $\alpha(x)$ is a random phase of magnitude $1$. Measuring $\ket{\psi}$ in the computational basis still yields the same outcome distribution. The random phases, however, allow us to eliminate the coherence between components, as captured by the following lemma.

\begin{lemma}[Informal, adapted from~\cite{AG25}]\label{lem:informal-averaging}
    For every $t\in\mathbb N_+$ and every fixed function $g$, the state
    \[
        \EE_\alpha\left[
            \ketbra{\psi}{\psi}^{\otimes t}
        \right]
    \]
    can be efficiently produced by an algorithm $\mathrm{Sim}$ that obtains samples $(x, g(x))$ for uniformly random $x \in \{0,1\}^n$, but no other information about $g$.
\end{lemma}

This lemma suggests that, after averaging over $\alpha$, access to the graph state provides no more information about $g$ than ordinary computational-basis measurements. Combining this observation with a BBBV-type argument for the verification oracle $f$, it then seems that the output of any pseudodeterministic algorithm must be independent of the oracles, thus breakable by a $\mathbf{PSPACE}$-complete language. The argument now appears complete... \emph{or is it?}

\paragraph{What is missing in the previous argument.}
To show that no pd-OWF exists relative to our oracle, we need to \emph{fix} an oracle and break pd-OWFs relative to that \emph{fixed} oracle. In the previous argument, however, we \emph{averaged} over $\alpha$ (as in Lemma~\ref{lem:informal-averaging}). This does not suffice to break pd-OWFs for a fixed choice of $\alpha$.

Consider a toy example in which the pd-OWF generator somehow gets evaluation access to $\alpha$ and treats $\alpha$ itself as a one-way function. For each fixed $\alpha$, the generator's output is deterministic, but after averaging over $\alpha$, it becomes random. Thus, the averaged experiment can obscure precisely the pseudodeterminism in a fixed oracle world. We therefore need to be cautious: an \emph{averaging} argument, whether over $g$ or over the random phases $\alpha$, is not enough, as it may only rule out candidates that remain pseudodeterministic even \emph{after averaging}, whereas pseudodeterminism is only required relative to a \emph{fixed} oracle.

\paragraph{A concentration theorem on phases.}
Motivated by the discussion in the previous paragraph, we now focus on proving that, for a pseudodeterministic algorithm that only accesses the random-phase graph state $\ket{\psi}$ and the verification function $f$, its principal output does not vanish or change substantially, even after averaging over $\alpha$ and $g$.

Naturally, this leads to a concentration argument. Suppose that, on a given input, the output distribution is \emph{concentrated} around its mean. This means that, for most choices of $\alpha$ and $g$, the output distribution changes very little after averaging. In particular, if the original distribution places almost all its probability on a single output, the averaged distribution must do so as well, with the same principal output. The $\mathbf{PSPACE}$ inverter can then simulate the generator's output averaged over $\alpha$ and $g$. Since the principal output is preserved after averaging, it also succeeds in inverting the pd-OWF relative to a typical fixed choice of $\alpha$ and $g$.

We first present our concentration theorem for the phases $\alpha$, which is one of our theoretical contributions.

\begin{theorem}[Informal]\label{thm:informal-phase-concentration}
    Fix $t\in\mathbb N_+$ and a function $g$. For any fixed algorithm $\mathcal A$ running on $t$ copies of $\ket{\psi}$, let $\mathcal D_\alpha$ denote its classical output distribution. Then,
    \[
        \EE_\alpha\left[
            \left\|
                \mathcal D_\alpha-\EE_{\alpha'}\mathcal D_{\alpha'}
            \right\|_2^2
        \right]
        \leq O\!\left(\frac{t}{2^n}\right)\,,
    \]
    where $\|\cdot\|_2$ denotes the $\ell_2$-norm of distributions.
\end{theorem}

In natural language, Theorem~\ref{thm:informal-phase-concentration} says that, for any algorithm using polynomially many copies of $\ket{\psi}$, its output distribution is tightly concentrated around its  average over phases. This gives the first part of the concentration we need.

\paragraph{A concentration argument on $g$.}
Theorem~\ref{thm:informal-phase-concentration} shows that the output distribution is concentrated over $\alpha$, while Lemma~\ref{lem:informal-averaging} shows that the average state of copies of $\ket{\psi}$ can be simulated using uniform random samples $(x,g(x))$. Thus, it remains to examine whether a pseudodeterministic algorithm can make essential use of such sampling access to $g$ or the access to the verification oracle $f(x,\cdot)=\mathbf{1}_{g(x)}$.

We formalize the intuition that it cannot, using the Efron--Stein inequality informally stated as follows.

\begin{lemma}[Efron--Stein inequality, informal]
    Let $F$ be a function of independent random variables $X_1,\ldots,X_n$. Then, the variance of $F$ is bounded by the sum of the sensitivities for each $X_i$ (that is, the expected squared changes when a single $X_i$ is independently resampled).
\end{lemma}

To apply this inequality, we consider the probability of the pd-OWF generator outputting a particular value, where we treat each function value $g(x)$ as an independent random variable. 
Resampling a single $g(x)$ has little effect: the sampler is unlikely to return this particular $x$, while detecting the change through $f$ essentially requires searching for a hidden point, which is difficult by a BBBV-type argument. Crucially, Efron--Stein sums the \emph{squared} effects of these changes, allowing the total to remain small even after summing over all inputs $x$. This gives concentration of the output probabilities with respect to $g$. A natural generalization to the entire output probability vector yields the desired concentration.

We point out that an $\ell_2$ concentration breaks pd-OWF since it yields a concentration on the probability of each individual output. This means, if an algorithm has a principal output (with output probability larger than constant), the $\ell_2$ concentration yields a simulation algorithm whose output probability of the principal output is still greater than constant. However, $\ell_2$ concentration does not provide enough power to break our proposed EV-OWPuzz. This is because the sampler samples each output with exponentially small probability, and thus a simulator provides little information on the puzzle--key correspondence, even if the simulator is guaranteed to have a close $\ell_2$ distance in output distribution.

\paragraph{Lifting to a better oracle model.}
We have now described the main idea for showing that an EV-OWPuzz exists relative to $\ket{\psi}$ and $f$, while pd-OWFs do not. However, as emphasized by~\cite{Zha25-model-unitary-oracles}, a quantum oracle is desired to be a unitary with controlled access and access to its inverse and conjugate. We now explain how to lift our construction to this model.

First, we modify the oracle to provide, instead of $\ket{\psi}$, the unitary $U_{\ket{\psi}}$ that swaps $\ket{0}\ket{0}$ and $\ket{1}\ket{\psi}$. Using the standard simulation technique~\cite{JLS,Goldin-Zhandry}, queries to $U_{\ket{\psi}}$ can be simulated using copies of the state
$\ket{\psi^-}=(\ket{0}\ket{0}-\ket{1}\ket{\psi})/{\sqrt{2}}$.
By extending Lemma~\ref{lem:informal-averaging} and Theorem~\ref{thm:informal-phase-concentration} to $\ket{\psi^-}$, we can therefore carry the argument over to the unitary oracle.

The unitary $U_{\ket{\psi}}$ is self-inverse, and controlled access can also be handled by the same simulation technique~\cite{JLS}. Thus, the only remaining issue is conjugate access. To address it, we modify $\alpha$ to use random signs $\pm1$ instead of continuous complex phases. This makes $U_{\ket{\psi}}$ self-conjugate, but comes at a technical cost: random signs make the phase-concentration argument (Theorem~\ref{thm:informal-phase-concentration}) more delicate. Continuous phases distinguish the number of occurrences of each graph entry, whereas $\pm1$ phases distinguish only whether that number is even or odd. Consequently, additional cross terms survive the averaging. Our proof must carefully control these surviving terms, rather than rely on their cancellation.

\paragraph{Outline of the proof.} Section~\ref{sec:owf-separation} is dedicated to the proof of the separation of pd-OWF from EV-OWPuzz.
We define the oracle in Section~\ref{sec:owf-oracle} and prove the two sides of the separation separately. Section~\ref{sec:existence_of_EV-OWPuzz_1} establishes the security of the EV-OWPuzz. Here, an additional challenge is that security must hold against non-uniform adversaries: their classical advice may depend on the oracle, and since non-uniform algorithms are not countable, a bound for each fixed algorithm does not immediately suffice. We address this using an argument inspired by~\cite{FK-Set-size-estimation} that controls all polynomial-length advice strings at each security parameter. Section~\ref{sec:owf-nonexistence} then rules out pd-OWFs using the simulation and concentration arguments outlined above. A subtlety is that a candidate generator can also query oracle components at smaller widths, where our concentration bounds are not sufficiently strong. Following the approach in~\cite{Kre21}, we separate the oracle into small and large components. The large components can be simulated, while the small components can be learned: their dimensions are only polynomial in the security parameter, allowing us to obtain classical descriptions through process tomography and direct queries to the verification oracle. Section~\ref{sec:L2-concentration} proves the two concentration arguments, with the phase-concentration and Efron--Stein ingredients developed in Appendices~\ref{sec:phase-concentration}~and~\ref{sec:Appendix-efron-stein}.

\subsection{QCCC Commitment from Good EV-OWPuzz: a Quantum NOVY Theorem}\label{sec:tech-overview-novy}
Unlike Theorem~\ref{thm:informal-1}, whose separation generally matches our heuristic, Theorem~\ref{thm:informal-2} is perhaps more surprising.
Indeed, in light of the result of~\cite{CountCrypt}, \cite{CGG-EV-OWPuzz} explicitly asks whether QCCC commitments can be constructed even from (inefficiently-verifiable) OWPuzz.
To understand why, we first examine how additional properties of an EV-OWPuzz can make such a construction \emph{possible} in Theorem~\ref{thm:informal-NOVY-security}.

\paragraph{A construction of QCCC commitment when puzzles are uniform.}
A natural starting point is to search for a black-box construction by revisiting the classical routes from OWFs to commitments.

The most familiar route first constructs a PRG from an OWF~\cite{HILL} and then applies Naor's commitment scheme~\cite{Naor-commitment}. However, it is not clear how to carry this route over to a quantum sampler: we need the explicit random seed of the PRG in Naor's commitment, whereas, as discussed above, there is no general efficient way to pull out the intrinsic randomness of a quantum algorithm.

We therefore turn our attention to the statistically hiding commitment constructions in~\cite{NOVY,SHcommit}. Suppose that the puzzle distribution produced by $\mathsf{Samp}$ is uniform over all possible puzzles. With this additional property, we can follow the approach in~\cite{NOVY}. The committer first samples a puzzle--key pair $(s,k)$ and engages in the following interactive hashing protocol with the receiver: each round the receiver challenges with some $h_i$, and the committer has to answer with $h_i\cdot s$. In the end, the interactive hashing scheme produces two puzzles $s_0,s_1$, one of which is the original puzzle $s$; the committer knows the index $d$ such that $s=s_d$, while the receiver does not. The committer commits to $b$ by sending $b\oplus d$, and opens by revealing $d$ and the key $k$.

Intuitively speaking, why is this protocol hiding and binding? Since $s$ is uniform, either $s_0$ or $s_1$ could be the sampled puzzle, and the receiver cannot tell which one is sampled given only the interactive hashing transcript. Thus, the index $d$ hides the committed bit. On the other hand, \emph{a priori}, opening the same commitment to both $b=0$ and $b=1$ requires valid keys for both $s_0$ and $s_1$. Since interactive hashing prevents the committer from freely choosing both puzzles, finding keys for both would contradict the one-wayness of the EV-OWPuzz.

However, it is not an easy task to prove the protocol above is binding against quantum adversaries. To understand why the quantum binding is nontrivial, we need to dive into the proof strategy of~\cite{NOVY}.

\paragraph{The classical security of NOVY.} Now we sketch the proof of classical binding of~\cite{NOVY}.

Suppose for contraposition that the previously described NOVY-style commitment is not binding, i.e.\ there is a malicious committer that can open both $b=0$ and $b=1$. Let us pretend that the EV-OWPuzz is classically generated, and to get a contradiction, we build an adversary as follows. To invert a puzzle $s$, the adversary simulates the interaction between the malicious committer and the honest receiver in the NOVY commitment. The central claim is as follows~\cite[Claim~1]{NOVY}: at each round of the interactive hashing, the committer's answer agrees with a selected typical puzzle with roughly half probability.
Thus, for a typical $s$, at each round, the adversary repeatedly samples the randomness, until the committer's output agrees with the value $s$ we want to invert. With high probability, the adversary gets a compatible run of the simulated committer with a polynomial overhead. Given the transcript in the interactive hashing step, the committer can open both $b=0$ and $b=1$ with non-negligible probability, so it must invert $s$ with non-negligible probability.

\paragraph{How to migrate to the quantum setting.}
However, the security proof in~\cite{NOVY} no longer applies in the quantum setting. The binding proof of NOVY heavily relies on rewinding the transcript to force the malicious commitment to be consistent with the desired puzzle, but we do not know how to generalize this to the quantum case.
This was observed in~\cite{DMS00,Yan22}, both of which bypass this issue by diverting to a commitment scheme with quantum communication. Therefore, we dedicate the rest of this section to describing an efficient way of ``quantum rewinding''. Crucially, this quantum rewinding amounts to asking the committer to output a selected outcome, which is essentially postselecting that outcome. We will argue that this postselection can be performed efficiently.

First, we prove a quantum version of~\cite[Claim~1]{NOVY}, which essentially states: for a typical puzzle $s$, there is about half chance that $s$ will be preserved in a single round. To better understand the meaning, we interpret this in two extremes.

In one extreme, under each challenge $h$, there is approximately half probability of the malicious committer outputting each answer. Recall that, quantumly, the committer relies on a projective measurement to decide its classical output. In this case, this measurement projector on the committer's private state shrinks the private state to a state with norm roughly $1/\sqrt{2}$, so postselection (renormalizing) requires a multiplicative factor $\sqrt{2}$.

In another extreme, the committer outputs deterministically for each challenge, and there are around half of the challenges $r$ preserving the puzzle $s$. In this case, the private state does not need renormalization. However, the transcript (i.e.\ the challenges) requires postselection: regarding $h$ as randomly selecting from a uniform superposition of all challenges, this postselection proposes a multiplicative factor of $\sqrt{2}$ on the \emph{challenge} register, after such purification by the transcript (challenges).

Thus, we regard the challenges as public coins that ``internally" sampled in the simulation, and purify such public coins. What is needed to simulate the postselection is the following: after such purification, to postselect $s$ to be preserved in each round is to multiply by roughly $\sqrt{2}$ after the coherent action of the committer's measurement.

However, multiplication by $\sqrt{2}$ is not a physical action. But what it requires is even simpler: it suffices to \emph{amplify} the amplitude of the ``consistent" portion, which is of norm roughly $1/\sqrt{2}$.
To be more precise, we employ the quantum singular value transformation (QSVT) with a polynomial that sends $1/\sqrt{2}$ to $1$, which has the effect of transmitting the singular values of the committer's (coherent) measurement operator from $1/\sqrt{2}$ to $1$. Therefore, application of QSVT amplifies the ``consistent'' amplitudes to norm roughly $1$, which is enough for our postselection.

\paragraph{Outline of the proof.} Section~\ref{sec:appendix-uniform} is devoted to the proof of the quantum NOVY theorem, where Theorem~\ref{thm:uniform-puzzle-to-commitment} is the main theorem, and Lemma~\ref{lem:resampler-variance} is the quantum analog of~\cite[Claim~1]{NOVY}. The only remaining subtlety is that the bound in Lemma~\ref{lem:resampler-variance} may be too small in the last few rounds. The classical result~\cite{NOVY} resolves it by simply random guessing (instead of repeatedly resampling) at the final few rounds.\footnote{Here we point out that the original reduction in~\cite{NOVY} is actually not fully black-box: it relies on knowing the security of the one-way permutation.} We found out that a na\"ive random guessing is not sufficient in our proof of quantum security. Instead, we do a ``mini" renormalization by simulating the multiplication of a number less than $\sqrt{2}$ but greater than $1$. And finally, we can guess the cutoff point from the regular renormalization by $\sqrt{2}$ to the ``mini" renormalization, which finishes the proof.

\subsection{The separation of QCCC commitment from EV-OWPuzz}\label{sec:tech-overview-commitment}
\paragraph{Why the construction fails for a general EV-OWPuzz.}
A natural next step, following~\cite{SHcommit}, is to generalize to a ``regular'' version of EV-OWPuzz. As we will see below, difficulty arises in this step, which will later motivate our oracle separation.

Inspired by the definition of a regular OWF, we call an EV-OWPuzz regular if there is a subset $S$ of all puzzles such that the puzzle distribution produced by $\mathsf{Samp}$ is supported on and uniform over $S$. If the size of $S$ is known, the standard approach in~\cite{SHcommit} is to apply a suitable pairwise-independent hash function to the puzzle. The leftover hash lemma guarantees that the resulting distribution is close to uniform. The question then becomes: is the one-wayness of an EV-OWPuzz preserved after such hashing?

Unlike in the OWF case, the answer turns out to be negative, with the issue lying in the verifier of the EV-OWPuzz. The security of an EV-OWPuzz only requires that finding a valid key for $s$ be hard \emph{when $s$ is honestly sampled}. However, the verifier may deliberately accept easily found keys for puzzles that can never be generated by the sampler, leaving room for an adversary to produce valid openings after hashing.

Consider the following toy example. Let $(\Samp,\Ver)$ be an EV-OWPuzz. Define $(\Samp',\Ver')$ as follows:
\begin{itemize}
    \item The sampler $\Samp'(1^\lambda)$ runs $(s,k)\gets\Samp(1^\lambda)$ and outputs $(s'=0^\lambda\|s,k'=k)$.
    \item The verifier $\Ver'(1^\lambda,s',k')$ checks whether the first $\lambda$ bits of $s'$ are all zero. If so, it writes $s'=0^\lambda\|s$ and outputs the result of $\Ver(1^\lambda,s,k')$. Otherwise, it simply outputs $1$.
\end{itemize}
It is easy to verify that $(\Samp',\Ver')$ remains a regular EV-OWPuzz if $(\mathsf{Samp},\mathsf{Ver})$ is regular. However, with the usual pairwise-independent hash family, an adversary can efficiently find a puzzle $s'$ whose first $\lambda$ bits are not all zero and whose hash equals a prescribed value.\footnote{Here, the ``usual" pairwise-independent hash function means the affine hashing function, which does not require any hardness assumption. Pairwise-independence alone does not guarantee collision-resistance, but they can both hold under additional assumptions~\cite{GGH96}, stronger than the OWF assumption.} Since $\Ver'$ accepts any key for such a puzzle, the adversary can solve the hashed puzzle without solving the original, honestly sampled one.

Hence, at a high level, the difficulty in constructing QCCC commitments in a black-box manner lies in preventing a malicious committer from exploiting ``\emph{easy}'' puzzles that are easy to solve but are never honestly sampled.
This suggests a starting point for our oracle construction: the sampler should produce ``hard" puzzles from a small region, while the verifier needs to accept keys for ``easy" puzzles outside that region. The challenge is to make these additional puzzles useful to a cheating committer without allowing the receiver to simply recognize and reject them. We first describe the basic oracle and explain how it supports an attack in a non-standard, asymmetric model, where adversaries against both primitives receive additional oracle access unavailable to the honest algorithms. We prove that this access model enables attacks on QCCC commitments while preserving the security of the EV-OWPuzz. We then use diagonalization to realize this asymmetry within the standard oracle model and combine the attacks into a single oracle separation.

\paragraph{The base oracle construction.}
To prevent the honest protocol from distinguishing whether a puzzle is ``easy'' or ``hard'', we begin by choosing a small random subset $S$ as the region of honestly sampled puzzles. Each puzzle in $S$ will have a unique valid key, as is the case in the construction in Section~\ref{sec:tech-overview-owf}.

However, we also need to hide the valid keys for puzzles outside $S$. Otherwise, the receiver could recognize these ``easy" puzzles simply by testing arbitrary keys. For each puzzle $s$, we therefore choose a random injective map $J(s,\cdot)$ into a much larger key space. Its image will be the set of valid keys when $s$ is ``easy''.

The base for our oracle separation is as follows. Choose a random function $g$. The sampling part of the oracle provides access to the state
\[
    \ket{\phi}
    =\frac{1}{\sqrt{|S|}}
      \sum_{s\in S}
      \alpha(s)\ket{s}\ket{J(s,g(s))},
\]
where the $\alpha(s)$ are independent random signs, as before. Measuring this state produces a valid puzzle--key pair, with the puzzle uniformly distributed over $S$. The oracle also provides a verifier $f$ defined as follows: for $s\in S$, it accepts only the key $J(s,g(s))$; for $s\notin S$, it accepts every key in the image of $J(s,\cdot)$.

We now briefly describe the main idea behind the construction. Informally, we first work in a non-standard oracle model: we supplement the oracle with a helper $\mathsf{Help}$ which depends on $J$ but no other hidden data.
Crucially, $\mathsf{Help}$ is available to adversaries (against both primitives), while remaining unavailable to the primitives themselves. 
The resulting EV-OWPuzz remains secure even against adversaries with access to $J$: finding $g(s)$ is still hard given an honestly sampled puzzle $s$. For commitments, however, access to $J$ allows an adversary to find valid keys for the ``easy'' puzzles, providing the starting point for our attack.

\paragraph{How to break QCCC commitments.}
We now explain how to construct such an auxiliary oracle $\mathsf{Help}$ and use it to break a fixed QCCC commitment scheme with access to $\ket{\phi}$ and $f$.

Write $\Theta=(g,\alpha,S)$ for all the hidden information other than $J$. Since $\mathsf{Help}$ must be independent of $\Theta$, we use an \emph{averaging} argument: find an adversary that does not know the actual choice of $\Theta$, but still breaks the commitment on average over $\Theta$, with $J$ held fixed.

As a motivating example, consider the following candidate protocol built from the EV-OWPuzz constructed above. To commit to $b$, the committer samples a puzzle--key pair $(s,k)$ produces a hardcore bit $k\cdot r$ where $r$ is uniformly randomly sampled, and sends $(s,r,k\cdot r\oplus b)$. To open, the committer sends $k$, and the receiver checks that $(s,k)$ is valid.

This candidate is semi-honest binding: if the committer acts honestly during the commit stage, it cannot change the committed bit during the opening stage, since an honestly sampled puzzle has only one valid key. Thus, an attack cannot simply leave the commit stage honest and modify the opening strategy. This suggests changing how the committer chooses its puzzle: it should instead use an ``easy'' puzzle, for which the adversary's access to $J$ provides many valid keys and allows it to find keys supporting either opening.

We implement this idea by giving the committer an independently sampled oracle. Specifically, sample $\widehat{\Theta}=(\widehat g,\widehat\alpha,\widehat S)$ independently of $\Theta$, while keeping the same $J$. A puzzle sampled from $\widehat S$ typically lies outside the sparse true support $S$. Its key, nevertheless, belongs to $\operatorname{Im}J(s,\cdot)$ and is therefore accepted by the true verifier. In other words, the committer can follow its original sampling procedure in the fake oracle world induced by $\widehat \Theta$ and $J$, while producing puzzles that are easy in the receiver's true oracle world.

This gives a systematic way to modify the committer's behavior (in the commit stage). What remains is that the receiver cannot detect the replacement through the interaction. The key observation is that, after averaging, replacing the committer's oracles in this way has only a negligible effect on the transcript distribution. More precisely, we prove the following.

\begin{theorem}[Oracle decoupling; informal]\label{thm:informal-decoupling}
Fix an efficient QCCC commitment protocol $\Pi$ with both parties accessing $\ket{\phi}$ and $f$. For a typical choice of $J$ and for either committed bit $b\in \{0,1\}$, the following two transcript distributions are negligibly close:
\begin{itemize}
    \item \emph{The original experiment.} Sample $\Theta$ and run $\Pi$, with both parties accessing $\ket{\phi}$ and $f$, constructed from $\Theta$ and $J$.
    \item \emph{The decoupled experiment.} Run the following decoupled protocol $\widetilde{\Pi}$: sample $\Theta$ and $\widehat{\Theta}$ independently, and then
    \begin{itemize}
        \item The receiver behaves exactly as in $\Pi$, accessing the original state $\ket{\phi}$ and verifier $f$.
        \item The committer behaves as in $\Pi$, except that it accesses $\ket{\widehat{\phi}}$ and $\widehat f$, constructed from $\widehat{\Theta}$ and the same $J$.
    \end{itemize}
\end{itemize}
\end{theorem}

The main ingredient of Theorem~\ref{thm:informal-decoupling} is the decoupling of phases (Theorem~\ref{thm:LOCC-decoupling}). Informally, we prove that the LOCC indistinguishability, defined in~\cite{AGL24}, is also satisfied to the graph state with random phases, as reflected in the following lemma.

\begin{lemma}[LOCC phase decoupling; informal]\label{lem:locc-decoupling-informal}
    Define $$\ket{\psi^\alpha}=\frac{1}{2^{n/2}}\sum_{x\in \{0,1\}^n}\alpha(x)\ket{x}$$ where each $\alpha(x)$ is a random phase of magnitude $1$. Then for any $t$ that is a polynomial in $n$,
    $$\EE_\alpha \left[\ketbra{\psi^\alpha}{\psi^\alpha}^{\otimes 2t}_{\mathsf{AB}}\right]\approx_{\mathrm{LOCC}} \EE_{\alpha,\beta}\left[\ketbra{\psi^\alpha}{\psi^\alpha}^{\otimes t}_{\mathsf{A}}\otimes \ketbra{\psi^\beta}{\psi^\beta}^{\otimes t}_{\mathsf{B}}\right]\,.$$
    Here $\approx_{\mathrm{LOCC}}$ means that the states are indistinguishable against any two-party LOCC protocols between $\mathsf{A}$ and $\mathsf{B}$.
\end{lemma}
With Lemma~\ref{lem:locc-decoupling-informal}, we can replace the state oracle with the corresponding sampling oracle: only if the phases of the states received by the two parties are decoupled, the sample-based simulator (Lemma~\ref{lem:informal-averaging}) can be applied locally on each side, keeping the simulated protocol LOCC.

But why does Theorem~\ref{thm:informal-decoupling} help? We follow the following method, whose idea originates from~\cite{AGL24,AGL25}. Once $J$ is fixed, we can view the receiver and the committer as internally sampling $\Theta$ and $\widehat{\Theta}$, respectively, and simulating their corresponding oracles. This gives a plain-model QCCC protocol without oracle access, though possibly an inefficient one. But a QCCC commitment protocol cannot be both statistically hiding and statistically binding, regardless of its efficiency. Thus, the decoupled protocol admits an information-theoretic attack against either hiding or binding. Theorem~\ref{thm:informal-decoupling} guarantees that this attack also breaks the original protocol on average over $\Theta$.

We have constructed a possibly inefficient informational-theoretic attack with only access to $J$, now the only thing to do is to construct the auxiliary oracle $\mathsf{Help}$ for carrying it out efficiently. If the decoupled protocol is not statistically hiding, $\mathsf{Help}$ identifies which committed bit is more likely given the transcript; if the decoupled protocol is not statistically binding, $\mathsf{Help}$ teaches the malicious committer what to send to open the other bit. Since these instructions are determined by the averaged decoupled protocol $\widetilde \Pi$, $\mathsf{Help}$ remains independent of the actual choice of $\Theta$.
Thus, for each fixed commitment scheme, we obtain a helper oracle that breaks the scheme while preserving the security of the EV-OWPuzz. It remains to combine the attacks into a single oracle and then lift to the standard oracle model. We accomplish this through diagonalization, completing the proof of Theorem~\ref{thm:informal-2}.

\paragraph{Outline of the proof.} Section~\ref{sec:commit-separation} is dedicated to the proof of the separation of QCCC commitment schemes from EV-OWPuzz. Section~\ref{sec:commitment-oracle} describes the base oracle, as presented informally in this technical overview; we also lift the access model to a full unitary with controlled access and access to its inverse and conjugate, similar to the previous informal argument in Section~\ref{sec:tech-overview-owf}. Section~\ref{sec:commitment-ev-owpuzz} briefly discusses the reason why EV-OWPuzz still exists relative to the oracle. 
Section~\ref{sec:LOCC-decoupling} formalizes and proves Theorem~\ref{thm:informal-decoupling}. Its central technical ingredient is the LOCC decoupling of phases, proved in Appendix~\ref{sec:phase-concentration}. By exploiting the restriction to classical communication, this lemma allows us to decouple the phases seen by the two parties, providing the crucial step that enables our attack on QCCC commitments. Establishing this lemma is one of the main technical challenges of the proof, and the lemma may also be of independent interest.
Section~\ref{sec:diagonalization-oracle} introduces the diagonalization technique and Section~\ref{sec:commitment-nonexistence} formally proves the non-existence of QCCC commitment schemes relative to the oracle.

\ifnum\llncs=1
\bibliographystyle{splncs04}
\else
\bibliographystyle{alpha}
\fi

\renewcommand*{\theHsection}{appendix.\Alph{section}}
\section{Preliminaries}
\subsection{Probability theory}

\begin{lemma}[Borel--Cantelli] \label{Borel-Cantelli}
    Suppose $\{E_n\}_{n\in \NN}$ is a series of events in a probability space $\Omega$.
    If $$\sum_{n\in \NN}\Pr(E_n)<\infty $$ then the probability that infinitely many of the events occur is $0$.
\end{lemma}

\begin{lemma}[Efron--Stein inequality]\label{lem:Efron-stein-original}
    Let $X_1,X_2,\cdots,X_n$ and $X_1',X_2',\cdots,X_n'$ be independent random variables taking values on $\mathcal X$, with $X_i$ and $X_i'$ having the same distribution for all $i$.
    Let $$X=(X_1,X_2,\cdots,X_n)\,,$$
    $$X^{(i)}=(X_1,\cdots,X_{i-1},X_i',X_{i+1},\cdots,X_n)\,.$$
    Then, for any measurable function $f:\mathcal X^n\to \mathbb R$,
    $$\Var(f(X))\leq \frac{1}{2}\sum_{i=1}^n \EE[(f(X)-f(X^{(i)}))^2]\,.$$
\end{lemma}

\begin{proof}
    This is first proven in~\cite{Steele-Efron-Stein}.
    A modern proof can be found in~\cite[Section 3.1]{BLM-concentration-inequalities}.
\end{proof}

\subsection{QCCC primitives}

\begin{definition}[pd-OWF, adapted from \cite{BBO+bot-OWF}; see also~\cite{BM24}]
    Let $n=n(\lambda)$ and $m=m(\lambda)$ be polynomials in the security parameter $\lambda$.
    A QPT algorithm $\Gen$ which on input $x\in \{0,1\}^n$ output a string in $\{0,1\}^m$ is a pseudodeterministic one-way function, if for sufficiently large $\lambda$, there exist a function $H_\lambda:\{0,1\}^n\to \{0,1\}^m$ and a set $\mathcal G\subset \{0,1\}^n$ such that the following conditions hold.
    \begin{itemize}
        \item (Pseudodeterminism)
        \begin{itemize}
            \item For every $x\in \mathcal G$,\footnote{The original definition in~\cite{BBO+bot-OWF} of pseudodeterminism only requires $\Gen$ to output the deterministic function value with probability $1-1/\poly(\lambda)$. With a repetition and majority voting of the function value, we can boost the probability to $1-\negl(\lambda)$.}
            $$\Pr\left[\Gen(1^\lambda,x)=H(x)\right]\geq 1-\negl(\lambda)\,.$$
            \item The size of $\mathcal G$ is at least $(1-\mu(\lambda))2^n$.
            In other words, $$\Pr_{x\gets \{0,1\}^n}[x\in \mathcal G]\geq 1-\mu(\lambda)\,.$$
        \end{itemize}
        \item (Security) For any QPT algorithm $\mathcal A$ with advice $c$,
        \[
            \Pr_{\substack{x \gets \{0,1\}^n \\ y := H(x)}}
            \Bigl[H\bigl(\mathcal{A}(y, c)\bigr) = y\Bigr]
            \leq \negl(\lambda).
        \]
    \end{itemize}

    We say $\Gen$ is a pseudodeterministic one-way function (with inverse-polynomial error), if $\mu(\lambda)\leq 1/3$ for sufficiently large $\lambda$.

    Without loss of generality, we may assume $n(\lambda) = \lambda$ later on. Our results also apply to general parameter regimes with a standard padding argument.
\end{definition}

\begin{definition}[EV-OWPuzz, adapted from \cite{CGG-EV-OWPuzz}]
    An efficiently verifiable one-way puzzle is a pair of a sampling and verification algorithms $(\mathsf{Samp},\mathsf{Ver})$ with the following syntax.
    \begin{itemize}
        \item $\mathsf{Samp}(1^\lambda)\to (s,k)$ is a QPT algorithm that outputs a pair of classical strings $(s,k)$.
        We refer to $s$ as the puzzle and $k$ as the key.
        Without loss of generality we may assume that $k\in \{0,1\}^\lambda$.
        \item $\mathsf{Ver}(s,k)\to b$ is a QPT algorithm that on input any pair of classical strings $(s,k)$, outputs a bit $b$.
    \end{itemize}
    And these algorithms satisfy the following properties:
    \begin{itemize}
        \item (Correctness) Outputs of the sampler pass verification with overwhelming probability.
        In other words,
        $$\Pr_{(s,k)\gets \mathsf{Samp}(1^\lambda)}\bigl[\mathsf{Ver}(s,k)=1\bigr]\geq 1-\negl(\lambda)\,.$$
        \item (Security) Given a puzzle $s$ sampled by $\mathsf{Samp}$, it is hard to find a key $k$ that can pass verification with non-negligible probability. In other words, for all QPT algorithms $\mathcal A$ with advice $c$,
        $$\Pr_{(s,k)\gets \mathsf{Samp}(1^\lambda)}\Bigl[\mathsf{Ver}\bigl(s,\mathcal A(s,c)\bigr)=1\Bigr]\leq \negl(\lambda)\,.$$
    \end{itemize}
\end{definition}

\begin{definition}[QCCC one-time signature, adapted from \cite{CGG-EV-OWPuzz}]
    A QCCC one-time signature scheme is a set of QPT algorithms $(\mathsf{KeyGen},S,V)$ with the following syntax.
    \begin{itemize}
        \item $\mathsf{KeyGen}(1^\lambda)\to (vk,sk)$ takes in the security parameter as input and outputs a signing key $sk$ and a verification key $vk$;
        \item $S(sk,m)\to \sigma$ takes in the signing key and a message as input, and outputs a classical signature $\sigma$;
        \item $V(vk,m,\sigma)\to b$ takes in a verification key $vk$, a message $m$, and a signature $\sigma$ as input, and outputs a single bit $b$.
    \end{itemize}
    And these algorithms satisfy the following properties:
    \begin{itemize}
        \item (Correctness) For all $m$ in the message space,
        \begin{equation*}
          \Pr_{\substack{
            (vk,sk) \gets \mathsf{KeyGen}(1^\lambda) \\
            \sigma \gets S(sk,m)
          }}
          \bigl[V(vk,m,\sigma)=1\bigr]
          \geq 1-\negl(\lambda).
        \end{equation*}
        \item (One-time Security) For all $m_0\neq m_1$ in the message space and for all QPT algorithms $\mathcal A$ with advice $c$,
        \begin{equation*}
          \Pr_{\substack{
            (vk,sk) \gets \mathsf{KeyGen}(1^\lambda) \\
            \sigma_0 \gets S(sk,m_0)
          }}
          \Bigl[V\bigl(vk,m_1, \mathcal{A}(vk,\sigma_0,c)\bigr)=1\Bigr]
          \leq \negl(\lambda).
        \end{equation*}
    \end{itemize}
\end{definition}

\begin{theorem}[\cite{CGG-EV-OWPuzz}]\label{thm:EV-OWPuzz-equal-QCCC-OTS}
    There exists a fully black-box construction of QCCC one-time signature scheme from EV-OWPuzz.
    Also, there exists a fully black-box construction of EV-OWPuzz from QCCC one-time signature scheme.
\end{theorem}

\begin{definition}[QCCC commitment]\label{def:QCCC-commitment}
    We say a two-phase protocol with two parties $C$ and $R$ is a $\eta_b$-correct, $\epsilon$-computationally (resp.\ statistically) hiding, and $\delta$-computationally (resp.\ statistically) (sum-)binding, if it satisfies the following syntax:
    \begin{itemize}
        \item (Commit stage) Both parties receive the unary security parameter $1^\lambda$.
        The committer gets a bit $b\in \{0,1\}$.
        The committer interacts with the receiver using only classical messages, and together they produce a transcript $z$.
        At the end of the stage, both parties hold a private quantum state $\rho_C$ and $\rho_R$ respectively.
        \item (Opening stage) Both parties receive the transcript $z$ as well as their private quantum state $\rho_C$ and $\rho_R$ respectively.
        The committer interacts with the receiver using only classical messages, and together they produce a transcript $d$.
        At the end of the stage, the receiver either outputs a message or rejects.
    \end{itemize}
    And the protocol satisfies the following properties:
    \begin{itemize}
        \item ($\eta_b$-Correctness) For all messages $b$, when $C$ and $R$ interact honestly, the probability that $R$ outputs $b$ at the end of the opening stage is at least $1-\eta_b$.
        \item ($\epsilon$-computational (resp.\ statistical) hiding) At the commit stage, for any QPT (resp.\ unbounded) adversarial receiver $R'$, the transcript of the interaction between $R'$ and $C$ are computationally $\epsilon$-indistinguishable between when the committer gets $0$ and when the committer gets $1$.
        \item ($\delta$-computational (resp.\ statistical) sum-binding) For all QPT (resp.\ unbounded) adversarial committer $C'$, the probability that $C'$ wins the following security game is at most $(1+\delta)/2$.
        \begin{enumerate}
            \item In the commit stage, the adversarial committer does not receive the message to commit $b$.
            The honest receiver $R$ interacts with the adversarial committer $C'$ to produce a transcript $z$ and the private quantum states $\rho_C$, $\rho_R$.
            \item After the commit stage and before the opening stage, $C'$ samples a uniformly random bit $b$.
            \item In the opening stage, the honest receiver $R$ is given $\rho_R$ and $z$, while $C'$ is given $z$, $\rho_C$, and $b$.
            The adversarial committer $C'$ interacts with the honest receiver $R$, and they produce a transcript $d$.
            At the end of the stage, the receiver outputs a message $b'$ or rejects.
            $C'$ wins if $b'=b$.
        \end{enumerate}
    \end{itemize}
    If $\eta_b,\epsilon,\delta$ are all negligible in $\lambda$ and the protocol is efficient, then we will drop the parameters and call it a QCCC commitment scheme.
\end{definition}

Informally, the following theorem states that there does not exist a QCCC bit commitment scheme that is secure to any unbounded adversaries.

\begin{lemma}\label{lem:dichotomy-hiding-binding}
For any $\eta_b$-correct, $\epsilon$-statistically hiding, and $\delta$-statistically binding QCCC (possibly inefficient) commitment scheme $\Pi$,
$$
\eta_0+\eta_1+\epsilon+\delta \geq 1.
$$
\end{lemma}

\begin{proof}
For $b\in\{0,1\}$, let $P_b$ denote the distribution of the commit stage transcript $z$ in an honest execution committing to $b$.
We have
\begin{equation*}
  \frac12\sum_z |P_0(z)-P_1(z)|\leq\epsilon\,.
\end{equation*}

We first recall a useful property of QCCC protocols.
Since all communication is classical and the parties start with a product state, conditioned on any fixed transcript $z$, the joint state of the committer and receiver at the end of the commit stage is a product state.
Thus, in an honest execution committing to $b$, conditioned on transcript $z$, the joint state can be written as
\[
\rho^{b,z}_{C}\otimes \rho^z_R.
\]
Importantly, the receiver's conditional state $\rho^z_R$ does not depend on $b$: during the commit stage the receiver receives no input $b$, and its local evolution is completely determined by its own operations and the classical transcript $z$.

In the honest execution of the protocol of committing and opening some $b\in \{0,1\}$, define
$$q_b(z)=\begin{cases}\Pr\big[\text{Receiver outputs }b\mid z\big] & P_b(z)>0\\
0 & P_b(z)=0\end{cases}\,.$$
By correctness,
\begin{equation*}
  \EE_{z\leftarrow P_b}[q_b(z)]\geq 1-\eta_b\,.
\end{equation*}

We now construct an unbounded cheating committer $C'$.
During the commit stage, $C'$ honestly executes the commit stage for $b=0$, producing the resulting transcript distribution $P_0$.

Suppose $C'$ is asked to open $b=0$.
The probability that the receiver output $b=0$ is, by definition, $\EE_{z\leftarrow P_0}[q_0(z)]\geq 1-\eta_0$.

Suppose $C'$ is asked to open $b=1$.
Given the transcript $z$, it discards its current private state and, whenever $P_1(z)>0$, prepares the conditional honest-committer state $\rho_C^{1,z}$, and then executes the honest opening procedure for $b=1$.
This is possible for an unbounded committer: it may simulate the honest execution conditioned on $z$ to arbitrary precision.
Since conditioned on $z$ the receiver's state is exactly $\rho_R^z$, the joint state from this point onward is precisely
\[
\rho_C^{1,z}\otimes \rho_R^z,
\]
and therefore its conditional probability of successfully opening $b=1$ is $q_1(z)$.
Thus, the probability of the receiver outputting $b=1$ is at least $\EE_{z\leftarrow P_0}[q_1(z)]$.

Since $q_b(z)\in[0,1]$, the variational characterization of the statistical distance gives
\begin{equation*}
  \mathbb{E}_{z\leftarrow P_0}[q_1(z)]
  \geq \mathbb{E}_{z\leftarrow P_1}[q_1(z)]-\epsilon\,.
\end{equation*}

Combining with the correctness, we obtain that the unbounded malicious committer wins the sum-binding game with probability at least $1-(\eta_0+\eta_1+\epsilon)/2$.
\begin{equation*}
  1-\frac{\eta_0+\eta_1+\epsilon}{2}
  \leq \Pr[C'\text{ wins}]
  \leq \frac{1+\delta}{2}\,.
\end{equation*}
\end{proof}

\subsection{Quantum computing}

The BBBV theorem we will use in this paper is as follows. For completeness, we include the proof here.

\begin{lemma}[BBBV]\label{BBBV}
    Set $f:\{0,1\}^n\to \{0,1\}$ to be sampled according to some distribution $\mathcal D$.
    Then any oracle circuit $C$ that queries $q$ times to a boolean function oracle must satisfy
    $$ \EE_{f\sim \mathcal D}\left[\| C^{f}\ket{0}- C^{0}\ket0\|^2\right]\leq 4q^2p\,,$$
    where
    $$p=\max_{z} \Pr_{f\sim \mathcal D}[f(z)=1]\,.$$
    
    As a corollary,
    $$\EE_{f\sim \mathcal D}\left[\| C^{f}\ket{0}- C^{0}\ket0\|\right]\leq 2q\sqrt{p}\,.$$
\end{lemma}

\begin{proof}
Let \(w_{t,z}\) denote the query magnitude of \(z\) immediately before the \(t\)-th query in the computation \(C^0\ket{0}\).
Thus, for every $t\in [q]$, $\sum_z w_{t,z}=1$.
By the BBBV hybrid argument \cite[Theorem~3.3]{BBBV}, for every fixed \(f\),
\[
\left\|C^f\ket{0}-C^0\ket{0}\right\|
 \le
2\sum_{t=1}^{q}
\sqrt{\sum_{z:f(z)=1}w_{t,z}}\,.
\]
Therefore, by Cauchy--Schwarz,
\[
\left\|C^f\ket{0}-C^0\ket{0}\right\|^2
\le
4q\sum_{t=1}^{q}\sum_z
\mathbf 1_{f(z)=1}w_{t,z}\,.
\]
Taking expectation over \(f\leftarrow\mathcal D\) and using
\(\Pr_f[f(z)=1]\le p\), we obtain
\[
\EE_f\!\left[
 \left\|C^f\ket{0}-C^0\ket{0}\right\|^2
\right]
\le
4q\sum_{t,z}w_{t,z}\Pr_f[f(z)=1] 
\le 4q^2p\,.
\]
Notice that no independence among the values \(f(z)\) is needed.
Finally, Jensen's inequality gives
\[
\EE_f\!\left[
\left\|C^f\ket{0}-C^0\ket{0}\right\|
\right]
\le 2q\sqrt p\,.
\]
\end{proof}

We also use the following gentle measurement lemma.

\begin{lemma}\label{lem:gentle-measurement}
    Let $\rho$ be a quantum state, and let $0\preceq M\preceq I$ be a two-outcome measurement.
    Let $\epsilon=\Tr((1-M)\rho)<1$ be the rejecting probability of this measurement, and let
    $$\rho'=\frac{\sqrt{M}\rho\sqrt{M}}{\Tr(M\rho)}$$
    be the post-measurement state when the accept outcome is observed.
    Then
    $$\|\rho-\rho'\|_1\leq 2\sqrt{\epsilon}\,.$$
\end{lemma}

\begin{proof}
    This is straightforward from \cite[Corollary 3.15]{Watrous-book} and \cite[Theorem 3.33]{Watrous-book}.
\end{proof}

\begin{lemma}[Equation (3.306) of~\cite{Watrous-book}]\label{lem:diamond-to-output}
   Let $\Phi_1, \dots,\Phi_n$ and $\Psi_1, \dots, \Psi_n$ be channels, and $\Phi = \Phi_n \circ\dots\circ \Phi_1, \Psi = \Psi_n \circ \dots \circ \Psi_1$, then we have
   \begin{align*}
       \|\Phi(\rho) - \Psi(\rho)\|_1\leq \|\Phi-\Psi\|_{\diamond} \leq \sum_{i=1}^n \|\Phi_i - \Psi_i\|_{\diamond}\,,
   \end{align*}
   for any mixed state $\rho$.
\end{lemma}

\subsection{Process tomography for unitaries}
\begin{theorem}[\cite{HKOT-tomography}]\label{thm:unitary-tomography}
    There are quantum algorithms $\mathcal A$ and $\mathcal U$ that, given black-box access to a $d$-dimensional unitary $U$ and a parameter $\epsilon$, $\mathcal A$ queries the black-box unitary for $O(d^2/\epsilon)$ times, using $\poly(d,1/\epsilon)$ gates, and outputs a classical description $Z$ of a unitary, which can be compiled by the algorithm $\mathcal U$, such that
    $$\EE\|U-\mathcal U(Z)\|_\diamond^2\leq \epsilon^2\,.$$
    As a corollary,
    $$\Pr[\|U-\mathcal U(Z)\|_\diamond\geq \eta]\leq \frac{\epsilon^2}{\eta^2}\,.$$
\end{theorem}

\subsection{Reflection-to-state simulation}

Throughout the paper, we consistently use the following theorem from~\cite[Theorem 4]{JLS}.
This theorem allows us to simulate a reflection unitary with the state up to an arbitrary inverse-polynomial precision, which turns out to be helpful in our analysis.

\begin{theorem}[\cite{JLS}]\label{thm:JLS-simulation}
    Let $\ket{\psi}$ be a quantum state.
    Define the oracle $\mathcal O_{\ket{\psi}}:I-2\ketbra{\psi}{\psi}$ be the reflection about $\ket{\psi}$.
    Let $\ket{\xi}$ be a state not necessarily independent of $\ket{\psi}$.
    Let $\mathcal A^{\mathcal O_{\ket{\psi}}}$ be an oracle algorithm that makes $q$ queries to $\mathcal O_{\ket{\psi}}$.
    For any integer $l>0$, there is a quantum algorithm $\mathcal B$ that makes no queries to $\mathcal O_{\ket{\psi}}$ such that
    $$\left\Vert\mathcal A^{\mathcal O_{\ket{\psi}}}\ket{\xi}-\mathcal B(\ket{\psi}^{\otimes l}\otimes \ket{\xi})\right\Vert_1\leq O\left(\frac{q}{\sqrt{l}}\right)\,.$$
    Moreover, the running time of $\mathcal B$ is in polynomial in that of $\mathcal A$ and $l$.
\end{theorem}

\section{An Oracle Separation of EV-OWPuzz from Quantum One-way Functions}\label{sec:owf-separation}
In this section, we will prove the oracle separation between EV-OWPuzz and pseudodeterministic OWFs. More formally, we have the following theorem,

\begin{theorem}\label{thm:4-main}
    There exists a unitary oracle relative to which EV-OWPuzz exists and pd-OWF does not. Moreover, the separation holds even if we allow controlled access and access to the inverse, conjugate, and transpose.
\end{theorem}
The section is organized as follows. Section \ref{sec:owf-oracle} introduces the construction of the oracle. Section \ref{sec:existence_of_EV-OWPuzz_1} shows that a secure EV-OWPuzz exists relative to the oracle with a BBBV-type argument. Section \ref{sec:owf-nonexistence} shows that pd-OWFs do not exist relative to the oracle, which relies on reducing the unitary oracle to a sampling oracle (proven in Appendix \ref{sec:phase-concentration}) and the $\ell^2$-concentration of the sampling distribution (proven in Section \ref{sec:L2-concentration}).

\subsection{The oracles}\label{sec:owf-oracle}

For every $n\in\mathbb N_+$, independently sample uniformly random functions
\[
    g_n:\{0,1\}^n\longrightarrow \{0,1\}^{3n}
    \qquad\text{and}\qquad
    R_n:\{0,1\}^n\longrightarrow \{0,1\}.
\]
For simplicity, we abuse notations by denoting the sequence of them by not adding $n$ explicitly.
For example, we write $g=(g_n)_{n\in\mathbb N_+}$ and $R=(R_n)_{n\in\mathbb N_+}$.

For each $n$, define the random-phase graph state
\[
    \ket{\psi_{g_n,R_n}}
    :=
    2^{-n/2}
    \sum_{x\in\{0,1\}^n}
        (-1)^{R_n(x)}
        \ket{x}\ket{g_n(x)}.
\]
And let $U_{g_n,R_n}$ be a unitary that swaps $\ket{0^{4n+1}}$ and $\ket{1}\ket{\psi_{g_n,R_n}}$:
\[
    U_{g_n,R_n}
    :=
    I-
    \left(
        \ket{0^{4n+1}}-\ket{1}\ket{\psi_{g_n,R_n}}
    \right)
    \left(
        \bra{0^{4n+1}}-\bra{1}\bra{\psi_{g_n,R_n}}
    \right).
\]
We also define the membership function by
\[
    f_{g_n}(x,y):=\mathbf 1[y=g_n(x)].
\]
We give quantum query access to it as:
\[
    \ket{x}\ket{y}\ket{b}
    \longmapsto
    \ket{x}\ket{y}
    \ket{b\oplus f_{g_n}(x,y)}.
\]
Finally, fix a classical oracle $\pspace$ for a $\mathbf{PSPACE}$-complete language.
Our oracle is
\[
    \calO_{g,R}
    :=
    \bigl(U_{g,R},f_g,\pspace\bigr),
\]

Note that we allow controlled access. $U_{g, R}$ is self-inversed and self-conjugated, so the inverse, conjugate, and transpose queries of $U_{g,R}$ and controlled $U_{g,R}$ are also allowed.

\subsection{EV-OWPuzz survives}
\label{sec:existence_of_EV-OWPuzz_1}

For $S\subseteq\{0,1\}^n$, define
\begin{itemize}
    \item the state
    \[
    \ket{\widetilde\psi_{g_n,R_n,S}}
    :=
    2^{-n/2}
    \left(
        \sum_{x\in S}\ket{x}\ket{0^{3n}}
        +
        \sum_{x\notin S}
            (-1)^{R_n(x)}
            \ket{x}\ket{g_n(x)}
    \right),
\]
    \item the unitary $\widetilde U_{g_n,R_n,S}$ that swaps $\ket{0^{4n+1}}$ and $\ket{1}\ket{\widetilde\psi_{g_n,R_n,S}}$, and acts as the identity on other orthogonal states.
    \item the membership function 
    \[
        \widetilde f_{g_n,S}(x,y)
        :=
        \begin{cases}
            0, & x\in S,\\
            f_{g_n}(x,y), & x\notin S.
        \end{cases}
    \]
\end{itemize}
The difference between $\widetilde U_{g_n,R_n,S}$ and $U_{g_n,R_n}$ is unsurprisingly upper-bounded by the size of the set $S$, as the following lemma shows.
\begin{lemma}
\label{lem:psi-approximation}
For every $g_n,R_n$ and $S\subseteq\{0,1\}^n$,
\[
    \left\|
        \widetilde U_{g_n,R_n,S}
        -
        U_{g_n,R_n}
    \right\|_{\mathrm{op}}
    \leq
    2\sqrt{\frac{2|S|}{2^n}}.
\]
\end{lemma}

\begin{proof}
$\lVert\cdot\rVert$ denotes the Euclidean distance, and the graph states satisfy
\[
    \left\|
        \ket{\widetilde\psi_{g_n,R_n,S}}
        -
        \ket{\psi_{g_n,R_n}}
    \right\|
    \leq
    2\sqrt{\frac{|S|}{2^n}}.
\]
Define
\[
    \ket{\phi^-}
    :=
    \frac{
        \ket{0^{4n+1}}
        -
        \ket{1}\ket{\psi_{g_n,R_n}}
    }{\sqrt2}
\]
and $\ket{\widetilde\phi^-}$ correspondingly, then
\[
    \left\|
        \ket{\widetilde\phi^-}-\ket{\phi^-}
    \right\|
    \leq
    \sqrt{\frac{2|S|}{2^n}}.
\]
which yields
\[
    \left\|
        \ket{\widetilde\phi^-}\bra{\widetilde \phi^-}
        -
        \ket{\phi^-}\bra{\phi^-}
    \right\|_{\mathrm{op}}
    \leq
    \sqrt{\frac{2|S|}{2^n}}.
\]
Note that $\widetilde U_{g_n,R_n,S}$ and $U_{g_n,R_n}$ are reflections about $\ket{\widetilde\phi^-}$ and $\ket{\phi^-}$, respectively, and they satisfy
\[
    \widetilde U_{g_n,R_n,S}
    -
    U_{g_n,R_n}
    =
    -2\left(
        \ketbra{\widetilde\phi^-}{\widetilde\phi^-}
        -
        \ketbra{\phi^-}{\phi^-}
    \right).
\]
That completes the proof.
\end{proof}

\begin{theorem}[Existence of an EV-OWPuzz]
\label{thm:ev-owpuzz-existence}
With probability $1$ over the choice of $g$ and $R$, an EV-OWPuzz exists relative to $\calO_{g,R}$.
Moreover, the security is against nonuniform QPT adversaries with polynomial-length classical advice.
\end{theorem}

\begin{proof}
We define the EV-OWPuzz $(\mathsf{Samp},\mathsf{Ver})$ as follows.
\begin{itemize}
\item On security parameter $1^n$, $\mathsf{Samp}$ measures $U_{g_n,R_n}\ket{0^{4n+1}}$ in the computational basis. Write the outcome as $(1,x,g_n(x))$, where $x\in\{0,1\}^n$ and $g_n(x)\in\{0,1\}^{3n}$.
It then outputs $x$ as the puzzle and $g_n(x)$ as the key.
\item The verifier $\mathsf{Ver}(s,k)$ returns $f_{g_n}(s, k)$.
\end{itemize}

\emph{Correctness.}
By definition of the unitary,
\[
    U_{g_n,R_n}\ket{0^{4n+1}}
    =
    \ket1\ket{\psi_{g_n,R_n}}.
\]
The measurement produces $(1,x,g_n(x))$ for a uniformly random $x\in\{0,1\}^n$, and the resulting puzzle-key pair satisfies $f_{g_n}\big(x,g_n(x)\big)=1$.
Hence the EV-OWPuzz candidate has perfect correctness.

    \emph{Security.}
    It is relatively simple to prove the security against any uniform adversary. Given a sample from EV-OWPuzz $(x, g_n(x))$, we can replace the unitary $U_{g_n, R_n}$ with $U_{\tilde{g}_n,R_n}$ with a small loss in operator distance, where $\tilde{g}_n$ agrees with $g_n$ except that $\tilde{g}_n(x)=0$. Then $g_n(x)$ can only be accessed via querying $f_{g_n}$: given a random label $x$ and the verification oracle $f_{g_n}$, find the hidden index $g_n(x)$. A BBBV-type argument demonstrates that, on average over $R_n,g_n$, the probability of finding the hidden index is negligible. Combining the Borel--Cantelli lemma and the fact that the total number of all quantum algorithms is countable, with probability $1$ over all choices of $R,g$, any uniform quantum algorithm cannot break the security of the proposed EV-OWPuzz above.
    
    However, the usual security requirement, as described in this paper, requires the security against non-uniform adversaries that can accept a classical advice for each security parameter, resulting in more subtleties.
    In particular, the total number of quantum algorithms with advice is not countable, thus the usual application of Borel--Cantelli theorem, as in e.g.~\cite{CountCrypt}, is actually not correct in this scenario.
    
    The following proof is inspired by the technique in \cite{FK-Set-size-estimation} showing that the set-size estimation problem is not in $\mathbf{QCMA}$; see also the introduction in~\cite{BHNZ26-QMA-QCMA,BHV-QMA-QCMA}.

Throughout this proof, we fix any $R$ satisfying the syntax of the oracles. We assume $R$ is publicly accessible.
For simplicity, we assume that everyone has access to $\pspace$ in the proof, so we will omit it.
And in this proof, let
\[
    \epsilon_n:=2^{-n/50},
    \qquad
    L_n:=\left\lfloor2^{9n/10}\right\rfloor.
\]
For adversary $\mathcal A$, define 
\[
    Z_x^{g,R}(c)
    :=
    \Pr[
        \calA^{\calO_{g,R}}(x,c)=g_n(x)
    ],
    \qquad
    Z^{g,R}(c)
    :=
    \EE_{x\gets\{0,1\}^n}
    Z_x^{g,R}(c).
\]

For contradiction, we assume that there exists a non-uniform adversary that can invert EV-OWPuzz with probability at least $\epsilon_n$. More formally, there exists polynomial functions $q(\cdot)$ and $w(\cdot)$ be polynomial functions and $q(n)$-query algorithm $\mathcal A$  with advice of size at most $w(n)$ satisfying
\begin{align*}
    \max_{c\in \{0,1\}^{\leq w(n)}} Z^{g, R}(c) \geq \epsilon_n.
\end{align*}

For any $c \in \{0,1\}^{\leq w(n)}$, define
\[
    \mathsf{Good}_{g_n,c}
    :=
    \left\{
        x:
        Z_x^{g,R}(c)
        \geq
        \frac{\epsilon_n}{2}
    \right\}.
\]

Let $c^\star \coloneqq \arg\max_{c\in \{0,1\}^{\leq w(n)}}Z^{g,R}(c)$. Using standard averaging argument, we have $|\mathsf{Good}_{g_n,c^\star}|\geq2^{n-1}\epsilon_n > L_n$ for all sufficiently large $n$.

Now we define $S$ as any $L_n$-size subset of $\mathsf{Good}_{g_n,c^\star}$. We claim that replacing all the oracle queries of $\mathcal A$ to $U_{g, R}$ with $\widetilde U_{g, R, S}$ do not affect the algorithm much. In fact, by Lemma~\ref{lem:psi-approximation}, the corresponding change in the output probability is at most
\[
    O\left(
        q(n)\sqrt{\frac{L_n}{2^n}}
    \right)
    =
    O\left(q(n)2^{-n/20}\right).
\]
Thus, for every $x\in S\subseteq \mathsf{Good}_{g_n,c^\star}$, the computation with $\widetilde{U}_{g,R,S}$ still outputs $g_n(x)$ with probability at least $\epsilon_n/2 - O(q(n)2^{-n/20}) > \epsilon_n/3$ for $n$ large enough.

Next, we aim to replace all the oracle queries of $\mathcal A$ to $f_{g_n}$ with $\widetilde f_{g_n,S}$. This replacement is probably not indifferentiable for the adversary as the advice can record valid pairs of $g$ in $S$. So we need to distinguish the following two cases for the argument:
\begin{itemize}
    \item {\bf (Case 1)} The replacement of $f_{g_n}$ with $\widetilde f_{g_n,S}$ decreases the probability of $\calA$ outputting $g_n(x)$ by at least $\epsilon_n/6$.
    \item {\bf (Case 2)} $\calA^{\widetilde U_{g_n,R_n,S}, \widetilde f_{g_n,S}}$ still outputs $g_n(x)$ with probability at least $\epsilon_n/6$.
\end{itemize}

We now construct an extractor $\Decfull$ that, with abnormally high probability, $\Decfull$ recovers $g_n|_S$ from information independent of $g_n|_S$, leading to an information-theoretic contradiction.

We first construct a subroutine $\mathsf{Ext}$ to extract a valid pair in $g_n|_S$. $\mathsf{Ext}$ receives
$
    S,
    g_n|_{\{0,1\}^n\backslash S}, c^\star$, and $R_n$
as inputs, where $S$ is a subset of $\mathsf{Good}_{g_n,c^\star}$ whose size is less than $L_n$. It outputs a pair $(x,y)$ such that $x \in S$ and $g_n(x)=y$ with high probability. This subroutine is described as in Algorithm~\ref{alg:extractor}.

\begin{algorithm}[H]
    \caption{The extractor subroutine $\mathsf{Ext}$}
    \label{alg:extractor}
    \begin{algorithmic}
        \State {\bf Input: } $S, g_n|_{\{0,1\}^n\backslash S}, c^\star, R_n$
        \State Sample $x\gets S$
        \State Sample $\mathrm{coin}\gets \{0,1\}$
        \If{$\mathrm{coin}=0$}
            \State Sample $t\gets [q(n)]$
            \State Simulate $\mathcal A^{\left(
        \widetilde U_{g_n,R_n,S},
        \widetilde f_{g_n,S}
    \right)}(x,c^\star)$ to just before the $t$-th query to the oracle $\widetilde f_{g_n,S}$
            \State Measure the register that being used for the $t$-th query to $\widetilde f_{g_n,S}$ to get $(u,v)$
            \State Output the measurement outcome $(u,v)$
        \Else
            \State Run $\mathcal A^{\left(
        \widetilde U_{g_n,R_n,S},
        \widetilde f_{g_n,S}
    \right)}(x,c^\star)$ to get the output $y$
            \State Output $(x,y)$
        \EndIf
    \end{algorithmic}
\end{algorithm}

The subroutine $\mathsf{Ext}$ is computationally inefficient as it needs to take the truth table of $g_n|_{\{0,1\}^n \backslash S}$ and $R_n$ as input, and simulates the oracle $\widetilde U_{g_n, R_n, S}$, $\widetilde f_{g_n, S}$. But an inefficient extractor is enough for our information-theoretic argument.

We show that $\mathsf{Ext}$ can extract a valid pair $(x, y)$ in $g_n|_S$ with high probability.

\begin{lemma}
    The extractor $\Dec$ outputs $(u,g_n(u))$ for some $u\in S$ with probability at least $\alpha_n$, where
    \[
        \alpha_n
        :=
        \frac{\epsilon_n^2}{288q(n)^2}.
    \]
\end{lemma}
\begin{proof}    
If the extractor picks $x$ such that {\bf{Case 1}} holds, then we analyze the $\mathrm{coin}=0$ branch as follows:
Let $\mu_{x,t}$ be the probability that $\Dec$ returns $(u,g_n(u))$ for some $u\in S$ when applying the measurement right before the $t$-th query to $\widetilde f_{g_n,S}$ in the first branch.
The standard BBBV hybrid argument gives
\[
    \frac{\epsilon_n}{6}
    \leq
    2\sum_{t=1}^{q(n)}\sqrt{\mu_{x,t}}.
\]
Applying Cauchy--Schwarz,
\[
    \frac1{q(n)}
    \sum_{t=1}^{q(n)}\mu_{x,t}
    \geq
    \frac{\epsilon_n^2}{144q(n)^2} = 2 \alpha_n.
\]
Thus by definition of $\mu_{x,t}$, $\Dec$ outputs $(u, g_n(u))$ with probability at least $2\alpha_n$.

If the extractor picks $x$ such that {\bf{Case 2}} holds, by definition, the $\mathrm{coin}=1$ branch clearly outputs $(x,g_n(x))$ with probability $\frac{\epsilon_n}{6}\ge 2 \alpha_n$.

Thus, for every pick of $x$, the algorithm succeeds with probability at least $2\alpha_n$ if the branch matches the case of $x$. Therefore, $\mathsf{Ext}$ outputs $(u,g_n(u))$ for some $u\in S$ with probability at least $\alpha_n$.
\end{proof}

We next lift $\Dec$ from recovering one $g_n$ to $\Decfull$ as defined in Algorithm~\ref{alg:full-extractor}, which is designed to recover the entire $g_n|_{S}$.
\begin{algorithm}[H]
    \caption{The full extractor $\Decfull$}
    \label{alg:full-extractor}
    \begin{algorithmic}
        \State {\bf Input:} $S,g_n|_{\{0,1\}^n\backslash S}, c^\star,R_n$
        \State Set $S^{(0)}:=S$, $g_n^{(0)}|_{\{0,1\}^n\backslash S}:=g_n|_{\{0,1\}^n\backslash S}$
        \State Set $T:=|S|$
        \For{$i$ in $1$ to $T$}
            \State Set $(u,v):= \mathsf{Ext}(S^{(i-1)},g_n^{(i)}|_{\{0,1\}^n\backslash S^{(i)}},c^\star,R_n)$
            \State Set $S^{(i)}:= S\backslash  \{u\}$
            \State Set $g_n^{(i)}$ to be the function same as $g_n^{(i-1)}$ except for $g_n^{(i)}(u)=v$.
        \EndFor
        \State Output $g_n^{(T)},R_n$
    \end{algorithmic}
\end{algorithm}

With the chain rule for conditional probabilities, we have
\[
    \Pr\!\Bigl[
        \Decfull\text{ outputs }
        (g_n,R_n)
    \Bigr]
    \geq
    \alpha_n^{L_n}.
\]
For the counting argument, define
\[
    \mathcal H_n
    :=
    \Bigl\{
        g_n:
        \exists c\in\{0,1\}^{\leq w(n)}
        \text{ such that }
        Z^{g,R}(c)\geq \epsilon_n
    \Bigr\}.
\]

For each $g_n\in\mathcal H_n$, let $c_{g_n}$ be the lexicographically first advice such that $Z^{g,R}(c_{g_n})\geq \epsilon_n$, and let $S_{g_n}$ consist of the lexicographically first $L_n$ elements of $\mathsf{Good}_{g_n,c_{g_n}}$.
Define
\[
    \tau(g_n)
    :=
    \Bigl(
        c_{g_n},
        S_{g_n},
        \left.g_n\right|_{\{0,1\}^n\backslash S_{g_n}}
    \Bigr),
\]
and
$$\mathcal T_n=\{\tau(g_n):g_n\in\mathcal H_n\}\,.$$
For any fixed $\tau\in \mathcal T_n$, define
\[
F_\tau:=\{g_n\in\mathcal H_n:\tau(g_n)=\tau\}.
\]
We note that for any $\tau$ and $g_n \in F_\tau$, $\Decfull(\tau,R_n)$ outputs $g_n$ with probability at least $\alpha_n^{L_n}$, so for all possible $\tau$,
\[
    1\geq \sum_{g_n\in F_\tau}
    \Pr\!\Bigl[
        \Decfull(\tau,R_n)\text{ outputs }
        (g_n,R_n)
    \Bigr]\geq |F_\tau|\alpha_n^{L_n}\,.
\]
Note that $R_n$ is accessible to the extractor, counting on $\tau$ therefore gives
\begin{align*}
    |\mathcal H_n|\leq{}& |\mathcal T_n| \max_{\tau \in \mathcal T_n}\big| F_\tau \big| \\
    \leq{}&
    2^{w(n)+1}\, 
    \binom{2^n}{L_n}\,
    2^{3n(2^n-L_n)}\,
    \max_{\tau\in \mathcal T_n}\big| F_\tau \big|\\
    \leq{}&
    2^{w(n)+1}\,
    \binom{2^n}{L_n}\,
    2^{3n(2^n-L_n)}\,
    \alpha_n^{-L_n}.
\end{align*}
Using $\binom{2^n}{L_n}\leq2^{nL_n}$ and the definition of $\alpha_n$, and taking the logarithm, we obtain
\begin{align*}
    \log_2\!\Biggl(
        \frac{|\mathcal H_n|}{2^{3n2^n}}
    \Biggr)
    &\leq
    w(n)+1-2nL_n
    +
    L_n\log_2\!\Bigl(
        288\epsilon_n^{-2}q(n)^2
    \Bigr)
    \\
    &\leq
    -nL_n
\end{align*}
for all sufficiently large $n$.
Here we used the fact that, by definition,
\[
    \log_2\!\Bigl(
        288\epsilon_n^2q(n)^2
    \Bigr)
    =
    \frac{n}{25}+O(\log n),
\]
and $w(n)/L_n=o(1)$.
Hence
\[
    \Pr_{g_n}\!\bigl[g_n\in \mathcal H_n\bigr]
    \leq
    2^{-nL_n}
\]
for all sufficiently large $n$.
Choose $n_0$ so that this inequality holds for every $n\geq n_0$; we have
\[
    \sum_{n\geq1}\Pr_{g_n}\!\bigl[g_n\in \mathcal H_n\bigr]
    \leq
    n_0
    +
    \sum_{n\geq n_0}2^{-nL_n}
    <
    \infty.
\]
The Borel--Cantelli lemma (Lemma \ref{Borel-Cantelli}) therefore implies that, with probability $1$ over $(g, R)$, $g_n\in \mathcal H_n$ occurs for only finitely many $n$.
This shows that, with probability $1$, the adversary $\calA$ satisfies \begin{align*}
    \max_{c\in \{0,1\}^{\leq w(n)}} \EE_{x\gets \{0,1\}^n}\Pr[
        \calA^{\calO_{g,R}}(x,c)=g_n(x)
    ] < \epsilon_n.
\end{align*}
Since the set of all quantum algorithms is countable, this concludes the security of the EV-OWPuzz $(\mathsf{Samp},\mathsf{Ver})$ against non-uniform adversaries.
\end{proof}

\subsection{Non-existence of pseudodeterministic one-way functions}\label{sec:owf-nonexistence}

In this section, we show that no one-way functions exist, even if we allow pseudodetermninistic generation.

\begin{theorem}[No pseudodeterministic one-way functions]
\label{thm:no-pd-owf}
With probability $1$ over the choice of $g$ and $R$, there is no pseudodeterministic one-way function with inverse-polynomial error relative to $\calO_{g,R}$.
\end{theorem}

To prove Theorem~\ref{thm:no-pd-owf}, we will show that for any pseudodeterministic OWF generator candidate relative to $\mathcal O_{g,R}$, its pseudodeterministic output cannot depend on $g_n$ and $R_n$. This follows from a more general statement: the output of any algorithm that queries $U_{g_n,R_n}$ and $f_{g_n}$ can be approximated by an algorithm in $L^2$-norm without querying $U_{g_n, R_n}$ and $f_{g_n}$.

We first derive Theorem~\ref{thm:no-pd-owf} from Lemma~\ref{thm:main-concentration}.

\begin{proof}[Proof of Theorem~\ref{thm:no-pd-owf}]
To prove the theorem, we will need the following two lemmas.
\begin{lemma}\label{lem:pspace-simulation}
    For any pd-OWF generator $\Gen$ relative to $\mathcal O_{g, R}$ whose underlying function is $H$, there exists an extraction algorithm $\mathcal E$ and a $\mathbf{PSPACE}$ algorithm $\mathcal S$. The extraction algorithm $\mathcal E$ queries $\mathcal O_{g, R}$ polynomially many times and records classical information $Z, F$.  
    The simulation algorithm $\mathcal S$ can simulate the pd-OWF in the following sense: with probability at least $1- \frac{4}{\lambda^2}$ over $g,R$ and the intrinsic randomness of $\mathcal E$,
    \begin{align*}
        \#\{x:\mathcal S(x;Z,F) = H(x)\} \geq \frac{2^\lambda}{3}\,.
    \end{align*}
\end{lemma}

\begin{lemma}\label{lem:pspace-inverter}
    With probability $1$ over the choice of $g,R$, for any pd-OWF generator $\Gen$ and the corresponding $\mathbf{PSPACE}$ simulator $\mathcal S(\cdot;Z,F)$ in Lemma~\ref{lem:pspace-simulation}, we can invert the pd-OWF with the assistance of $\pspace$ with probability at least $1/9$.
\end{lemma}

The theorem is straightforward given Lemma~\ref{lem:pspace-simulation} and Lemma~\ref{lem:pspace-inverter}. Given any pd-OWF candidate $\Gen^{\mathcal O_{g,R}}$, there exists a efficient $\pspace$ simulator $\mathcal S$ according to Lemma~\ref{lem:pspace-simulation}. Then we can invert the pd-OWF with the simulator and $\pspace$ oracle as in~\ref{lem:pspace-inverter}.

\begin{proof}[Proof of Lemma~\ref{lem:pspace-simulation}]
    The simulator algorithm $\mathcal S$ will use the simulation algorithm in Lemma~\ref{thm:main-concentration}. The proof of Lemma~\ref{thm:main-concentration} is deferred to Section~\ref{sec:L2-concentration}.

    Let $T = 20 \log (\lambda q(\lambda))$. Given any pseudodeterministic OWF candidate $\Gen^{\mathcal O_{g,R}}$, according to Lemma~\ref{thm:main-concentration}, for any $n \geq T$, set $\epsilon=1/\lambda q(\lambda)$, $\eta=2^{-n/4}$, there exists a simulator $\widetilde{\mathcal S}$, independent of $x$, such that
    \begin{equation*}
    \Pr\left[\big\|\Gen^{U_{g_n,R_n},f_{g_n},\mathcal O}(1^\lambda,x)-\widetilde{\mathcal S}^{\mathcal O}(1^\lambda,x)\big\|_2\geq \frac{1}{2^{n/4}}+\frac{1}{\lambda q(\lambda)}+O\left(\frac{\lambda^2 q(\lambda)^4}{2^{n/2}}\right)\right]\\\leq \frac{\lambda^4 q(\lambda)^8}{2^{n/2}}\,.
    \end{equation*}
    
    Since there are at most $q(\lambda)$ many of such $n$, each of them is greater than $T$, apply the simulator to all $\lambda \in [T, q(\lambda)]$, we get that there exists a simulator (which we will still use $\widetilde{\mathcal S}$ as an abuse of notation), independent of $x$, that only queries $U_{g_n,R_n}$ and $f_{g_n}$ for $n\leq T$, and $\mathsf{PSPACE}$, such that
    $$\Pr\left[\big\|\Gen^{U_{g,R},f_{g},\mathsf{PSPACE}}(1^\lambda,x)-\widetilde{\mathcal S}^{U_{g,R},f_g,\mathsf{PSPACE}}(1^\lambda,x)\big\|_2\geq \frac{2}{\lambda}\right]\leq \frac{1}{\lambda^2}\,.$$

    This means
    $$\Pr\left[\#\left\{x:\big\|\Gen^{U_{g,R},f_{g},\mathsf{PSPACE}}(1^\lambda,x)-\widetilde{\mathcal S}^{U_{g,R},f_g,\mathsf{PSPACE}}(1^\lambda,x)\big\|_2\geq \frac{2}{\lambda}\right\}\geq \frac{1}{3}2^\lambda\right]\leq \frac{3}{\lambda^2}\,.$$

    Now we describe the behavior of the extractor. For each $n\leq T$, the extractor $\mathcal E$ performs process tomography~\ref{thm:unitary-tomography} to $U_{g_n,R_n}$ or its controlled version, producing the classical description $Z_n$ such that
    $$\Pr\left[\big\|U_{g_n,R_n}-\mathcal U(Z_n)\big\|_\diamond\geq \frac{1}{\lambda^2q(\lambda)^2}\right]\leq \frac{1}{\lambda^4}\,.$$
    This can be done by setting $\epsilon=\lambda^{-4}q(\lambda)^{-2}$, $\eta=\lambda^{-2}q(\lambda)^{-2}$, as in the notation of Theorem~\ref{thm:unitary-tomography}.
    It is clear that this tomography is of time $\poly(\lambda)$.
    Next, $\mathcal E$ queries $f_{g_n}$ to get the whole truth table, denoted $F_n$.
    Now, given $Z_n$ and $F_n$ for each $n\leq T$, we have the full simulation algorithm $\bar{\mathcal S}^{\mathsf{PSPACE}}(1^\lambda,x;Z,F)$, with $Z,F$ hardcoded into the algorithm, that simulates $\widetilde{\mathcal S}^{U_{g,R},f_g,\mathsf{PSPACE}}$, replacing each query to $U_{g_n,R_n}$ with $\mathcal U(Z_n)$, each query to $f_{g_n}$ with $F_n$.
    By the assumption,
    $$\Pr\left[\exists x,\ \big\|\widetilde{\mathcal S}^{U_{g,R},f_{g},\mathsf{PSPACE}}(1^\lambda,x)-\bar{\mathcal S}^{\mathsf{PSPACE}}(1^\lambda,x;Z,F)\big\|_2\geq\frac{1}{\lambda}\right]\leq \frac{1}{\lambda^2}\,.$$

    Therefore, because $$\Pr\left[\Pr[\Gen(1^\lambda,x)=H(x)]\geq 1-\frac{1}{p(x)}\right]\geq \frac{2}{3}\,,$$ we have
    $$\Pr\left[\#\left\{x:\Pr[\bar{\mathcal S}^{\mathsf{PSPACE}}(1^\lambda,x;Z,F)=H(x)]\geq 1-\frac{1}{p(\lambda)}-\frac{3}{\lambda}\right\}\geq \frac{2^\lambda}{3}\right]\geq 1-\frac{4}{\lambda^2}\,.$$

    Since the derandomization of $\mathbf{PSPACE}$ algorithm is also in $\mathbf{PSPACE}$, we can derandomize the $\mathbf{PSPACE}$ simulator $\bar{\mathcal S}$, and using majority vote to output a deterministic value. Let $\mathcal S$ be the derandomized simulator. Since $1-\frac{1}{p(\lambda)} - \frac 3\lambda \geq \frac 12$ for large enough $\lambda$, the deterministic $\mathbf{PSPACE}$ algorithm satisfies
    \begin{align*}
        \Pr\left[\#\left\{x:\mathcal S^{\mathsf{PSPACE}}(1^\lambda,x;Z,F)=H(x)\right\}\geq \frac{2^\lambda}{3}\right]\geq 1-\frac{4}{\lambda^2}\,.
    \end{align*}
\end{proof}

\begin{proof}[Proof of Lemma~\ref{lem:pspace-inverter}]
    Recall that $\calS^\pspace$ is a simulator of the pd-OWF generator whose underlying function is $H$.
    Define the set
    $$C:=\left\{x:\mathcal S^{\mathsf{PSPACE}}(1^\lambda,x;Z,F)=H(x)\right\}\,.$$
    Moreover, for $y\in \{0,1\}^{m(\lambda)}$, let
    $$C(y):=\left\{x:\mathcal S^{\mathsf{PSPACE}}(1^\lambda,x;Z,F)=H(x) = y\right\}\,,$$
    and
    \begin{equation*}
        S^{-1}(y):=\left\{x:\mathcal S^{\mathsf{PSPACE}}(1^\lambda,x;Z,F)=y\right\}\,.
    \end{equation*}
    By definition, $C = \bigcup_y C(y)$, and $C(y)\subseteq S^{-1}(y)$ for all $y$. 
    
    Define the inverter $\calA$ to take in $y$ as input and outputs a uniformly random element in $S^{-1}(y)$. This a probabilistic polynomial-space inverter.
    Therefore, the probability of successfully find a preimage of $y$ is at least $\frac{|C(y)|}{|S^{-1}(y)|}\,,$
    So the total probability for the adversary to find a preimage is
    \begin{align*}
        \EE_x \frac{|C\big(H(x)\big)|}{|S^{-1}\big(H(x)\big)|}
        \geq{} & \sum_{y\in \{0,1\}^{m(\lambda)}} \frac{|C(y)|^2}{2^\lambda|S^{-1}(y)|}\\
        \geq {} & \frac{\left(\sum_{y\in \{0,1\}^{m(\lambda)}}\big|C(y)\big|\right)^2}{\sum_{y\in \{0,1\}^{m(\lambda)}}2^\lambda|S^{-1}(y)|}\\
        \geq{} & \left(\frac{|C|}{2^\lambda}\right)^2\,.
    \end{align*}
    The second inequality here follows from the Cauchy--Schwarz inequality.
    By the requirements in Lemma~\ref{lem:pspace-simulation}, with at least $1-4/\lambda^2$ probability over the choice of $g,R$,
    \begin{equation*}
        |C|\ge \frac{2^\lambda}{3}.
    \end{equation*}
    In this case, using queries to $\pspace$, the adversary can invert with probability at least $1/9$.
    Moreover, since $\sum 4/\lambda^2< \infty$, applying Borel-Cantelli lemma (Lemma~\ref{Borel-Cantelli}) we get that with probability $1$ over the choice of $g, R$, we can invert one-way function $H$. 
    Using the fact that the set of all pd-OWF candidates is countable, summing up over the failure probability of all $\Gen$, we conclude that with probability $1$ over the choice of $g, R$, there does not exist a pseudodeterministic one-way function with inverse-polynomial error under the oracle $\mathcal O_{g,R}$.
\end{proof}
So we can conclude the proof of Theorem~\ref{thm:no-pd-owf}.
\end{proof}

\subsection{Two oracle replacements with concentration}
\label{sec:L2-concentration}
The remaining task is to prove the following.
\begin{lemma}
\label{thm:main-concentration}
Let $\calA$ be an efficient uniform quantum oracle algorithm with classical output; suppose that it makes at most $q$ queries to $U_{g_n,R_n},f_{g_n}$, and a side oracle $\calO$ independent with $g_n$ and $R_n$.
For every $0<\varepsilon<1$, there is a uniformly constructed simulator $\calS$, with access only to $\calO$, such that for every $x$ and every $0<\eta<1$,
\begin{align*}
  &\Pr_{g_n,R_n}\!\Biggl[
    \Bigl\|
      \calA^{U_{g_n,R_n},f_{g_n},\calO}(1^\lambda, x)
      -
      \calS^{\calO}(1^\lambda, x)
    \Bigr\|_2
    >
    \varepsilon+\eta+\frac{Cq^2}{\varepsilon^2 2^{n/2}}
  \Biggr]
  \leq
  \frac{Cq^4}{\varepsilon^4\eta^2\,2^n},
\end{align*}
where $C>0$ is a universal constant independent of $\calA$.
Moreover, the running time of $\calS$ is polynomial in the running time of $\calA$, $n$, and $1/\varepsilon$.
\end{lemma}

We prove Lemma~\ref{thm:main-concentration} by removing the $R_n$ and then the $g_n$ using Lemma~\ref{lem:phase-information-replacement} and Lemma~\ref{lem:graph-information-replacement}, respectively. 
Intuitively speaking, in Lemma~\ref{lem:phase-information-replacement}, we will simulate the unitary oracle $U_{g_n, R_n}$ with sampling $(u, g_n(u))$ where $u$ is uniformly random, hence erasing $R_n$. The main technique used in this step is the phase concentration lemma in Appendix~\ref{sec:phase-concentration}. Then in Lemma~\ref{lem:graph-information-replacement}, we further show that the simulation may not depend on $g_n$. Observe that any efficient algorithm should not distinguish whether a single $g_n(x)$ is resampled or not. The Efron--Stein inequality therefore indicates that most of $g_n$ could be replaced by a fresh resampled $g_n'$.

We now start introducing Lemma~\ref{lem:phase-information-replacement}.
\begin{lemma}
\label{lem:phase-information-replacement}
Let $\calA$ be an efficient uniform algorithm with classical output that makes at most $q$ queries to $U_{g_n,R_n}$, $f_{g_n}$, and a side oracle $\calO$. Let $\mathsf{Samp}_{g_n}$ be a CPTP oracle that outputs random samples of $(u, g_n(u))$.
Then for every $0<\varepsilon<1$, there is a uniform algorithm $\calA_{\mathrm{graph}}$ with at most $r = \Theta(q^2/\varepsilon^2)$ oracle queries, such that
\begin{equation*}
  \Pr_{R_n}\!\left[
    \left\|
      \calA^{U_{g_n,R_n},f_{g_n},\calO}(1^\lambda,x)
      -
      \calA_{\mathrm{graph}}^{\mathsf{Samp}_{g_n},f_{g_n},\calO}(1^\lambda,x)
    \right\|_2
    >
    \varepsilon+\alpha+\frac{Cq^2}{\varepsilon^2 2^{n/2}}
  \right]
  \leq
  \frac{Cq^2}{\varepsilon^2\alpha^2 2^n},
\end{equation*}
holds for every $x$ and $0<\alpha<1$.
The running time of $\calA_{\mathrm{graph}}$ is polynomial in the running time of $\calA$, $n$, and $1/\varepsilon$.
\end{lemma}

\begin{proof}
Fix $g_n$ and the classical input $(1^\lambda,x)$, which we suppress below.
Set $t:=\lceil c_0q^2/\varepsilon^2\rceil$ for a sufficiently large universal constant $c_0$, and write
\begin{equation*}
  \ket{\phi_{g_n,R_n}^{-}}
  :=
  \frac{\ket{0^{4n+1}}-\ket{1}\ket{\psi_{g_n,R_n}}}{\sqrt2}.
\end{equation*}
Then by definition $U_{g_n,R_n}=I-2\ketbra{\phi_{g_n,R_n}^{-}}{\phi_{g_n,R_n}^{-}}$.
Theorem~\ref{thm:JLS-simulation} replaces the reflection queries by $t$ copies of $\ket{\phi_{g_n,R_n}^{-}}$, giving an algorithm $\mathcal B^{f_{g_n},\calO}(\ket{\phi_{g_n,R_n}^{-}}^{\otimes t})$.
Measurement cannot increase trace norm, and the $\ell^2$-norm of a vector is at most its $\ell^1$-norm, so
\begin{equation}
  \label{eqa:phase-reflection-to-copies}
  \left\|
    \calA^{U_{g_n,R_n},f_{g_n},\calO}
    -
    \mathcal B^{f_{g_n},\calO}
    \bigl(\ket{\phi_{g_n,R_n}^{-}}^{\otimes t}\bigr)
  \right\|_2
  \leq
  \frac{C_1q}{\sqrt t}
  \leq
  \varepsilon,
\end{equation}
where the last inequality holds when $c_0\geq C_1^2$.

We now show the following using Theorem~\ref{thm:phase-concentration}. For every $\alpha > 0$,
\begin{equation}
  \label{eqa:phase-copies-concentration}
  \Pr_{R_n}\!\left[
    \left\|
      \mathcal B^{f_{g_n},\calO}
      \bigl(\ket{\phi_{g_n,R_n}^{-}}^{\otimes t}\bigr)
      -
      \EE_{R_n'}
      \mathcal B^{f_{g_n},\calO}
      \bigl(\ket{\phi_{g_n,R_n'}^{-}}^{\otimes t}\bigr)
    \right\|_2>\alpha
  \right]
  \leq
  \frac{2t}{\alpha^2 2^n}.
\end{equation}
We use the notations including $N, f, p_i$ in Theorem~\ref{thm:phase-concentration}. Set $N=2^n$, $f= R_n$, and
\begin{align*}
  p_0&=\frac12,
  &\ket{\phi_0}&=\ket{0^{4n+1}},\\
  p_u&=2^{-n-1},
  &\ket{\phi_u}&=-\ket1\ket u\ket{g_n(u)},
  \qquad u\in\{0,1\}^n.
\end{align*}
These choices give exactly $\ket{\phi_{g_n,R_n}^{-}}$.
The algorithm $\mathcal B^{f_{g_n},\calO}$ induces a fixed POVM on the $t$ copies, so the expected squared $\ell^2$ distance from the phase average is at most
\begin{equation*}
  4t\max_u p_u=\frac{2t}{2^n}.
\end{equation*}
Markov's inequality gives the claimed bound.

After the concentration, we now show how to simulate the average distribution. Theorem~\ref{thm:GMMY-simulation} gives an efficient simulator satisfying
\begin{equation*}
  \left\|
    \EE_{R_n'}\!\left[
      \ketbra{\phi_{g_n,R_n'}^{-}}{\phi_{g_n,R_n'}^{-}}^{\otimes t}
    \right]
    -
    \Sim^{\mathsf{Samp}_{g_n}}(1^t)
  \right\|_1
  \leq
  C_2t\,2^{-n/2}.
\end{equation*}
By contractivity of trace norm and the inequality between the norms, applying the same channel $\mathcal B^{f_{g_n},\calO}$ to both states yields
\begin{equation}
  \label{eqa:phase-average-to-sampler}
  \left\|
    \EE_{R_n'}
    \mathcal B^{f_{g_n},\calO}
    \bigl(\ket{\phi_{g_n,R_n'}^{-}}^{\otimes t}\bigr)
    -
    \mathcal B^{f_{g_n},\calO}
    \bigl(\Sim^{\mathsf{Samp}_{g_n}}(1^t)\bigr)
  \right\|_2
  \leq
  C_2t\,2^{-n/2}.
\end{equation}
Define the resulting algorithm by
\begin{equation*}
  \calA_{\mathrm{graph}}^{\mathsf{Samp}_{g_n},f_{g_n},\calO}
  :=
  \mathcal B^{f_{g_n},\calO}
  \bigl(\Sim^{\mathsf{Samp}_{g_n}}(1^t)\bigr).
\end{equation*}
Combining equations~\ref{eqa:phase-reflection-to-copies},~\ref{eqa:phase-copies-concentration}, and~\ref{eqa:phase-average-to-sampler}, the triangle inequality shows
\begin{equation*}
  \left\|
    \calA^{U_{g_n,R_n},f_{g_n},\calO}
    -
    \calA_{\mathrm{graph}}^{\mathsf{Samp}_{g_n},f_{g_n},\calO}
  \right\|_2
  \leq
  \varepsilon+\alpha+C_2t\,2^{-n/2}.
\end{equation*}
with failure probability $2t/(\alpha^2 2^n)$.
Substituting $t$ derives the desired result.
The efficient constructions in Theorems~\ref{thm:JLS-simulation} and~\ref{thm:GMMY-simulation} make sure the polynomial running-time bound. 
Moreover, $\calA_{\mathrm{graph}}$ makes $t$ sampler calls and $q$ retained calls to $f_{g_n}, \calO$, and $t+q=\Theta(q^2/\varepsilon^2)$.
\end{proof}
\begin{remark}
Simulation in Lemma~\ref{lem:phase-information-replacement} also applies to different types of oracle queries, including control, conjugate, inverse, transpose queries.
The state $\ket{\phi_{g_n,R_n}^{-}}$ has real amplitudes, so $U_{g_n,R_n}$ is a real symmetric unitary:
\begin{equation*}
  U_{g_n,R_n}^{-1}
  =\overline{U_{g_n,R_n}}
  =U_{g_n,R_n}^{\mathsf T}
  =U_{g_n,R_n}.
\end{equation*}
Moreover, the controlled unitary is the reflection
\begin{equation*}
  \operatorname{ctrl}(U_{g_n,R_n})
  =
  I-2\ketbra{1}{1}\otimes
  \ketbra{\phi_{g_n,R_n}^{-}}{\phi_{g_n,R_n}^{-}}.
\end{equation*}
Hence the above arguments work for all query types claimed before.
\end{remark}

Then, we will remove the dependence of $\calA$ on the sampling oracle and the choice of $g_n$.
\begin{lemma}
\label{lem:graph-information-replacement}
Let $\calA_{\mathrm{graph}}$ be an efficient and uniform algorithm with classical output that makes at most $r$ queries to $\mathsf{Samp}_{g_n}$, $f_{g_n}$, and a side oracle $\calO$.
Then there is a uniform simulator $\calS$ with at most $r$ access only to $\calO$ such that,
\begin{equation*}
  \Pr_{g_n}\!\left[
    \left\|
      \calA_{\mathrm{graph}}^{\mathsf{Samp}_{g_n},f_{g_n},\calO}(1^\lambda,x)
      -
      \calS^{\calO}(1^\lambda,x)
    \right\|_2
    >
    \beta+Cr2^{-3n/2}
  \right]
  \leq
  \frac{Cr^2}{\beta^2 2^n}.
\end{equation*}
for every $x$ and $0<\beta<1$,
The running time of $\calS$ is polynomial in the running time of $\calA_{\mathrm{graph}}$, $n$, and $r$.
\end{lemma}

\begin{proof}
We first show the concentration lemma on the output distribution. That is, for every $\beta > 0$,
\begin{equation}
  \label{eqa:graph-tail-concentration}
  \Pr_{g_n}\!\left[
    \left\|
      \calA_{\mathrm{graph}}^{\mathsf{Samp}_{g_n},f_{g_n},\calO}
      -
      \EE_{g_n'}
      \calA_{\mathrm{graph}}^{\mathsf{Samp}_{g_n'},f_{g_n'},\calO}
    \right\|_2>\beta
  \right]
  \leq
  \frac{Cr^2}{\beta^2 2^n}.
\end{equation}
Fix $u\in\{0,1\}^n$, define $\mathsf{Samp}_{g_n}^{\backslash u}$ to sample a uniform $v$ and return $(v,g_n(v))$ if $v\neq u$, and $(u,0^{3n})$ otherwise. We  also let $f_{g_n}^{\backslash u}(v,y):=\mathbf 1[v\neq u]f_{g_n}(v,y)$.
A hybrid over at most $r$ calls therefore gives, 
\begin{equation*}
  \left\|
    \calA_{\mathrm{graph}}^{\mathsf{Samp}_{g_n},f_{g_n},\calO}
    -
    \calA_{\mathrm{graph}}^{\mathsf{Samp}_{g_n}^{\backslash u},f_{g_n},\calO}
  \right\|_2
  \leq
  \frac{2r}{2^n}.
\end{equation*}
After this replacement, the only remaining dependence on $g_n(u)$ is in $f_n$, which could be further lifted by a BBBV-type argument in Theorem~\ref{BBBV}:
\begin{equation*}
  \EE_{g_n}\!\left[
    \left\|
      \calA_{\mathrm{graph}}^{\mathsf{Samp}_{g_n}^{\backslash u},f_{g_n},\calO}
      -
      \calA_{\mathrm{graph}}^{\mathsf{Samp}_{g_n}^{\backslash u},f_{g_n}^{\backslash u},\calO}
    \right\|_2^2
  \right]
  \leq
  Cr^2 2^{-3n}.
\end{equation*}
Let $g_n^u$ be identical to $g_n$ except that with a freshly sampled $g_n(u)$, clearly both experiment above in the equation cannot tell the difference between $g_n$ and $g_n^u$. Plus with triangle inequality,
\begin{equation*}
  \EE_{g_n,g_n^u}\!\left[
    \left\|
      \calA_{\mathrm{graph}}^{\mathsf{Samp}_{g_n},f_{g_n},\calO}
      -
      \calA_{\mathrm{graph}}^{\mathsf{Samp}_{g_n^u},f_{g_n^u},\calO}
    \right\|_2^2
  \right]
  \leq
  Cr^2\left(2^{-2n}+2^{-3n}\right).
\end{equation*}
Sampling $g_n$ uniformly amounts to sampling $g_n(x)$ independently for each $x$, applying the Hilbert-valued Efron--Stein inequality of Theorem~\ref{thm:efron-stein} hence giving
\begin{equation*}
  \label{eqa:graph-output-concentration}
  \EE_{g_n}\!\left[
    \left\|
      \calA_{\mathrm{graph}}^{\mathsf{Samp}_{g_n},f_{g_n},\calO}
      -
      \EE_{g_n'}
      \calA_{\mathrm{graph}}^{\mathsf{Samp}_{g_n'},f_{g_n'},\calO}
    \right\|_2^2
  \right]
  \leq
  Cr^2 2^{-n},
\end{equation*}
for a large constant $C$.
Finally, Markov's inequality proves Equation~\ref{eqa:graph-tail-concentration}.

It remains to simulate the mean distribution without querying $g_n$. 
Before the interaction, sample all entries of $X=\{x_1,\ldots,x_r\}$ independently and uniformly from $\{0,1\}^n$. Now for $i = 1,2,\dots,r$, if there is a $j < i$ such that $x_j = x_i$, pick $y_i = y_j$, else sample $y_i$ independently and uniformly from $\{0,1\}^{3n}$.
Write $\mathsf{Samp}_{X,Y}$ for the local procedure that returns $(x_j,y_{j})$ on its $j$-th call.
Let $T$ denote the sampled partial table: $T(x_j)=y_j$, and corresponding predicate function
\begin{equation*}
  f_T(u,v)
  :=
  \mathbf 1[u\in X\text{ and }v=T(u)].
\end{equation*}
The simulator $\calS^{\calO}$ samples this table and runs $\calA_{\mathrm{graph}}^{\mathsf{Samp}_{X,Y},f_{T},\calO}$ while simulating $\mathsf{Samp}_{X,Y}$ and $f_T$ locally.
Hence
\begin{equation*}
  \calS^{\calO}
  =
  \EE_{X,Y}
  \calA_{\mathrm{graph}}^{\mathsf{Samp}_{X,Y},f_T,\calO}.
\end{equation*}
For fixed $X,T$, define $\mathcal D_T$ as the following distribution over Boolean functions on $\{0,1\}^{4n}$: sample $g_n$ uniformly conditioned on extending $T$, and output
\begin{equation*}
  h_{g_n}(u,v)
  :=
  f_{g_n}(u,v)\oplus f_T(u,v)
  =
  \mathbf 1[u\notin X\text{ and }v=g_n(u)].
\end{equation*}
Moreover, for every fixed $(u,v)\in\{0,1\}^{4n}$,
\begin{equation*}
  \Pr_{h\gets\mathcal D_T}[h(u,v)=1]
  =
  \begin{cases}
    0, & u\in X,\\
    2^{-3n}, & u\notin X.
  \end{cases}
\end{equation*}
Apply Theorem~\ref{BBBV} using problem distribution $\mathcal D_T$, hence yields
\begin{equation}
  \label{eqa:graph-mean-simulation}
  \left\|
    \EE_{g_n'}
    \calA_{\mathrm{graph}}^{\mathsf{Samp}_{g_n'},f_{g_n'},\calO}
    -
    \calS^{\calO}
  \right\|_2
  \leq
  Cr2^{-3n/2}.
\end{equation}
Outside the exceptional event in equation~\ref{eqa:graph-tail-concentration}, the triangle inequality and equation~\ref{eqa:graph-mean-simulation} give
\begin{align*}
  &\left\|
    \calA_{\mathrm{graph}}^{\mathsf{Samp}_{g_n},f_{g_n},\calO}
    -
    \calS^{\calO}
  \right\|_2\\
  \leq&
  \left\|
    \calA_{\mathrm{graph}}^{\mathsf{Samp}_{g_n},f_{g_n},\calO}
    -
    \EE_{g_n'}
    \calA_{\mathrm{graph}}^{\mathsf{Samp}_{g_n'},f_{g_n'},\calO}
  \right\|_2
  +
  \left\|
    \EE_{g_n'}
    \calA_{\mathrm{graph}}^{\mathsf{Samp}_{g_n'},f_{g_n'},\calO}
    -
    \calS^{\calO}
  \right\|_2\\
  \leq&
  \beta+Cr2^{-3n/2}.
\end{align*}
This gives the claimed running time by construction.
\end{proof}

We now show the proof of Lemma~\ref{thm:main-concentration} formally.
\begin{proof}[Proof of Lemma~\ref{thm:main-concentration}]
Set $\alpha=\beta:=\eta/2$ and apply Lemma~\ref{lem:phase-information-replacement} to obtain $\calA_{\mathrm{graph}}$ with query bound $r = \Theta(q^2/\varepsilon^2)$, and then apply Lemma~\ref{lem:graph-information-replacement} to obtain $\calS$.
Outside the two exceptional events, the triangle inequality gives
\begin{align*}
  \left\|
    \calA^{U_{g_n,R_n},f_{g_n},\calO}(1^\lambda,x)
    -
    \calS^{\calO}(1^\lambda,x)
  \right\|_2
  &\leq
  \varepsilon+\alpha+\beta
  +\frac{Cq^2}{\varepsilon^2 2^{n/2}}
  +Cr2^{-3n/2}
  \\
  &\leq
  \varepsilon+\eta+
  \frac{Cq^2}{\varepsilon^2 2^{n/2}}.
\end{align*}
And the probability of the two exceptional events are at most
\begin{equation*}
  \frac{Cq^2}{\varepsilon^2\eta^2 2^n}
  +
  \frac{Cr^2}{\eta^2 2^n}
  \leq
  \frac{Cq^4}{\varepsilon^4\eta^2 2^n}.
\end{equation*}
This proves Lemma~\ref{thm:main-concentration}.
\end{proof}

\begin{proof}[Proof of Theorem~\ref{thm:4-main}]
    The theorem is straightforward by combining Theorem~\ref{thm:ev-owpuzz-existence} and Theorem~\ref{thm:no-pd-owf}.
\end{proof}

\section{Constructing QCCC Commitment from EV-OWPuzz with Nearly Uniform Puzzle Distribution}\label{sec:appendix-uniform}
As sketched in the technical overview section (Section~\ref{sec:tech-overview-commitment}), to prove the separation of EV-OWPuzz and QCCC commitment, we make essential use of the fact that the proposed EV-OWPuzz verifier can not only verify the ``hard" puzzles that are output by the sampler, but also a large fraction of other puzzles, where the puzzles might be easy to be inverted.
When all puzzles are ``hard" and get sampled (nearly) uniformly from the EV-OWPuzz sampler, it is possible to construct a QCCC commitment scheme from it.
We formalize this intuition by presenting Theorem~\ref{thm:uniform-puzzle-to-commitment}.

\begin{lemma}[Quantum singular value transformation (QSVT), \cite{GSLW-QSVT}]\label{lem:QSVT}
    Let $K$ be a matrix such that $K$ has the following implementation
    $$K=(\bra{1}\otimes I) C (I\otimes \ket{0^\eta})\,,$$
    where $C$ is a unitary and $\eta$ is the number of ancillas. Write the singular value decomposition of $K$ as $K=U\Sigma V^\dagger$. Then, for any real odd polynomial $p$ of degree $d$ such that
    $$\forall x\in [-1,1],\ |p(x)|\leq 1\,,$$
    there exists an implementation of $Up(\Sigma)V^\dagger$, which makes $O(d)$ queries to $C$ and $C^\dagger$ and $\poly(d,\eta)$ other gates.
\end{lemma}

\begin{theorem}\label{thm:uniform-puzzle-to-commitment}
    Let $(\mathsf{Samp},\mathsf{Ver})$ be an EV-OWPuzz. Suppose the marginal distribution of the puzzle is $\epsilon$-close to the uniform distribution, where $\epsilon=\epsilon_\lambda$ is negligible in $\lambda$.
    Then, there is a statistically $2\epsilon$-hiding and computationally sum-binding QCCC bit commitment scheme.
\end{theorem}

\begin{proof}
    Without loss of generality, we assume that the puzzle length is $m=m(\lambda)$, and the key length is $\lambda$.

    \emph{Construction.} The honest protocol is as follows. To commit a bit $b$:
\begin{enumerate}
    \item The committer samples \((s,k)\leftarrow\mathsf{Samp}(1^\lambda)\).
    \item Then, the committer and the receiver interact $m-1$ rounds. In the $i$-th round, the receiver samples a uniform \(r_i\in\{0,1\}^{m-i}\) and sends $r_i$ to the committer. The committer sets $h_i=0^{i-1}1r_i$, computes \(y_i=(h_i\cdot s )\bmod 2\), and sends $y_i$ to the receiver.\footnote{This is exactly the interactive hashing scheme used in the construction of~\cite{NOVY}.}
   \item The set of equations \(\{h_i\cdot z=y_i\}\) on the vector space $\mathbb F_2^n$ has exactly two solutions \(z_0<z_1\) (ordered lexicographically). Let \(s=z_d\). The committer sends \(c=b\oplus d\).
   \item At the opening stage, the committer sends \((b,d,s,k)\). The receiver checks \(c=b\oplus d\), \(s=z_d\), and accepts exactly when \(\mathsf{Ver}(s,k)\) accepts.
\end{enumerate}

    The correctness of this protocol follows directly from the correctness of the EV-OWPuzz. If $\epsilon=0$, hiding follows directly from~\cite[Lemma~1]{NOVY}. When $\epsilon>0$, simply notice that changing the true puzzle distribution to the uniform distribution changes each of the two commit stage transcript by at most $\epsilon$.
    
    \emph{Binding.} Suppose there is a quantum malicious committer that wins the sum-binding game with advantage at least $\delta$. We will gradually describe how to build the EV-OWPuzz adversary that breaks the security with advantage $\Omega(\poly(\delta)/m-\epsilon-\negl(\lambda))$ based on this malicious committer, which is a black-box security reduction.

    The general idea to build the EV-OWPuzz adversary is to simulate the malicious committer and somehow ``force" its outputs at step (c) to contain $s$, where $s$ is the puzzle one would like to invert. We describe how to ``force" the malicious committer as follows.

    We select $i\in [m-1]$ and focus on the $i$-th round of step (b).
    Conditioned on the previous transcript, the remaining puzzle candidates form an affine subspace $V\subset \mathbb F_2^m$, whose final $m-i+1$ coordinates are mutually different. For simplicity, we represent elements in $V$ by their final $m-i+1$ coordinates, and omit the dependency on $i$ and the previous transcript when it is clear from context.

    Without loss of generality, we may assume the committer's private state is always pure, and, conditioned on the previous transcript $\tau$ and the challenge $r=r_i$, we assume the malicious committer uses $\eta$ qubits of ancillas that initialized to $\ket0$, performs some circuit $C_r$, and measures the first qubit as its answer. In other words, the malicious committer applies
    $$K_{\tau,r,b}=(\bra{b}\otimes I)C_r(I\otimes \ket{0^\eta})$$
    to its private state, when it gets the challenge $r$ and answers $b$.
    Define $E_{\tau,r,b}=K_{\tau,r,b}^\dagger K_{\tau,r,b}$.
    
    For any $s\in V$, write $s=(a,z)$ where $a\in \{0,1\}$, $z\in \{0,1\}^{m-i}$. To keep $s$ as a remaining candidate puzzle after the $i$-th round, given a challenge $r\in \{0,1\}^{m-i}$, the committer has to \emph{postselect} the result $a+r\cdot z$, i.e. to perform the following operation
    \begin{equation}\label{eqa:postselection}
        \rho\mapsto \frac{K_{\tau,r,a+r\cdot z}\rho K_{\tau,r,a+r\cdot z}^\dagger}{\Tr(E_{\tau,r,a+r\cdot z}\rho)}\,.
    \end{equation}
    This is not efficiently implementable. However, we will prove that there exists a good approximation of this postselection, for a typical choice of $s=(a,z)$ and after purification by the challenge $r$.

    Define
    $$A_{\tau,s}=\EE_{r\gets \{0,1\}^{m-i}}\left[E_{\tau,r,a+r\cdot z}\right],\quad D_{\tau,s}=A_{\tau,s}-\frac{I}{2}\,.$$
    We prove the following analog of \cite[Claim~1]{NOVY}.
    \begin{lemma}\label{lem:resampler-variance}
        $$\EE_{s\gets V} (D_{\tau,s})^2\preceq \frac{I}{2^{m-i+2}}\,,\quad \EE_{s\gets V}A_{\tau,s}D_{\tau,s}^2\preceq \frac{I}{2^{m-i+3}}\,.$$
    \end{lemma}
    \begin{proof}
        By definition,
        \begin{align*}
            D_{\tau,s}={}&\frac{1}{2^{m-i}}\sum_{r\in \{0,1\}^{m-i}}\left(E_{\tau,r,a+r\cdot z}-\frac{I}{2}\right)\\
            ={}& \frac{1}{2^{m-i}}\sum_{r\in \{0,1\}^{m-i}}\frac{1}{2}\left(E_{\tau,r,a+r\cdot z}-E_{\tau,r,1+a+r\cdot z}\right)\\
            ={} & \frac{1}{2^{m-i+1}}\sum_{r\in \{0,1\}^{m-i}}(-1)^{a+r\cdot z}(E_{\tau,r,0}-E_{\tau,r,1})\,.
        \end{align*}
        Set $B_{\tau,r}=E_{\tau,r,0}-E_{\tau,r,1}$, then
        \begin{align*}
            \sum_{s}(D_{\tau,s})^2={} & \sum_{a,z}\frac{1}{2^{2(m-i+1)}}\sum_{r\in \{0,1\}^{m-i}}\sum_{r'\in \{0,1\}^{m-i}}(-1)^{(r\oplus r')\cdot z}B_{\tau,r}B_{\tau,r'}\\
            ={} & \frac{1}{2^{2(m-i)+1}}\sum_{r\in \{0,1\}^{m-i}}2^{m-i}(B_{\tau,r})^2\\
            ={}& \frac{1}{2^{m-i+1}}\sum_{r\in \{0,1\}^{m-i}}(B_{\tau,r})^2\,.
        \end{align*}
        Since $(B_{\tau,r})^2\preceq I$, we have
        $$\sum_{s}(D_{\tau,s})^2\preceq \frac{I}{2}\,.$$
        Notice that for $s'=(a\oplus 1,z)$, we have $A_{\tau,s'}=I-A_{\tau,s}$, $D_{\tau,s'}=-D_{\tau,s}$. Thus,
        \begin{align*}
            \sum_s A_{\tau,s}(D_{\tau,s})^2={}& \sum_{z}(A_{\tau,0,z}D_{\tau,0,z}^2+A_{\tau,1,z}D_{\tau,1,z}^2)\\
            ={} & \sum_z A_{\tau,0,z}D_{\tau,0,z}^2+(I-A_{\tau,0,z})D_{\tau,0,z}^2\\
            ={} & \sum_z D_{\tau,0,z}^2=\frac{1}{2}\sum_s D_{\tau,s}^2\preceq \frac{I}{4}\,.
        \end{align*}
        This concludes the proof of Lemma~\ref{lem:resampler-variance}.
    \end{proof}

    We now begin to analyze the error of this QSVT compared to the ideal postselection (i.e.\ renormalization).
    For $s\in V$, write the sampling of $r$ and the measurement jointly as the following operator $G_{\tau,s,i}$:
    \begin{align*}
        G_{\tau,s,i}
        ={}& \frac{1}{\sqrt{2^{m-i}}}\sum_{r\in \{0,1\}^{m-i}}\ket{r}\otimes \ket{a+r\cdot z}\otimes K_{\tau,r,a+r\cdot z}\,.
    \end{align*}
    One may think of $G_{\tau,s,i}$ as recording the $i$-th round transcript $(r,a+r\cdot z)$ and coherently applying $K_{\tau,r,a+r\cdot z}$. Note that $G_{\tau,s,i}^\dagger G_{\tau,s,i}=A_{\tau,s}$.
    Define
    $$f_1(x)=\sqrt{2}\left(\frac{3}{2}x-x^3\right)\,.$$
    Conditioned on the previous transcript, we implement the QSVT on $G_{\tau,s,i}$ with $f_1$. Suppose $G_{\tau,s,i}=U\Sigma V^\dagger$ is the singular value decomposition of $G_{\tau,s,i}$. Since $G_{\tau,s,i}A_{\tau,s}=G_{\tau,s,i}G_{\tau,s,i}^\dagger G_{\tau,s,i}=U\Sigma^3V^\dagger$, the QSVT computes the following operator
    $$F_{1,\tau,s,i}=\sqrt{2}G_{\tau,s,i}\left(\frac{3}{2}I-A_{\tau,s}\right)\,.$$
    Thus, for any $s$, $F_{1,\tau,s,i}-\sqrt{2}G_{\tau,s,i}=-\sqrt{2}G_{\tau,s,i}D_{\tau,s}$, and
    \begin{equation}\label{eqa:sqrt2-bound}
        \EE_s (F_{1,\tau,s,i}-\sqrt{2}G_{\tau,s,i})^\dagger(F_{1,\tau,s,i}-\sqrt{2}G_{\tau,s,i})=2\EE_s A_{\tau,s}D_{\tau,s}^2\preceq \frac{I}{2^{m-i+2}}\,.
    \end{equation}
    
    We have now bounded the error introduced by replacing the ideal postselection by a QSVT algorithm. However, when $i$ is large (hence $m-i$ is small), the bound~\eqref{eqa:sqrt2-bound} may become too loose.
    As discussed in the technical overview section (Section~\ref{sec:tech-overview-novy}), we need to do a ``mini" renormalization by a multiplicative factor of $5/4$. Define
    $$f_2(x)=\frac{3}{2}x-\frac{1}{2}x^3\,.$$
    Similar to above, we implement the QSVT on $G_{\tau,s,i}$ with the polynomial $f_2$ to compute the following operator
    $$F_{2,\tau,s,i}=\frac{5}{4}G_{\tau,s,i}\left(\frac{6}{5}I-\frac{2}{5}A_{\tau,s}\right)\,.$$
    Thus,
    $$F_{2,\tau,s,i}-\frac{5}{4}G_{\tau,s,i}=-\frac{1}{2}G_{\tau,s,i}D_{\tau,s}\,,$$
    and
    \begin{equation}\label{eqa:5/4-bound}
        \EE_s \left(F_{2,\tau,s,i}-\frac{5}{4}G_{\tau,s,i}\right)^\dagger \left(F_{2,\tau,s,i}-\frac{5}{4}G_{\tau,s,i}\right)=\frac{1}{4}\EE_s A_{\tau,s}D_{\tau,s}^2\preceq\frac{I}{2^{m-i+5}}\,.
    \end{equation}
    
    We now start to formally describe the adversary of the EV-OWPuzz that is built on the malicious committer. Define $F_{1,s,i}$ (resp.\ $F_{2,s,i}$) the operator that coherently applies $F_{1,\tau,s,i}$ (resp.\ $F_{2,\tau,s,i}$) to the committer's private state controlled on the previous transcript $\tau$.
    \begin{algorithm}[H]
        \caption{The EV-OWPuzz adversary}
        \label{alg:uniform-ev-owpuzz-reduction}
        \begin{algorithmic}
            \State \textbf{Input}: a puzzle $s\in \{0,1\}^m$
            \State Sample $T\gets \{1,2,\cdots,m-1\}$
            \State Prepare the malicious committer's private (pure) state $\ket{\phi_{s,0}}=\ket{\phi_0}$ before step (b)
            \For{$i=1,2,\cdots ,T$}
                \State Compute $\ket{\phi_{s,i}}:=F_{1,s,i}\ket{\phi_{s,i-1}}$
            \EndFor
            \For{$i=T+1,\cdots,m-1$}
                \State Compute $\ket{\phi_{s,i}}:=F_{2,s,i}\ket{\phi_{s,i-1}}$
            \EndFor
            \State Measure the transcript register of $\ket{\phi_{s,m-1}}$ to get $r_i,y_i,d$ and the committer's private state
            \State Simulate step (c) of the malicious committer to get $c$
            \State Simulate the opening of $b=c\oplus d$ to get $k$
            \State \textbf{Output} $k$
        \end{algorithmic}
    \end{algorithm}
    
    We start to show that Algorithm~\ref{alg:uniform-ev-owpuzz-reduction} succeeds in breaking the security of the EV-OWPuzz, as long as the malicious committer has non-negligible advantage in the sum-binding game. Since the puzzle distribution of $\mathsf{Samp}$ is nearly uniform, we only need to prove that Algorithm~\ref{alg:uniform-ev-owpuzz-reduction} finds a key for a uniformly random puzzle $s$.

    \begin{lemma}\label{lem:uniform-adv-prob}
        Denote $\mathrm{adv}_{s,T}$ the probability that the output $k$ of Algorithm~\ref{alg:uniform-ev-owpuzz-reduction} passes the verification when given the input $s$ and the sampled $T$. Let $q$ be the probability that the malicious committer wins the sum-binding game (Definition~\ref{def:QCCC-commitment}). Then
        $$\EE_{s\gets \{0,1\}^m} \mathrm{adv}_{s,T}\geq \frac{\left(\max\left\{0,(5/4)^{m-1-T}(\sqrt{2q}-1)-1/\sqrt{2}\right\}\right)^2}{2^{m-T}}\,.$$
    \end{lemma}
    \begin{proof}
        Fix $i\in [m-1]$, let $\tau=(r_1,y_1,r_2,y_2,\cdots,r_{i},y_{i})$ be the transcript after the $i$-th round, and let $\ket{\psi_\tau}$ be the private state of the malicious committer after the $i$-th round (with some appropriate global phase), conditioned on the previous transcript $\tau$. Let $p_\tau$ be the probability of getting the transcript $\tau$.

        We first argue about the vanilla execution of the interaction between the malicious committer and an honest receiver. To invert a puzzle $s$, we reject all branches that produce a transcript inconsistent with $s$ (but not postselect or renormalize). Purified by the previous transcript, the committer's private state should be
        $$\ket{\psi_{s,i}}:=\sum_{\tau:s\in V_\tau}\ket{\tau}\otimes \sqrt{p_\tau}\ket{\psi_\tau}\,,$$
        where $V_\tau$ be the (affine) subspace of all puzzles consistent with $\tau$.
        Averaging over $s$ yields
        \begin{align*}
            \EE_s\|\ket{\psi_{s,i}}\|^2={}&\frac{1}{2^m}\sum_s \sum_{\tau:s\in V_\tau}p_\tau\\
            ={}& \frac{1}{2^m}\sum_\tau \sum_{s\in V_\tau}p_\tau\\
            ={} & \frac{1}{2^{i}}\sum_\tau p_\tau=\frac{1}{2^i}\,.
        \end{align*}
        
        We denote $G_{s,i}$ the operator of coherently applying $G_{\tau,s,i}$ to the committer's private state controlled on the previous transcript $\tau$. Note that $G_{s,i}\ket{\psi_{s,{i-1}}}$ is exactly $\ket{\psi_{s,i}}$.

        Next, set
        $$u_i=\begin{cases}
            \sqrt{2} & i\leq T\\
            5/4 & i>T
        \end{cases}\,,\quad F_{\tau,s,i}=\begin{cases}
            F_{1,\tau,s,i} & i\leq T\\
            F_{2,\tau,s,i} & i>T
        \end{cases}\,.$$
        Define $\gamma_i=u_1u_2\cdots u_i$, and
        $$\widetilde G_{s,i}=u_iG_{s,i}\,.$$
        Then, by Equation~\eqref{eqa:sqrt2-bound} and Equation~\eqref{eqa:5/4-bound},
        \begin{align*}
            &\EE_s \left\|(F_{s,i}-\widetilde G_{s,i})\ket{\psi_{s,i-1}}\right\|^2\\
            ={} & \EE_s\left\|\sum_{\tau:s\in V_\tau}\ket{\tau}\otimes \sqrt{p_\tau}(F_{\tau,s,i}-\widetilde G_{\tau,s,i})\ket{\psi_\tau}\right\|^2\\
            \leq {} & \frac{1}{2^m}\sum_{s\in \{0,1\}^m} \sum_{\tau:s\in V_\tau} p_\tau\braket{\psi_\tau|(F_{\tau,s,i}-\widetilde G_{\tau,s,i})^\dagger (F_{\tau,s,i}-\widetilde G_{\tau,s,i})|\psi_\tau} \\
            ={} & \frac{2^{m-i+1}}{2^m}\sum_{\tau}p_\tau \braket{\psi_\tau|\EE_{s\in V_\tau}(F_{\tau,s,i}-\widetilde G_{\tau,s,i})^\dagger (F_{\tau,s,i}-\widetilde G_{\tau,s,i})|\psi_\tau}\\
            \leq{} & \begin{cases}
                2^{-(m+1)} & i\leq T\\
                2^{-(m+4)} & i>T
            \end{cases}\,.
        \end{align*}
        and hence,
        \begin{equation}\label{eqa:intermediate-hybrid}
            \EE_s \left\|(F_{s,i}-\widetilde G_{s,i})\gamma_{i-1}\ket{\psi_{s,i-1}}\right\|^2\leq \begin{cases}
            \gamma_{i-1}^22^{-(m+1)} & i\leq T\\
            \gamma_{i-1}^2 2^{-(m+4)} & i>T
        \end{cases}\,.
        \end{equation}

        Now we construct a hybrid as follows. We compare the action of $F_{s,m-1}F_{s,m-2}\cdots F_{s,1}$ with $\widetilde G_{s,m-1}\widetilde G_{s,m-2}\cdots \widetilde G_{s,1}$ on the initial state $\ket{\phi_0}$. The intermediate hybrid would be
        $$F_{s,m-1}\cdots F_{s,i+1}(F_{s,i}-\widetilde G_{s,i})\widetilde G_{s,i-1}\cdots \widetilde G_{s,1}\,.$$
        Thus, the left part is a physical operation (a contraction), and the right part would be
        $$\widetilde G_{s,i-1}\cdots \widetilde G_{s,1}\ket{\phi_0}=\gamma_{i-1}\ket{\psi_{s,i-1}}\,.$$
        Notice that the root-mean-square of a norm is still a norm, Equation~\eqref{eqa:intermediate-hybrid} yields
        \begin{align}
            & \sqrt{\EE_s\|\ket{\phi_{s,m-1}}-\gamma_{m-1}\ket{\psi_{s,m-1}}\|^2}\notag\\
            \leq{} & \sum_{i=1}^{m-1} \sqrt{\EE_s \left\|(F_{s,i}-\widetilde G_{s,i})\gamma_{i-1}\ket{\psi_{s,i-1}}\right\|^2}\notag\\
            \leq{} &  \frac{1}{\sqrt{2^{m+1}}}\sum_{i=1}^T(\sqrt{2})^{i-1}+\sqrt{\frac{2^T}{2^{m+4}}}\sum_{i=1}^{m-1-T}\left(\frac{5}{4}\right)^{i-1}\notag \\
            \leq{} & \frac{1+1/\sqrt{2}}{\sqrt{2^{m-T}}}+\frac{(5/4)^{m-1-T}-1}{\sqrt{2^{m-T}}}\label{eqa:uniform-error-term}
        \end{align}

        Next, we use that the malicious committer wins the sum-binding game with probability $q$. Let $\{H_{\tau,0},H_{\tau,1}\}$ correspond to the circuit used to determine $c$, and let $M_{\tau,c,b}$ correspond to the successful opening of the bit $b$. Then,
        $$q=\frac{1}{2}\sum_{b\in \{0,1\}}\sum_\tau p_\tau \sum_{c\in \{0,1\}}\left\|M_{\tau,c,b}H_{\tau,c}\ket{\psi_{\tau}}\right\|^2$$
        Thus, defining $M_{c,b}$ and $H_b$ the coherent version of $M_{\tau,c,b}$ and $H_{\tau,b}$ respectively, we have
        \begin{align*}
            & \EE_{s\gets\{0,1\}^m} \sum_{c}\left\|M_{c,c\oplus d}H_{c}\gamma_{m-1}\ket{\psi_{s,m-1}}\right\|^2\\
            ={} &\frac{1}{2^m}\sum_{s\in \{0,1\}^m}\sum_{\tau:s\in V_\tau}p_\tau\sum_{c}\left\|M_{\tau,c,c\oplus d}H_{\tau,c}\gamma_{m-1}\ket{\psi_{\tau}}\right\|^2\\={}&\frac{2q\gamma_{m-1}^2}{2^m}=\frac{1}{2^{m-T}}\left(\frac{5}{4}\right)^{2(m-1-T)}2q\,,
        \end{align*}
        and after averaging over $s$, the probability that Algorithm~\ref{alg:uniform-ev-owpuzz-reduction} finds a key that passes the verification is
        \begin{align*}
            \sqrt{\EE_{s}\mathrm{adv}_{s,T}}={}& \sqrt{\EE_{s\gets \{0,1\}^m}\sum_{c\in \{0,1\}}\left\|M_{c,c\oplus d}H_{c}\ket{\phi_{s,m-1}}\right\|^2}\\
            \geq{}& \sqrt{\EE_{s} \sum_{c}\left\|M_{c,c\oplus d}H_{c}\gamma_{m-1}\ket{\psi_{s,m-1}}\right\|^2}-\sqrt{\EE_{s}\sum_c\|M_{c,c\oplus d}H_c(\ket{\phi_{s,m-1}}-\gamma_{m-1}\ket{\psi_{s,m-1}})\|^2}\\
            \geq{} & \sqrt{\EE_{s} \sum_{c}\left\|M_{c,c\oplus d}H_{c}\gamma_{m-1}\ket{\psi_{s,m-1}}\right\|^2}-\sqrt{\EE_{s}\|\ket{\phi_{s,m-1}}-\gamma_{m-1}\ket{\psi_{s,m-1}}\|^2}\\
            \geq{}&  \frac{(5/4)^{m-1-T}(\sqrt{2q}-1)-1/\sqrt{2}}{\sqrt{2^{m-T}}}\,.
        \end{align*}
        This concludes the proof of Lemma~\ref{lem:uniform-adv-prob}.
    \end{proof}

    Suppose the malicious committer wins the sum-binding game with probability $(1+\delta)/2$, where $0\leq \delta \leq 1$. Then $\sqrt{2q}-1\geq \delta/3$. Choosing $$T=m-1-\log_{5/4}\frac{6}{\delta}\,,$$ we have
    $$\EE_{s} \mathrm{adv}_{s,T}\geq \frac{(6/\delta)\cdot(\delta/3)-1/\sqrt{2}}{2\cdot (6/\delta)^{\log_{5/4}2}}\geq 0.002\,\delta^{3.106\dots}\,.$$
    If $\delta$ is non-negligible, Algorithm~\ref{alg:uniform-ev-owpuzz-reduction} samples such $T$ with at least $1/m$ probability, which is an inverse-polynomial.
    The average succeeding probability of Algorithm~\ref{alg:uniform-ev-owpuzz-reduction} towards a uniformly random puzzle is thus non-negligible. And finally, since the puzzle distribution is $\epsilon$-close from uniform, Algorithm~\ref{alg:uniform-ev-owpuzz-reduction} succeed in breaking the security of the EV-OWPuzz.
\end{proof}

\section{An Oracle Separation of EV-OWPuzz from QCCC Commitment Scheme}\label{sec:commit-separation}
This section is dedicated to the following theorem.
\begin{theorem}\label{thm:5-main}
    There exists a unitary oracle relative to which EV-OWPuzz exists and QCCC commitment schemes does not. Moreover, the separation holds even if we allow controlled access and access to the inverse, conjugate, and transpose.
\end{theorem}
Similar to Section \ref{sec:owf-separation}, this section is organized as follows. Section \ref{sec:commitment-oracle} introduces the construction of the oracle. Section \ref{sec:commitment-ev-owpuzz} shows that a secure EV-OWPuzz exists relative to the oracle with a BBBV-type argument. Section \ref{sec:commitment-nonexistence} shows that QCCC commitment schemes do not exist relative to the oracle. The core of the proof is illustrated in Section~\ref{sec:LOCC-decoupling}, which shows that when the committer and the receiver are given independent oracles, the transcripts are still approximately the same.   

\subsection{Oracle construction}\label{sec:commitment-oracle}

We now construct the oracle.
For each $n$, let
\begin{equation*}
    H_n:\{0,1\}^n\longrightarrow\{0,1\}^{4n}
\end{equation*}
be an uniformly random injection, and
\begin{equation*}
    g_n:\{0,1\}^n\longrightarrow\{0,1\}^{4n}
    \qquad\text{and}\qquad
    R_n:\{0,1\}^n\longrightarrow\{0,1\}.
\end{equation*}
be uniformly random functions.
The image of $H_n$ is the ``honest'' puzzle region, which has size only $2^n \ll 2^{4n}$.
However, directly using arbitrary keys for those ``cheating'' (or easy-to-invert) puzzles will be immediately exposed.
In order to hide the valid keys associated with each puzzle, independently sample
\begin{equation*}
    J_n:\{0,1\}^{4n}\times\{0,1\}^{4n}
    \longrightarrow
    \{0,1\}^{8n}
\end{equation*}
uniformly from the maps for which $J_n(z,\cdot)$ is injective for every $z\in\{0,1\}^{4n}$. We will always use $J_n(H_n(x), g_n(x))$ as the puzzle corresponding to $H_n(x)$ instead of plain $g_n(x)$.

Define $f_n:\{0,1\}^{4n}\times\{0,1\}^{8n}\to \{0,1\}$ as follows:
\begin{equation*}
    f_n(z,k)
    :=
    \begin{cases}
        \mathbf 1[k=J_n(z,g_n(x))],
        & z=H_n(x)\text{ for some }x\in\{0,1\}^n,\\
        \mathbf 1[k\in\Im(J_n(z,\cdot))],
        & z\notin\Im(H_n).
    \end{cases}
\end{equation*}
It is easy to see that by injectivity, every puzzle $z$ in the honest region $\Im H$ has exactly one valid key $J_n(z, g_n(H_n^{-1}(z)))$, whereas every puzzle outside this region has $2^{4n}$ valid keys $\Im(J_n(z, \cdot))$, and could be easily inverted.
At the same time, the valid keys (even for a puzzle $z\notin\Im H$) occupy only at most $2^{-4n}$ fraction of the key space, so their locations are hidden without the puzzle.

Define the graph state
\begin{equation*}
    \ket{\phi_n}
    :=
    2^{-n/2}
    \sum_{x\in\{0,1\}^n}
        (-1)^{R_n(x)}
        \ket{H_n(x)}
        \ket{J_n(H_n(x),g_n(x))}.
\end{equation*}
Let $U_n$ be the unitary that swaps $\ket{0^{12n+1}}$ and $\ket{1}\ket{\phi_n}$, equivalently,
\begin{equation*}
    U_n
    :=
    I-
    \left(
        \ket{0^{12n+1}}-\ket{1}\ket{\phi_n}
    \right)
    \left(
        \bra{0^{12n+1}}-\bra{1}\bra{\phi_n}
    \right).
\end{equation*}
Write $U=(U_n)_{n\in\mathbb N_+}$, $f=(f_n)_{n\in\mathbb N_+}$, and similarly write $H,g,R,J$ for the corresponding sequences.
Note that $U$ and $f$ are both dependent on $H, g, R$ and $J$.

Our oracle family will consist of $U$, $f$, and a classical oracle $\Help$ that is constructed from diagonalizing techniques and allows us to essentially break all QCCC commitment schemes. We defer the explicit construction of the oracle $\Help$ to Section~\ref{sec:commitment-nonexistence}, but claim that the oracle satisfying the following lemma, whose proof is also postponed to Section~\ref{sec:commitment-nonexistence}.

\begin{lemma}\label{lem:help-inaccessible}
For every polynomial $p$, there exists $n_0 \in \mathbb N$, such that for any $n \geq n_0$, conditioned on $J$, the restriction of $\Help$ to queries of width $d\leq p(n)$ is (jointly) independent of $H_n,g_n,R_n$.
\end{lemma}

\subsection{EV-OWPuzz survives}\label{sec:commitment-ev-owpuzz}

Now we state the secure EV-OWPuzz construction relative to $(U, f, \Help)$.
On input $1^\lambda$, the sampler $\Samp$ applies $U_\lambda$ to $\ket{0^{12\lambda+1}}$, measures in the computational basis, and obtains
\begin{equation*}
  1,H_\lambda(x),J_\lambda(H_\lambda(x),g_\lambda(x))
\end{equation*}
for a uniform $x\in\{0,1\}^\lambda$.
It outputs $H_\lambda(x)$ as the puzzle and $J_\lambda\big(H_\lambda(x),g_\lambda(x)\big)$ as the key.
The verifier $\Ver(z, k)$ only accepts queries for $z \in \{0,1\}^{4\lambda}$ and $k \in \{0,1\}^{8\lambda}$ for some $\lambda \in \mathbb N^+$, and it returns the query result $f_\lambda(z,k)$.

\begin{theorem}[Existence of an EV-OWPuzz]\label{thm:ev-owpuzz-survives-help}
Let $(U,f)$ be the oracle defined in Section~\ref{sec:commitment-oracle}, and $\Help$ be a side oracle satisfying the joint independence as required in Lemma~\ref{lem:help-inaccessible}, then with probability 1 over the choice of $U, f, \Help$, the EV-OWPuzz candidate above is secure against nonuniform QPT adversaries with polynomial-length classical advice and queries to $U,f,\Help$.
\end{theorem}
We can prove an even stronger result: even if $J, H$ are publicly accessible and the adversary has unbounded computation power but only makes polynomially many queries to $U, f$, only queries polynomial width oracles, and has polynomial-size advice, the adversary still cannot break the EV-OWPuzz. In case $J, H$ are public and the adversary has unbounded computation, the EV-OWPuzz can be transformed into $(x, g_n(x))$, which is identical with the EV-OWPuzz in Section~\ref{sec:existence_of_EV-OWPuzz_1}.

Therefore, the proof of Theorem~\ref{thm:ev-owpuzz-survives-help} is almost identical to the proof of Theorem~\ref{thm:ev-owpuzz-existence}, except that the side oracle $\Help$ is statistically dependent on $H, g, R$, and the parameter regime is slightly different. However, according to Lemma~\ref{lem:help-inaccessible}, for any polynomial-query adversary $\mathcal A$, the accessible part of the side oracle $\Help$ is independent of $H, g, R$, hence the original proof of Theorem~\ref{thm:ev-owpuzz-existence} still applies. 

\subsection{Main technique: LOCC decoupling of oracles}
\label{sec:LOCC-decoupling}
In this section, we introduce the decoupling techniques of the QCCC commitment scheme, which serves as the main technical ingredient of the proof of the non-existence of a QCCC commitment. First, we define the decoupled oracles and the corresponding decoupled commitment scheme.
\begin{definition}[Decoupled oracles]\label{def:decoupled-oracle}
    Let $\Theta_m = (H_m, g_m, R_m)$ be the hidden variables of $U, f$ except $J$. For each  $\widehat{\Theta}_m = (\widehat H_m, \widehat g_m, \widehat R_m)$ independent of $\Theta_m$, the corresponding decoupled oracle $\widehat{U}_m, \widehat{f}_m$ is defined as follows:
    \begin{itemize}
        \item The decoupled $\widehat{U}_m$ is defined according to $J_m$ and $\widehat{\Theta}_m$. More formally, $\widehat{U}_m$ is the reflection oracle relative to $\ket{\widehat{\phi}_m^-}$, where $\ket{\widehat{\phi}_m^-}$ is defined as the following state
        \begin{align*}
              \ket{\widehat\phi_m^-}:=\frac{1}{\sqrt{2}}\ket{0^{12m+1}}-\frac{1}{\sqrt{2}}\ket{1}\otimes \ket{\widehat\phi_m}\,.
        \end{align*}
        Here
        \begin{equation*}
    \ket{\widehat\phi_m}
    :=
    2^{-m/2}
    \sum_{x\in\{0,1\}^m}
        (-1)^{\widehat R_m(x)}
        \ket{\widehat H_m(x)}
        \ket{J_m(\widehat H_m(x),\widehat g_m(x))}.
\end{equation*}
        \item The decoupled $\widehat{f}_m$ is defined as the verification oracle according to $J_m$ and $\widehat{\Theta}_m$. More formally,
        \begin{equation*}
    \widehat f_n(z,k)
    :=
    \begin{cases}
        \mathbf 1[k=J_n(z,\widehat g_n(x))],
        & z=\widehat H_n(x)\text{ for some }x\in\{0,1\}^n,\\
        \mathbf 1[k\in\Im(J_n(z,\cdot))],
        & z\notin\Im(\widehat H_n).
    \end{cases}
\end{equation*}
    \end{itemize}
\end{definition}

\begin{definition}[Decoupled LOCC commitment protocols]\label{def:ideal-decoupled-protocol}
    For some arbitrary two-party QCCC commitment protocol $\Pi$ satisfying Definition~\ref{def:QCCC-commitment}, assume each party can only query $U_m$, $f_m$, and some classical oracles independent of $\Theta_m$ and $J_m$. Choose a parameter $a\in \mathbb N_+$.
    Define the \emph{decoupled} protocol $\widetilde \Pi$ as follows.
    It uses an independent $\widehat\Theta_m=(\widehat H_m,\widehat g_m,\widehat R_m)$ as follows:
    \begin{itemize}
    \item The receiver behaves exactly the same as the honest receiver of $\Pi$, and it can query $U_m$, $f_m$ that is dependent on $J_m$ and $\Theta_m$;
    \item The committer behaves the same as the honest committer of $\Pi$, except that it queries the decoupled oracle (Definition~\ref{def:decoupled-oracle}) $\widehat{U}_m$ and $\widehat{f}_m$ that are dependent on $J_m$ and  $\widehat{\Theta}_m$.
\end{itemize}
\end{definition}

Now we formally describe the proposed closeness theorem, which claims that, informally speaking, if $\Theta_m$ and $\widehat\Theta_m$ are independently sampled at random, the distributions of the transcripts are statistically indistinguishable between the original protocol and its decoupled version (Definition~\ref{def:ideal-decoupled-protocol}).

\begin{theorem}\label{thm:decoupling}
    Fix $m\in \mathbb N_+$.
    Let $\Pi$ be a two-party protocol satisfying the syntax in Definition~\ref{def:QCCC-commitment}, such that each party can only query $U_m$, $f_m$, and some classical oracles independent of $\Theta_m$ and $J_m$.
    Suppose each party of $\Pi$ queries the oracles at most $q$ times.
    Let $\rho(b;\Theta_m,J_m)$ be the distribution of the transcript between the honest committer and the receiver in $\Pi$; let $\sigma(b;\widehat\Theta_m,\Theta_m,J_m)$ be the distribution of the transcript of the decoupled protocol $\widetilde \Pi$ (Definition~\ref{def:ideal-decoupled-protocol}), with some independent $\widehat\Theta_m$.
    For each $b\in \{0,1\}$,
\begin{equation*}
  \EE_{J_m}\Biggl[\Bigl\Vert\EE_{\Theta_m}\bigl[\rho(b;\Theta_m,J_m)\bigr]-\EE_{\widehat \Theta_m,\Theta_m}\bigl[\sigma(b;\widehat\Theta_m,\Theta_m,J_m)\bigr]\Bigr\Vert_1\Biggr]\leq O\left(\frac{q^{2/3}}{2^{m/6}}\right).
\end{equation*}
\end{theorem}

\begin{proof}
    We use a hybrid argument.
    Starting from the distribution of the transcript $\rho$ of the true protocol $\Pi$, we use a hybrid method to gradually move toward the distribution of the transcript $\sigma$ of the decoupled protocol $\widehat \Pi$, and we will prove that there are only a small loss for each hybrid step.

\begin{table}[H]
  \centering
  \renewcommand{\arraystretch}{1.25}
  \begin{tabular}{cl}
    \hline
    Hybrid & Transcript distribution \\
    \hline
    $\mathsf{Hyb}_0$ & $\mathcal D_b[U_m,f_m;U_m,f_m]$ \\
    $\mathsf{Hyb}_1$ & $\mathcal D_b[\ket{\phi_m^-}^{\otimes a},f_m;\ket{\phi_m^-}^{\otimes a},f_m]$ \\
    $\mathsf{Hyb}_2$ & $\mathcal D_b[\ket{\phi_m^-(\widetilde R_m)}^{\otimes a},f_m;\ket{\phi_m^-}^{\otimes a},f_m]$ \\
    $\mathsf{Hyb}_3$ & $\mathcal D_b[\Samp_{H_m,g_m},f_m;\Samp_{H_m,g_m},f_m]$ \\
    $\mathsf{Hyb}_4$ & $\mathcal D_b[\Samp_{H_m,g_m},f_J;\Samp_{H_m,g_m},f_J]$ \\
    $\mathsf{Hyb}_5$ & $\mathcal D_b[\Samp_{\widehat H_m,\widehat g_m},f_J;\Samp_{H_m,g_m},f_J]$ \\
    $\mathsf{Hyb}_6$ & $\mathcal D_b[\widehat U_m,\widehat f_m;U_m,f_m]$ \\
    \hline
  \end{tabular}
  \caption{The list of the hybrids. $J_m, H_m, g_m, R_m, \widehat H_m, \widehat g_m, \widehat R_m, \widetilde R_m$ are independently sampled functions. $U_m, f_m, \Samp_{H_m, g_m}$ and $\widehat U_m, \widehat f_m, \Samp_{\widehat{H}_m,\widehat g_m}$ are the oracles corresponding to $J_m, H_m, g_m, R_m$ and $J_m, \widehat H_m, \widehat g_m, \widehat R_m$ respectively. $\ket{\phi_m^-}$ and $\ket{\phi_m^-(\widetilde R_m)}$ are the states corresponding to $J_m, H_m, g_m, R_m$ and $J_m, H_m, g_m, \widetilde R_m$ respectively. $f_J$ refers to the function $(z, k)\to \mathbf 1[k \in \Im(J_n(z, \cdot))]$.}\label{tab:hybrid-commitment-decoupling}
\end{table}

    Set $a=\lfloor q^{2/3}2^{m/3}\rfloor$. Write $\mathcal D_b[U_C,f_C;U_R,f_R]$ for the transcript distribution for the committer and receiver of the following protocol: they act the same as $\Pi$, but the committer's (resp.\ receiver's) queries to $U,f$ are replaced by querying $U_C, f_C$ (resp.\ $U_R, f_R$). As we will need to simulate the queries to $U$ with states or samplers (Theorem~\ref{thm:JLS-simulation} and Lemma~\ref{lem:phase-information-replacement}), given copies of states $\ket{\phi_m^-}^{\otimes a}$ or samples $(H_m(x),J_m(x,g_m(x)))\gets \mathsf{Samp}$, we can simulate the queries to $U$. We abuse the notation and denote $\mathcal D_b[\ket{\phi_C}^{\otimes a}, f_C;\ket{\phi_R}^{\otimes a}, f_R]$ and $\mathcal D_b[\mathsf{Samp}_C, f_C; \mathsf{Samp}_R, f_R]$ as the transcript distribution of the protocol where the committer's and receiver's queries to $U$ are simulated by such simulators. Here, we define the sampling oracle $\Samp_{H,g}$ to be the oracles sampling random $H(x), J(H(x), g(x))$. Set $a = q^{2/3}2^{m/3}.$ We define our hybrids in Table~\ref{tab:hybrid-commitment-decoupling}. 
    
Now, we will prove the hybrids are statistically close to each other step by step.
\begin{enumerate}
  \item \textbf{$\mathsf{Hyb}_0\to\mathsf{Hyb}_1$.}
  $\mathsf{Hyb}_1$ is the same as $\mathsf{Hyb}_0$, except that the queries to $U_m$ of both parties are simulated by the reflection-to-state simulator (Theorem~\ref{thm:JLS-simulation}).
  For both parties, for all queries to $U_m$, they are simulated using in total $a$ copies of $\ket{\phi_m^-}$.
  By Theorem~\ref{thm:JLS-simulation}, regardless of $\Theta_m$ and $J_m$, the accumulating statistical distance cost is at most $4q/\sqrt{a} = O\left(q^{2/3}{2^{-m/6}}\right)$.

  \item \textbf{$\mathsf{Hyb}_1\to\mathsf{Hyb}_2$.}
  $\mathsf{Hyb}_2$ is the same as $\mathsf{Hyb}_1$, except that each time the committer queries $\ket{\phi_m^-}$, it instead uses the following state
  \begin{equation*}
    \frac{1}{\sqrt{2}}\ket{0^{12m+1}}-\frac{1}{\sqrt{2}}\ket{1}\otimes \sum_{x\in \{0,1\}^m}(-1)^{{\widetilde R_m}(x)}\ket{H_m(x)}\ket{J_m\bigl(H_m(x),g_m(x)\bigr)}\,.
  \end{equation*}
  In other words, in this hybrid, only the phases of the states are decoupled.

  To prove $\mathsf{Hyb}_2$ is close to $\mathsf{Hyb}_1$ requires that the two parties have only classical communication.
  Any message in a QCCC protocol must be obtained by a fixed \emph{measurement} that only depends on the previous transcript.
  This means the transcript can be seen as generated by a fixed POVM measurement where each term is \emph{local}\footnote{Here, ``local" simply means each term can be denoted as a tensor product of two positive semidefinite operators.}.
  By Theorem~\ref{thm:LOCC-decoupling}, the statistical distance loss of this replacement is at most $O(a/\sqrt{2^m})=O(q^{2/3}2^{-m/6})$.

  \item \textbf{$\mathsf{Hyb}_2\to\mathsf{Hyb}_3$.}
  $\mathsf{Hyb}_3$ is the same as $\mathsf{Hyb}_2$, except that we apply the state-to-sample simulator in Theorem~\ref{thm:GMMY-simulation} to both parties. Since the phases of the states are decoupled in $\mathsf{Hyb}_2$, we can apply Theorem~\ref{thm:GMMY-simulation} sequentially.
  Therefore, the simulation will contribute a $O(a/\sqrt{2^m})=O(q^{2/3}2^{-m/6})$ loss in the statistical distance.

  \item \textbf{$\mathsf{Hyb}_3\to\mathsf{Hyb}_4$.}
$\mathsf{Hyb}_4$ is the same as $\mathsf{Hyb}_3$, except that each query to $f_m(x,y)$ made by either party is replaced by a query to $f_J$, which checks whether $y\in\operatorname{Im}J_m(x,\cdot)$.

If $x\notin\operatorname{Im}(H_m)$, this replacement has no effect. If $x=H_m(i)$ for some $i\in\{0,1\}^m$, we bound its effect as follows. Crucially, neither party in $\Pi$ has direct access to $J_m$. For this row of $J_m$, the sampler reveals only the key $k_i:=J_m\bigl(H_m(i),g_m(i)\bigr)$, at which $f_m$ and $f_J$ agree. Conditioned on $k_i$, replacing $f_m$ by $f_J$ therefore amounts to changing from $0$ to $1$ a uniformly random set of $2^{4m}-1$ entries in $\{0,1\}^{8m}\setminus\{k_i\}$. The BBBV theorem~\ref{BBBV} then bounds the average loss over the choice of $J_m$ by $O(a/\sqrt{2^{4m}}) = O(q^{2/3}2^{-3m/5})$.

  \item \textbf{$\mathsf{Hyb}_4\to\mathsf{Hyb}_5$.}
  Now that the verification oracle is completely independent of $\Theta_m$.
  Fix $J_m$.
  $\mathsf{Hyb}_5$ is the same as $\mathsf{Hyb}_4$, except that for each sampling oracle queried by the committer, it is replaced by another sampling oracle that samples according to $\widehat \Theta_m$ instead of $\Theta_m$.
  
  Consider the samples generated by the committer's and the receiver's sampling oracle. These samples can be seen as independently sampled in the whole space $\{0,1\}^{4m}\times \{0,1\}^{4m}$, except when the first $4m$ bits collide on any pair of the samples, either from the committer's side or from the receiver's side.
  Thus, the total statistical distance loss is at most $O(a^2/2^{m})=O(q^{4/3}2^{-m/3})$.

  \item \textbf{$\mathsf{Hyb}_5\to\mathsf{Hyb}_6$.}
  Note that in $\mathsf{Hyb}_5$, the committer's behavior is completely independent of $\Theta_m$.
  $\mathsf{Hyb}_6$ is the same as $\mathsf{Hyb}_5$, except that the receiver's oracle is restored to the original version, and the committer's oracle is restored to the decoupled version.
  The same argument used in $\mathsf{Hyb}_0$ vs.\ $\mathsf{Hyb}_4$ yields an average loss of $O(q^{2/3}2^{m/6})$ when $q$ is polynomial in $m$.

\end{enumerate}

    Note that $\mathsf{Hyb}_6$ is exactly the decoupled protocol $\widetilde\Pi$.
    Thus, our hybrid argument gives a $O(q^{2/3}2^{-m/6})$ statistical distance bound on $\rho$ and $\sigma$, when $q$ is polynomial in $m$.
\end{proof}

\subsection{The construction of the helper oracle}\label{sec:diagonalization-oracle}

We now formally describe the construction of the oracle $\mathsf{Help}$. As we discussed in the technical overview section~\ref{sec:tech-overview-commitment}, this oracle will be constructed using diagonalization, which we define as follows.

Consider an effective enumeration $\Pi_1,\Pi_2,\ldots$ of all uniform efficient two-party two-phase protocols satisfying the syntax of a QCCC commitment (Definition~\ref{def:QCCC-commitment}).
We may assume, without loss of generality, that each party in $\Pi_i(1^\lambda)$ runs in time at most $\lambda^i$.

Choose an efficiently computable function $\iota$, such that $\iota(s)\leq s$ for each $s$, and each positive integer occurs infinitely often in $\{\iota(s)\}_{s\in \mathbb N_+}$.\footnote{For example, we may set $\iota(s)$ to be the largest integer $r$ such that $2^{r-1}$ is a factor of $s$.}
Starting from $C_0=2$, we inductively define
\begin{equation*}
  \Lambda_s:=C_{s-1}^{s}+1,
  \qquad
  C_s:=\Lambda_s^{2\iota(s)+2}.
\end{equation*}

We now start the definition of $\mathsf{Help}$.

\begin{definition}[The helper oracle $\mathsf{Help}$]
    Given $J=(J_n)_{n\in \mathbb N_+}$ and $\widehat \Theta=(\widehat\Theta_n)_{n\in \mathbb N_+}$, we define the helper oracle $\mathsf{Help}$ as the set of tuples $$\{\mathsf{Help}_s:=(\mathsf{Help}^{\widehat \Theta}_s,\mathsf{Help}^J_s,\mathsf{Help}^{\text{hide}}_s,\mathsf{Help}^{\text{bind}}_s)\}_{s\in \mathbb N_+}\,,$$ where each function in $\mathsf{Help}_s$ is defined on $\{0,1\}^{64C_s}\to \{0,1\}$, and we inductively define $\mathsf{Help}_s$ as follows.
    \begin{itemize}
        \item For each $m$ such that $C_{s-1}<m\leq C_s$, $\mathsf{Help}_s^{\widehat \Theta}(m,\cdot)$ provides access to $\widehat\Theta_m=(\widehat H_m,\widehat g_m,\widehat R_m)$.
        \item For each $m$ such that $C_{s-1}<m\leq C_s$, $\mathsf{Help}^J_s(m,\cdot)$ provides access to $J_m$ and $\widehat f_m$.
        \item For every $z\in \{0,1\}^{\leq C_s}$, $\mathsf{Help}^{\text{hide}}_s(z,00\cdots0)$ equals the maximal likelihood estimation of the committed bit $b$ when the commit stage transcript of $\widehat \Pi$ is $z$. 
        \item For every $t\in \{0,1\}^{\leq C_s}$, $\mathsf{Help}^{\text{bind}}_s(t,00\cdots0)$ equals the optimal next bit to send (in the opening phase) if the committer wants to open $b=1$ and the previous transcript is $t$.
        \item All function values not mentioned above is set to $0$.
    \end{itemize}
    Here, $\widehat\Pi$ is the protocol where each party behaves the same as $\Pi_{\iota(s)}(1^{\Lambda_s})$\footnote{Note that $\Pi_{\iota(s)}$ may query $\mathsf{Help}$ syntactically. However, since $C_s>\Lambda_s^{\iota(s)}$, it can only query $\mathsf{Help}_{\leq s-1}$.}, except that
    \begin{itemize}
        \item Before the commit stage, for each $m$ such that $C_{s-1}<m\leq C_s$, the committer and the receiver independently samples $\Theta_{m,C}$ and $\Theta_{m,R}$ respectively. Denote $U_{m,C},f_{m,C}$ the oracles induced by $\Theta_{m,C}$ and $J_m$; denote $U_{m,R},f_{m,R}$ the oracles induced by $\Theta_{m,R}$ and $J_m$.
        \item Each time the committer (resp.\ the receiver) needs to query $U_m,f_m$ for some $C_{s-1}<m\leq C_s$, they query $U_{m,C},f_{m,C}$ (resp.\ $U_{m,R},f_{m,R}$).\footnote{Note that each time they need to query $U_m,f_m$ for $m\leq C_{s-1}$, they simply query the (true) oracles.}
    \end{itemize}
\end{definition}

With the definition of $\mathsf{Help}$, we can now prove that no polynomial-time adversary may access $\Theta_m$ through querying $\mathsf{Help}$.

\begin{proof}[Proof of Lemma~\ref{lem:help-inaccessible}.]
For every $n\in \mathbb N_+$, let $s$ be the unique index such that $C_{s-1}<n\leq C_s$. By construction, $\Help_{\le s}$ is independent of $\Theta_n$.
Note that $\Help_{s+1}$ has query length $64C_{s+1}$, and
\begin{equation*}
  C_{s+1}
  \geq
  \Lambda_{s+1}
  =
  C_s^{s+1}+1
  >
  n^{s+1}\,.
\end{equation*}
Since $s\to\infty$, we have that for every fixed polynomial $p$, $p(n)<C_{s+1}$. Thus, $\mathsf{Help}|_{\{0,1\}^{\leq p(n)}}$ is independent of $\Theta_n$ when $n$ is sufficiently large.
\end{proof}

\begin{remark}\label{rmk:remark-o(1)}
For each $s\in \mathbb N_+$ and each $C_{s-1}< m\leq C_s$, we instantiate the oracle calls to $\widehat U_m$ by using the reflection-to-state simulation using $a=\Lambda_s^{2\iota(s)}m^4$ copies of the state and eventually the sampler (Theorem~\ref{thm:GMMY-simulation}).
Therefore, we need to prepare at most $C_s(\Lambda_s^{2\iota(s)}C_s^4)=\Lambda_s^{O(\iota(s))}$ copies of the state, which is a polynomial.
Moreover, denote the statistical distance introduced as $\delta_m$.
Then,
$$\EE_{J}\sum_{m=C_{s-1}+1}^{C_s}\delta_m\leq O\left(\frac{1}{C_{s-1}}+\frac{C_s^5}{\sqrt{2^{C_{s-1}}}}\right)\,.$$
Recall that $C_{s-1}^s<C_s<C_{s-1}^{3s^2}$, which means $C_s^{O(1)}\ll 2^{C_{s-1}}$.
Thus, the total statistical distance introduced by the decoupling is at most $o(1)$, after averaging over $J$.
\end{remark}

\subsection{Non-existence of QCCC commitment scheme}\label{sec:commitment-nonexistence}

\begin{theorem}
    Under oracle $(U,f,\mathsf{Help})$, for any protocol $\Pi^{U, f, \Help}$, with probability 1 over the choice of oracles, at least on of the following property holds:
    \begin{itemize}
        \item $\Pi$ is not 0.05-correct.
        \item There exists a quantum adversary running in polynomial time and breaks the 0.05-hiding security of $\Pi^{U,f,\Help}$.
        \item There exists a quantum adversary running in polynomial time and breaks the 0.05-sum binding security of $\Pi^{U,f,\Help}$.
    \end{itemize}
\end{theorem}
\begin{proof}
    As our convention of enumerating commitment schemes in Section~\ref{sec:diagonalization-oracle}, we can assume the commitment scheme $\Pi_i$ can be described by a two-party Turing machine that runs in time $O(n^i)$. For any commitment scheme $\Pi_i$ running in time $O(n^i)$, let $\widehat \Pi_i$ be the decoupled commitment scheme corresponding to $\Pi_i$. Let $L_i = \{\Lambda_s:\iota(s)=i\}$ be the set of infinite security parameters corresponding to $\Pi_i$.

    Now we describe the adversary $\calA$ to break the security of $\Pi_i$ as follows. To break the hiding property, the (malicious) receiver behaves as the honest receiver in the commit phase, and use the optimal distinguisher $\Help^{\mathsf{hide}}$ to dinstinguish the transcript and output the prediction outcome $b$. To break the binding property, the (malicious) committer behaves as the honest committer to commit $b=0$ in $\widehat \Pi_i$ with queries $\widehat U$ and $\widehat f$ simulated by queries to $\Help_s^{\widehat \Theta}$ and $\Help_s^J$ (as described in Remark~\ref{rmk:remark-o(1)}). In case that the binding experiment requires to open $b=0$, it behaves as the honest adversary in $\widehat \Pi_i$. In case that the binding experiment requires to open $b=1$, it queries $\Help_s^{\mathsf{bind}}$ to answer the receiver to cheat.

    Assume the protocol $\Pi_i$ is $\eta_b$-correct, $\epsilon$-statistically hiding, and $\delta$-statistically binding.
    Assume the protocol $\widehat\Pi_i$ is $\widehat\eta_b$-correct, $\widehat\epsilon$-statistically hiding, and $\widehat\delta$ statistically binding\footnote{We assume that these correctness and security parameters are tight, acheivable with some statistical adversaries.}, then according to Lemma~\ref{lem:dichotomy-hiding-binding}
    $$\widehat\eta_0+\widehat\eta_1+\widehat\epsilon+\widehat\delta \geq 1\,.$$
    Moreover, by the definition of $\Help$ and the fact that $2C_s\leq \Lambda_s^{3\iota(s)}$ for sufficiently large $\Lambda_s \in L_i$, the receiver's statistical hiding advantage $\epsilon$ and the malicious committer's statistical binding advantage $\delta$ can both be achieved by a QPT algorithm that queries $\Help$, after averaging over $\widehat\Theta$. 
    \begin{align*}
        \widehat\eta_0 + \widehat\eta_1 + \EE_{\widehat \Theta}[\widehat\epsilon + \widehat\delta] \geq 1.
    \end{align*}

    Next, introduce the following transcript distributions:
    \begin{itemize}
        \item Let $\rho(b;\Theta,J)$ be the distribution of the transcript of $\Pi_{i}(1^{\Lambda_s})$.
        \item Let $\sigma(b;J)$ be the distribution of the transcript of $\widehat\Pi_i(1^{\Lambda_s})$.
    \end{itemize}
     Then, by Theorem~\ref{thm:decoupling} and Remark~\ref{rmk:remark-o(1)},
\begin{equation*}
  \EE_J\Biggl[\Bigl\Vert\EE_{\Theta}\bigl[\rho(b;\Theta,J)\bigr]-\sigma(b;J)\Bigr\Vert_1\Biggr]\leq o(1)\,.
\end{equation*}
    Thus,
\begin{equation*}
  \Pr_J\Biggl[\Bigl\Vert\EE_{\Theta}\bigl[\rho(b;\Theta,J)\bigr]-\sigma(b;J)\Bigr\Vert_1\leq 0.01\Biggr]\geq 1-o(1)\,.
\end{equation*}
    Let $\mathsf{Adv}_{\mathsf{hide}}$ and $\mathsf{Adv}_{\mathsf{bind}}$ be the advantage the adversary $\mathcal A$ can achieve for the commitment scheme $\Pi_i(1^{\Lambda_s})$. The previous equation immediately yields
    $$\Pr_J\left[\EE_{\Theta} \mathsf{Adv}_{\mathsf{hide}} \geq \EE_{\Theta, \widehat \Theta} \widehat\epsilon - 0.01 \right] \geq 1 - o(1)\,.$$
    Also, the receiver's strategy does not change between $\Pi_{i}(1^{\Lambda_s})$ and $\widehat\Pi_i(1^{\Lambda_s})$.
    Note that the final output can also be considered as a part of the transcript, and since the transcript distributions are close even if considering the output of the honest receiver,
    $$\Pr_{J}\left[\EE_{\Theta}\eta_b\geq \widehat \eta_b-0.01\right]\geq 1-o(1)\,.$$
    For the binding advantage, consider the behavior of the malicious committer constructed in Lemma~\ref{lem:dichotomy-hiding-binding} in the sum-binding game of $\Pi_{i}(1^{\Lambda_s})$.
    Since the receiver's strategies are exactly the same in $\Pi_{i}(1^{\Lambda_s})$ and $\widehat\Pi$ except that in $\widehat\Pi_i(1^\Lambda_s)$ the $\widehat\Theta$ is self-sampled, we have that for every $J,\widehat\Theta$,
    $$\EE_{\Theta}\mathsf{Adv}_{\mathsf{bind}}=\delta\,.$$
    Thus, for any sufficiently large $s$,
\begin{equation*}
  \Pr_{J}\Bigl[\EE_{\Theta,\widehat\Theta}\bigl[\eta_0+\eta_1+\mathsf{Adv}_{\mathsf{hide}}+\mathsf{Adv}_{\mathsf{bind}}\bigr]\geq 0.97\Bigr]\geq 0.99\,.
\end{equation*}

    Then according to the Markov's inequality, 
    $$\Pr_{J,\Theta,\widehat\Theta}[\eta_0+\eta_1+\mathsf{Adv}_{\mathsf{hide}} + \mathsf{Adv}_{\mathsf{bind}}\geq 0.2]\geq 0.2\,,$$
    Moreover, since the intervals $(C_{m-1},C_m]$ are disjoint, with probability $1$ over $(U,f,\Help)$, as long as $\Pi_i(1^\lambda)$ is 0.05-correct, there exists an adversary that achieves 0.05 advantage infinitely often on either binding or hiding experiments.
    This attack is uniform: the adversary simply performs queries to $\Help$ according to some efficiently computable padding rule.

    Since the number of all protocols are countable, we obtain that with probability $1$ over $(U,f,\Help)$, there does not exist a bit commitment scheme under the oracle $(U,f,\Help)$.
\end{proof}

\bibliography{main}
\appendix
\section{Useful Theorems to Eliminate the Phases}\label{sec:phase-concentration}
\newcommand{\type}{\operatorname{type}}

Here we show how to simulate the random phase states with a sampler.

Theorem~\ref{thm:phase-concentration} is directly related to the following observation: superpositions of states with a random phase can be simulated using samples of the components.
To our knowledge, this observation first appeared in \cite[Lemma 5.5]{Zhandry24a}.
However, Theorem~\ref{thm:phase-concentration} mainly concludes that this simulation is $L^2$-concentrated with respect to any POVM measurement, which is new, and may be of independent interest.

Our proof method is to expand in the type-vector basis, which mainly follows \cite[Theorem 5.4]{AG25}, and also its subsequent work \cite[Lemma 4.3]{CountCrypt}.

Due to the construction of the reflection-to-state simulator (Theorem~\ref{thm:JLS-simulation}), Theorem~\ref{thm:phase-concentration} also considers the situation where there is a significant magnitude on one of the chosen basis.

\begin{theorem}[Phase concentration]\label{thm:phase-concentration}
    Let $N\in \mathbb N_+$.
    Given a uniformly random function $f:[N]\to\{0,1\}$ and any set of pairwise orthogonal states $\{\ket{\phi_k}\}_{k\in \{0\}\cup [N]}$, define the corresponding phase state to be
    $$\ket{\phi^f}:=\sqrt{p_0}\ket{\phi_0}+\sum_{k\in [N]} (-1)^{f(k)}\sqrt{p_k}\ket{\phi_k}\,,$$
    where $p_0,p_1,\cdots,p_N$ are fixed nonnegative constants that add up to $1$.

    For any quantum circuit $\mathcal A$ that uses $t$ copies of $\ket{\phi^f}$ as input and outputs a classical bit, for any Hermitian $M$ independent of $f$ such that $0\preceq M\preceq I$,\footnote{
    Recall our convention that $[N]=\{1,2,\cdots, N\}$ does not contain $0$.
    Thus, $\ket{\phi^f}$ can have a large magnitude on $\ket{\phi_0}$ and still a small variance.
    }
        $$\Var\left(\bra{\phi^f}^{\otimes t}M\ket{\phi^f}^{\otimes t}\right)\leq 4t\max_{k\in [N]}p_k\,.$$
    Moreover, for any POVM measurement $\{M_i\}$, denote the output distribution of this measurement acting on $\ket{\phi^f}^{\otimes t}$ as $\mathcal D$.
    Then,
    $$\EE \big[||\mathcal D-\EE \mathcal D||_2^2\big]\leq 4t\max_{k\in [N]}p_k\,.$$
\end{theorem}

\begin{proof}
    With abuse of notation we identify $f$ with its truth table vector in $\{0,1\}^{N}$.
    
    For $\vec k\in (\{0\}\cup [N])^t$, define the type of $\vec k$, denoted $\type(\vec{k})$, to be as follows: $T=\type(\vec{k})$ is a vector in $\mathbb N^{N+1}$, and for each $i\in \{0\}\cup [N]$, the $i$-th entry $T_i$ equals the number of entries in $\vec k$ that equals $i$.
    Let $T_+$ be the vector in $\mathbb N^N$ that consists of all entries of $T$ but $T_0$.
    Define the set of all types of a vector in $(\{0\}\cup [N])^t$ as $\mathcal T_{N,t}$.
    Then
    \begin{align*}
        \ket{\phi^f}^{\otimes t}={}& \sum_{k_1,k_2,\cdots,k_t}\sqrt{p_{k_1}p_{k_2}\cdots p_{k_t}}(-1)^{\sum_{j=1}^t f(k_j)}\bigotimes_{j=1}^t\ket{\phi_{k_j}}\\
        ={}&\sum_{T\in \mathcal T_{N,t}}\sum_{\type(\vec k)=T}\sqrt{p_0^{T_0}p_1^{T_1}\cdots p_N^{T_N}}(-1)^{\sum_{j=1}^t f(k_j)}\bigotimes_{j=1}^t\ket{\phi_{k_j}}\\
        ={}&\sum_{T\in \mathcal T_{N,t}}(-1)^{f\cdot T_+}\sqrt{p_0^{T_0}p_1^{T_1}\cdots p_N^{T_N}}\sum_{\type(\vec k)=T}\bigotimes_{j=1}^t\ket{\phi_{k_j}}\,.
    \end{align*}
    Since $(-1)^{f\cdot T_+}$ is only determined by the parity of each entry in $T_+$, we group them by the parity: 
    $$\ket{\phi^f}^{\otimes t}=\sum_{\nu\in \{0,1\}^N}(-1)^{f\cdot \nu}\sum_{\substack{T\in \mathcal T_{N,t}\\ T_+\bmod 2=\nu}}\sqrt{p_0^{T_0}p_1^{T_1}\cdots p_N^{T_N}}\sum_{\type(\vec k)=T}\bigotimes_{j=1}^t\ket{\phi_{k_j}}\,.$$
    
    For any $\nu\in \{0,1\}^N$, define
    $$\mu_\nu=\sum_{\substack{T\in \mathcal T_{N,t}\\ T_+\bmod 2=\nu}}p_0^{T_0}p_1^{T_1}\cdots p_N^{T_N}\sum_{\type(\vec k)=T}1\,.$$ 
    Clearly $\mu_\nu$ is the probability of the following: sampling a vector $\vec k\in (\{0\}\cup [N])^t$ with each entry getting $i$ with independent probability $p_i$, and $\type(\vec k)_+\bmod 2=\nu$.
    If $\mu_\nu\neq 0$, we define
    $$\ket{\phi_\nu}=\frac{1}{\sqrt{\mu_\nu}}\sum_{\substack{T\in \mathcal T_{N,t}\\ T_+\bmod 2=\nu}}\sqrt{p_0^{T_0}p_1^{T_1}\cdots p_N^{T_N}}\sum_{\type(\vec k)=T}\bigotimes_{j=1}^t\ket{\phi_{k_j}}\,.$$
    One can verify that $\ket{\phi_\nu}$ is a well-defined quantum state, and for different $\nu,\nu'$, $\ket{\phi_\nu}$ and $\ket{\phi_{\nu'}}$ are orthogonal.
    
    We can write
    $$\ket{\phi^f}^{\otimes t}=\sum_{\nu} (-1)^{f\cdot \nu}\sqrt{\mu_\nu} \ket{\phi_\nu}\,.$$

    Now, consider any Hermitian $M$ such that $0\preceq M\preceq I$.
    Then
    \begin{align*}
        \bra{\phi^f}^{\otimes t}M\ket{\phi^f}^{\otimes t}={} & \sum_{\nu,\nu'} (-1)^{f\cdot (\nu\oplus \nu')}\sqrt{\mu_\nu\mu_{\nu'}} \braket{\phi_{\nu'}|M|\phi_\nu}\,.
    \end{align*}

    Taking the expectation of $f$, each summand is nonzero only if $\nu\oplus \nu'= 0$.
    Therefore,
    \begin{align*}
        &\Var\left(\bra{\phi^f}^{\otimes t}M\ket{\phi^f}^{\otimes t}\right)\\
        ={}& \EE_f \left[\left|\bra{\phi^f}^{\otimes t}M\ket{\phi^f}^{\otimes t}-\EE_f [\bra{\phi^f}^{\otimes t}M\ket{\phi^f}^{\otimes t}]\right|^2\right]\\
        ={}& \EE_f\left[\left|\sum_{\nu\oplus \nu'\neq 0}(-1)^{f\cdot(\nu\oplus \nu')}\sqrt{\mu_\nu\mu_{\nu'}}\braket{\phi_{\nu'}|M|\phi_\nu}\right|^2\right]\\
        ={}& \EE_f\left[\sum_{\xi\oplus \xi'\neq 0,\nu\oplus \nu'\neq 0}(-1)^{f\cdot (\xi\oplus \xi'\oplus \nu\oplus \nu')} \sqrt{\mu_\xi\mu_{\xi'}\mu_\nu\mu_{\nu'}}\overline{\braket{\phi_{\xi'}|M|\phi_\xi}}\braket{\phi_{\nu'}|M|\phi_\nu}\right]\\
        ={}& \sum_{\alpha\neq 0}\sum_{\substack{\xi\oplus \xi'=\alpha\\\nu\oplus \nu'=\alpha}}\sqrt{\mu_\xi\mu_{\xi'}\mu_\nu\mu_{\nu'}}\overline{\braket{\phi_{\xi'}|M|\phi_\xi}}\braket{\phi_{\nu'}|M|\phi_\nu}\\
        ={}& \sum_{\alpha\neq 0}\left|\sum_{\nu\oplus \nu'=\alpha}\sqrt{\mu_\nu\mu_{\nu'}}\braket{\phi_{\nu'}|M|\phi_\nu}\right|^2\,.
    \end{align*}

    Now, define
    $$M_{\nu\nu'}=\braket{\phi_{\nu'}|M|\phi_\nu}\,.$$

    The rest of this proof is to bound $$\Var\left(\bra{\phi^f}^{\otimes t}M\ket{\phi^f}^{\otimes t}\right)=\sum_{\alpha\neq 0}\left|\sum_{\nu\oplus \nu'=\alpha}\sqrt{\mu_\nu\mu_{\nu'}}M_{\nu\nu'}\right|^2\,.$$

    Note that $\alpha$ is a vector in $\{0,1\}^N$.
    For each $\alpha\neq 0$, choose a fixed $i=i_\alpha$ such that $\alpha_i\neq 0$.
    Define
    $$B_\alpha=\sum_{\substack{\nu\oplus \nu'=\alpha\\ \nu_i=1,\nu_i'=0}}\sqrt{\mu_\nu\mu_{\nu'}}M_{\nu\nu'}\,.$$
    Then
    $$\sum_{\alpha\neq 0}\left|\sum_{\nu\oplus \nu'=\alpha}\sqrt{\mu_\nu\mu_{\nu'}}M_{\nu\nu'}\right|^2=\sum_{\alpha\neq 0}|B_\alpha+\overline{B_\alpha}|^2\leq 4\sum_{\alpha\neq 0}|B_\alpha|^2\,.$$
    By Cauchy-Schwarz,
    $$|B_\alpha|^2\leq \left(\sum_{\substack{\nu\oplus \nu'=\alpha\\ \nu_i=1,\nu_i'=0}}\mu_\nu\right)\left(\sum_{\substack{\nu\oplus \nu'=\alpha\\ \nu_i=1,\nu_i'=0}}\mu_{\nu'}|M_{\nu\nu'}|^2\right)\,.$$
    
    Since $\nu'$ is determined by $\nu$ following this equation $\nu\oplus \nu'=\alpha$, we have
    $$\sum_{\substack{\nu\oplus \nu'=\alpha\\ \nu_i=1,\nu_i'=0}}\mu_\nu\leq \sum_{\nu:\nu_i=1}\mu_\nu\,.$$
    The right hand side of the inequality above is the probability that the number $i\in [N]$ occurs odd times in a vector $\vec k\in (\{0\}\cup[N])^t$.
    This means $i$ occurs at least once, hence
    $$\sum_{\nu:\nu_i=1}\mu_\nu\leq 1-\left(1-p_i\right)^t\leq tp_i\leq t\max_{k\in [N]}p_k\,.$$

    Since $M$ is an Hermitian matrix satisfying $0\preceq M\preceq I$, we have $M^2\preceq M$, which means
    $$M_{\nu\nu}\geq (M^2)_{\nu\nu}$$
    By the pairwise orthogonality of $\ket{\phi_\nu}$,
    $$(M^2)_{\nu\nu}\geq \sum_{\nu'}M_{\nu\nu'}M_{\nu'\nu}=\sum_{\nu'}|M_{\nu\nu'}|^2$$
    Therefore, for all $\nu\in \mathcal \nu_{N,t}$,
    $$\sum_{\nu'} |M_{\nu\nu'}|^2\leq  M_{\nu\nu}\,.$$
    
    Hence,
    \begin{align*}
        \Var\left(\bra{\phi^f}^{\otimes t}M\ket{\phi^f}^{\otimes t}\right)
        \leq{}& 4\sum_{\alpha\neq 0}\left(\sum_{\substack{\nu\oplus \nu'=\alpha\\ \nu_i=1,\nu_i'=0}}\mu_\nu\right)\left(\sum_{\substack{\nu\oplus \nu'=\alpha\\ \nu_i=1,\nu_i'=0}}\mu_{\nu'}|M_{\nu\nu'}|^2\right)\\
        \leq{}& \left(4t\max_{k\in [N]} p_k\right)\sum_{\alpha\neq 0}\sum_{\nu\oplus \nu'=\alpha}\mu_{\nu'}|M_{\nu\nu'}|^2\\
        \leq{}& \left(4t\max_{k\in [N]} p_k\right)\sum_{\nu\neq \nu'}\mu_{\nu'}|M_{\nu\nu'}|^2\\
        \leq{}& \left(4t\max_{k\in [N]} p_k\right)\sum_{\nu'}\mu_{\nu'}\sum_{\nu:\nu\neq \nu'}|M_{\nu\nu'}|^2\\
        \leq{}&\left(4t\max_{k\in [N]} p_k\right)\sum_{\nu}\mu_\nu M_{\nu\nu}
    \end{align*}

    Because $\sum_\nu \mu_\nu M_{\nu\nu}$ is nothing more than the weighted average of $M_{\nu\nu}\leq 1$, we conclude
    $$\Var\left(\bra{\phi^f}^{\otimes t}M\ket{\phi^f}^{\otimes t}\right)\leq 4t\max_{k\in [N]} p_k\,.$$

    Now, for any POVM measurement $\{M_i\}$, observe that
    $$\EE [||\mathcal D-\EE \mathcal D||_2^2]=\sum_{i}\Var\left(\bra{\phi^f}^{\otimes t}M_i\ket{\phi^f}^{\otimes t}\right)\leq \left(4t\max_{k\in [N]} p_k\right)\sum_i\sum_{\nu}\mu_\nu(M_i)_{\nu\nu}\,,$$
    we have $\sum_i M_i=I$, so $\sum_i (M_i)_{\nu\nu}=1$, hence
    $$\EE [||\mathcal D-\EE \mathcal D||_2^2]\leq 4t\max_{k\in [N]} p_k\,.$$
\end{proof}

Next, we discuss the LOCC decoupling of the phases which we consider in Section~\ref{sec:commit-separation}.
Intuitively, for any LOCC experiment (i.e.\ QCCC protocol), classical messages in it must be obtained by a fixed local measurement of the state in each party that may depend on previous messages.
But we have already seen by Theorem~\ref{thm:phase-concentration} that a fixed measurement cannot help much in extracting the phases.
Therefore, the outcome distribution of this LOCC experiment should not be detectable between giving two parties the same phases and giving different phases.

The proof we employ in Theorem~\ref{thm:LOCC-decoupling} consists of two parts. First, we replace the $\pm1$ phases with continuous phases\footnote{Note that the average of the continuous phase states is equal to the mixture of type states, as observed in~\cite{AG25}.}, and second, we use a partial transpose method to derive a bound from the collision probability of the type vectors.
The former part is well-known in the literature.
The idea of the latter part appeared in~\cite{AGL24}, the subsequent works~\cite{Goldin-Zhandry,AGL25}, and also in~\cite{Har23}.
However, their results mainly apply to the LOCC indistinguishability of Haar random states, not to the (non-uniform) mixture of type states.
Combining with the pseudorandomness of (binary) phase states~\cite{JLS,BS19}, our proof simplifies the proof of LOCC indistinguishability of Haar random states in~\cite{AGL24} by a direct combinatorial argument.%
\footnote{Also, note that the transcript of an LOCC experiment corresponds to the outcome of a local POVM measurement, but not the other way: a local POVM measurement may not be implementable by any LOCC experiment; see~\cite{nonlocality-without-entanglement}.}

\begin{theorem}\label{thm:LOCC-decoupling}
    Consider the state $\ket{\phi^f}$ defined in Theorem~\ref{thm:phase-concentration}.
    For any fixed local POVM measurement $M=\{M_i\otimes M'_i\}_{i\in \mathcal I}$, denote
    \begin{itemize}
        \item $\rho$ to be the output distribution when the input state is $$\EE_f \left[\ketbra{\phi^f}{\phi^f}^{\otimes t}\otimes \ketbra{\phi^f}{\phi^f}^{\otimes t}\right]\,,$$
        \item $\bar\rho$ to be the output distribution when the input state is $$\EE_{\hat f,f} \left[\ketbra{\phi^{\hat f}}{\phi^{\hat f}}^{\otimes t}\otimes \ketbra{\phi^f}{\phi^f}^{\otimes t}\right]\,.$$
    \end{itemize}
    Then $$\left\Vert \rho-\bar\rho\right\Vert_1\leq (2+8\sqrt{2})t\sqrt{\sum_{k\in [N]}p_k^2}+2t^2\sum_{k\in [N]}p_k^2\,.$$
    Nevertheless,
    $$\left\Vert \rho-\bar\rho\right\Vert_1\leq O\left(t\sqrt{\sum_{k\in [N]}p_k^2}\right)\,.$$
\end{theorem}

\begin{proof}
    The proof can be divided into two parts.
    In the first part, we prove that $t$ copies of a random $(\pm 1)$-signed state is close to the mixtures of its type state, which is defined below.

    For any type $T$ over $[N]\cup \{0\}$, define
    $$|T|=\sum_{\vec k: \type(\vec k)=T}1\,.$$
    Define the distribution $\mathcal D_t$ over types as
    $$\Pr_{\mathcal D_t}[T]=\mu_T:=p_0^{T_0}p_1^{T_1}\cdots p_N^{T_N}|T|\,.$$
    Define the type states
    $$\ket{\phi_T}=\frac{1}{\sqrt{|T|}}\sum_{\vec k:\type(\vec k)=T}\bigotimes_{j=1}^t \ket{\phi_{k_j}}\,.$$
    We have
    $$\ket{\phi^f}^{\otimes t}=\sum_{T}(-1)^{f\cdot T_+}\sqrt{\mu_T}\ket{\phi_T}\,.$$

    It is easy to compute that
    $$\EE_f\ketbra{\phi^f}{\phi^f}^{\otimes t}-\EE_{T\gets \mathcal D_t}\ketbra{\phi_T}{\phi_T}=\sum_{\substack{T\neq T'\\ (T_+-T'_+)\equiv 0\pmod 2}}\sqrt{\mu_T\mu_{T'}}\ketbra{\phi_{T'}}{\phi_T}\,.$$

     For the condition $$T\neq T',\quad (T_+-T_+')\equiv 0\pmod 2$$ to be satisfied, there must exist $k\neq 0$ such that $T_k-T_k'\geq 2$ or $T_k-T_{k}'\leq -2$.
    This means either $T_k\geq 2$ or $T_k'\geq 2$.
    Now, define
    $$\Pi=\sum_{T:\forall k\neq 0,T_k\leq 1}\ketbra{\phi_T}{\phi_T}\,.$$

    $\Pi$ is a projector, and $$\Pi\left(\EE_f\ketbra{\phi^f}{\phi^f}^{\otimes t}-\EE_{T\gets \mathcal D_t}\ketbra{\phi_T}{\phi_T}\right)\Pi=0\,.$$
    For convenience, we define
    $$\varsigma=\EE_f\ketbra{\phi^f}{\phi^f}^{\otimes t}\quad \sigma=\EE_{T\gets \mathcal D_t}\ketbra{\phi_T}{\phi_T}\,.$$
    Since $\Pi(\varsigma-\sigma)\Pi=0$, we define the normalized state
    $$\sigma':=\frac{\Pi\varsigma\Pi}{\Tr(\Pi\varsigma)}=\frac{\Pi\sigma\Pi}{\Tr(\Pi\sigma)}\,.$$
    By the gentle measurement lemma~\ref{lem:gentle-measurement}, we have
    $$\left\Vert\varsigma-\sigma\right\Vert_1\leq \left\Vert\varsigma-\sigma'\right\Vert_1+\left\Vert\sigma-\sigma'\right\Vert_1\leq 2\sqrt{\epsilon_\varsigma}+2\sqrt{\epsilon_\sigma}\,,$$
    where $\epsilon_\varsigma,\epsilon_\sigma$ are the possibility of getting the reject outcome when measuring $\Pi$ on $\varsigma,\sigma$, respectively.
    Since $\Pi\varsigma\Pi=\Pi\sigma\Pi$, we have
    $$\epsilon_\varsigma=\epsilon_\sigma=\Tr((1-\Pi)\sigma)=\EE_{T\gets \mathcal D_t}[\mathbf{1}_{\exists k\neq 0,T_k\geq 2}]\,.$$
    We conclude by the birthday paradox that
    $$\epsilon_\varsigma=\epsilon_\sigma\leq \frac{t(t-1)}{2}\sum_{k\in [N]}p_k^2\,.$$
    Thus,
    $$\left\Vert\EE_f\ketbra{\phi^f}{\phi^f}^{\otimes t}-\EE_{T\gets \mathcal D_t}\ketbra{\phi_T}{\phi_T}\right\Vert_1\leq 2\sqrt{2}t\sqrt{\sum_{k\in [N]}p_k^2}\,.$$

    In the second part, we compare the LOCC experiment output when the input states are respectively $$\sigma_{2t}=\EE_{S\gets \mathcal D_{2t}}\ketbra{\phi_S}{\phi_S}$$
    and
    $$\sigma_t\otimes \sigma_t=\EE_{T,T'\gets \mathcal D_t} \ketbra{\phi_T}{\phi_T}\otimes \ketbra{\phi_{T'}}{\phi_{T'}}\,.$$

    We have
    \begin{align*}
        \ket{\phi_S}={} & \frac{1}{\sqrt{|S|}}\sum_{\vec k_1,\vec k_2:\type(\vec k_1)+\type(\vec k_2)=S}\bigotimes_{j=1}^t\ket{\phi_{k_{1,j}}}\otimes \bigotimes_{j=1}^t\ket{\phi_{k_2,j}}\\
        ={} & \frac{1}{\sqrt{|S|}}\sum_{T+T'=S}\left(\sum_{\vec k_1:\type(\vec k_1)=T}\bigotimes_{j=1}^t \ket{\phi_{k_{1,j}}}\right)\otimes \left(\sum_{\vec k_2:\type(\vec k_2)=T'}\bigotimes_{j=1}^t \ket{\phi_{k_{2,j}}}\right)\\
        ={} & \sum_{T+T'=S}\sqrt{\frac{|T||T'|}{|S|}}\ket{\phi_T}\ket{\phi_{T'}}\\
        ={} & \sum_{T+T'=S}\sqrt{\frac{\mu_T\mu_{T'}}{\mu_S}}\ket{\phi_T}\ket{\phi_{T'}}\,,
    \end{align*}
    where
    \begin{align*}
    \frac{\mu_T\mu_{T'}}{\mu_S}
    =\frac{\bigl(\prod_i p_i^{T_i+T'_i}\bigr)|T|\,|T'|}
    {\bigl(\prod_i p_i^{S_i}\bigr)|S|}
    =\frac{|T|\,|T'|}{|S|}.
\end{align*}

    We explicitly compute
    \begin{align*}
        \sigma_{2t}={} & \sum_{S}\mu_S\ketbra{\phi_S}{\phi_S}\\
        ={} & \sum_{S}\mu_S\left(\sum_{T+T'=S}\sqrt{\frac{\mu_T\mu_{T'}}{\mu_S}}\ket{\phi_T}\ket{\phi_{T'}}\right)\left(\sum_{U+U'=S}\sqrt{\frac{\mu_U\mu_{U'}}{\mu_S}}\bra{\phi_U}\bra{\phi_{U'}}\right)\\
        ={} & \sum_{\substack{T,T',U,U'\\ T+T'=U+U'}}\sqrt{\mu_T\mu_{T'}\mu_U\mu_{U'}}\ketbra{\phi_T}{\phi_U}\otimes \ketbra{\phi_{T'}}{\phi_{U'}}\,.
    \end{align*}

    Therefore, for any observable $M_i\otimes M_i'$ in the local POVM measurement $M$, noticing that $M_i$ and $M_i'$ are both Hermitian matrices, we have
    \begin{align*}
        &
        \Tr(M_i\otimes M_i'(\sigma_{2t}-\sigma_t\otimes \sigma_t))
        \\
        ={} & 
        \sum_{\substack{T+T'=U+U'\\ T\neq U,T'\neq U'}}\sqrt{\mu_T\mu_{T'}\mu_U\mu_{U'}}\braket{\phi_U|M_i|\phi_T}\braket{\phi_{U'}|M_i'|\phi_{T'}}
        \\
        ={} & \sum_{\substack{T+T'=U+U'\\ T\neq U,T'\neq U'}}\sqrt{\mu_T\mu_{T'}\mu_U\mu_{U'}}\braket{\phi_U|M_i|\phi_T}\overline{\braket{\phi_{T'}|M_i'|\phi_{U'}}}\\
        ={} & \Tr\left(M_i\otimes \overline{M_i'}\sum_{\substack{T+T'=U+U'\\ T\neq U,T'\neq U'}}\sqrt{\mu_T\mu_{T'}\mu_U\mu_{U'}}\ket{\phi_T}\bra{\phi_U}\otimes \overline{\ket{\phi_{U'}}\bra{\phi_{T'}}}\right)\,.
    \end{align*}

    Choose $d_i\in \{1,-1\}$ such that $$\sum_{i\in\mathcal I}d_i\Tr(M_i\otimes M_i'(\sigma_{2t}-\sigma_t\otimes \sigma_t))=\sum_{i\in \mathcal I}|\Tr(M_i\otimes M_i'(\sigma_{2t}-\sigma_t\otimes \sigma_t))|\,.$$
    Then $\sum_i M_i\otimes \overline{M_i'}=I$, and each term $M_i\otimes \overline{M_i'}$ are also positive semidefinite, which means $-I\preceq\sum d_iM_i\otimes \overline{M_i'}\preceq I$, so
    \begin{align*}
        & \sum_{i\in \mathcal I}|\Tr(M_i\otimes M_i'(\sigma_{2t}-\sigma_t\otimes \sigma_t))\\
        ={} & \Tr\left(\sum_{i\in\mathcal I}d_iM_i\otimes \overline{M_i'}\sum_{\substack{T+T'=U+U'\\ T\neq U,T'\neq U'}}\sqrt{\mu_T\mu_{T'}\mu_U\mu_{U'}}\ket{\phi_T}\bra{\phi_U}\otimes \overline{\ket{\phi_{U'}}\bra{\phi_{T'}}}\right)\\
        \leq {}& \left\Vert \sum_{\substack{T+T'=U+U'\\ T\neq U,T'\neq U'}}\sqrt{\mu_T\mu_{T'}\mu_U\mu_{U'}}\ket{\phi_T}\bra{\phi_U}\otimes \overline{\ket{\phi_{U'}}\bra{\phi_{T'}}}\right\Vert_1\,.
    \end{align*}

    Now, let
    $$\Delta=\sum_{\substack{T+T'=U+U'\\ T\neq U,T'\neq U'}}\sqrt{\mu_T\mu_{T'}\mu_U\mu_{U'}}\ket{\phi_T}\bra{\phi_U}\otimes \overline{\ket{\phi_{U'}}\bra{\phi_{T'}}}\,.$$

    Note that $T+T'=U+U'$ if and only if $T-U'=U-T'=:\delta$.
    Define the set
    $$\mathcal T_\delta=\{(T,U):T-U=\delta\}\,.$$
    To simplify, for every element $X=(T,U)\in \mathcal T_\delta$, define a new basis vector
    $$\ket{X}=\ket{\phi_T}\overline{\ket{\phi_U}}\,,$$
    and
    $$\mu_X=\mu_T\mu_U\,.$$
    Then
    \begin{align*}
        \Delta={}&\sum_{\delta}\sum_{\substack{X,Y\in \mathcal T_\delta\\ X\neq Y}}\sqrt{\mu_X\mu_Y}\ketbra{X}{Y}\,.
    \end{align*}
    Thus, let
    $$\Delta_\delta=\sum_{\substack{X,Y\in \mathcal T_\delta\\ X\neq Y}}\sqrt{\mu_X\mu_Y}\ketbra{X}{Y}\,,$$
    we have
    $$||\Delta||_1=\sum_\delta ||\Delta_\delta||_1\,.$$

    It remains to calculate the $1$-norm for each $\Delta_\delta$.
    Define
    $$\mathcal T_{\delta}^{\text{col}}=\{(T,U)\in \mathcal T_\delta: \exists i\in [N], T_i\geq 1, U_i\geq 1\}\,.$$
    Roughly speaking, each element $\mathcal T_\delta^{\text{col}}$ consists of two types that contain at least one collision between them for some non-zero index.
    It is easy to see that for each $\delta$, there is at most one element in $\mathcal T_\delta\backslash \mathcal T_\delta^{\text{col}}$, and $$\sum_{\delta}\sum_{X\in \mathcal T_\delta^{\text{col}}}\mu_X\leq t^2\sum_{i\in [N]}p_i^2\,.$$

    Let
    \begin{align*}
        \Delta_\delta^{\text{col}}={} & \sum_{\substack{X,Y\in \mathcal T_\delta^{\text{col}}\\ X\neq Y}}\sqrt{\mu_X\mu_Y}\ketbra{X}{Y}\\
        ={} & \sum_{X,Y\in \mathcal T_\delta^{\text{col}}}\sqrt{\mu_X\mu_Y}\ketbra{X}{Y}-\sum_{X\in \mathcal T_\delta^{\text{col}}}\mu_X\ketbra{X}{X}\\
        ={} & \sum_{X\in \mathcal T_\delta^{\text{col}}}\sqrt{\mu_X}\ket{X}\sum_{X\in \mathcal T_\delta^{\text{col}}}\sqrt{\mu_X}\bra{X}-\sum_{X\in \mathcal T_\delta^{\text{col}}}\mu_X\ketbra{X}{X}\,.
    \end{align*}
    Notice that in the last line, the first term is a rank $1$ matrix; the second term is a diagonal matrix.
    Then
    \begin{align*}
        ||\Delta_\delta^{\text{col}}||_1\leq {} & \left\Vert\sum_{X\in \mathcal T_\delta^{\text{col}}}\sqrt{\mu_X}\ket{X}\sum_{X\in \mathcal T_\delta^{\text{col}}}\sqrt{\mu_X}\bra{X}\right\Vert_1+\left\Vert\sum_{X\in \mathcal T_\delta^{\text{col}}}\mu_X\ketbra{X}{X}\right\Vert_1\\
        \leq{} & 2\sum_{X\in \mathcal T_\delta^{\text{col}}}\mu_X\,.
    \end{align*}

    For those $\delta$ that $\mathcal T_\delta=\mathcal T_\delta^{\text{col}}$, this suffices to give a bound; however, given $\delta$, if there exists $Z_\delta\in \mathcal T_\delta\backslash \mathcal T_\delta^{\text{col}}$,
    $$\Delta_\delta-\Delta_\delta^{\text{col}}=\sqrt{\mu_{Z_\delta}}\sum_{X\in \mathcal T_\delta^{\text{col}}} \sqrt{\mu_X}(\ketbra{X}{Z_\delta}+\ketbra{Z_\delta}{X})\,.$$
    This is a scalar multiple of a matrix that swaps $\ket{Z_\delta}$ and $\sum_{X\in \mathcal T_\delta^{\text{col}}} \sqrt{\mu_X}\ket{X}$.
    Therefore,
    $$||\Delta_\delta-\Delta_\delta^{\text{col}}||_1=2\sqrt{\mu_{Z_\delta}\sum_{X\in \mathcal T_\delta^{\text{col}}}\mu_X}\,.$$

    Hence,
    \begin{align*}
        ||\Delta||_1\leq{} & \sum_{\delta}||\Delta_\delta^{\text{col}}||_1+\sum_{\delta}||\Delta_\delta-\Delta_\delta^{\text{col}}||_1\\
        \leq{} & 2\sum_{\delta}\sum_{X\in \mathcal T_\delta^{\text{col}}}\mu_X+2\sum_{\delta:\exists Z_\delta} \sqrt{\mu_{Z_\delta}\sum_{X\in \mathcal T_\delta^{\text{col}}}\mu_X}\\
        \leq{} & 2\sum_{\delta}\sum_{X\in \mathcal T_\delta^{\text{col}}}\mu_X+2\sqrt{\sum_{\delta:\exists Z_\delta}\mu_{Z_\delta}}\sqrt{\sum_{\delta}\sum_{X\in \mathcal T_\delta^{\text{col}}}\mu_X}\\
        \leq{} & 2t^2\sum_{i\in [N]}p_i^2+2\sqrt{t^2\sum_{i\in [N]}p_i^2}\,.
    \end{align*}

    And finally, combining the cost of simulating the $\pm 1$ phase state with the type vectors and the $1$-norm bound for the partial transpose of the difference, we get the desired inequality.
\end{proof}

Interestingly, for some arbitrary POVM $\{M_i\}$ and the corresponding output distribution $\mathcal D$ described in Theorem~\ref{thm:phase-concentration}, applying Theorem~\ref{thm:LOCC-decoupling} to $\{M_i\otimes M_j\}_{i,j}$ yields that the $L^2$ variance of $\mathcal D$ is exponentially small if $\max p_k$ is exponentially small.
This does not recover the explicit expression of Theorem~\ref{thm:phase-concentration}, though.

Also, note that the distribution is not $L^1$ concentrated: an example could be the optimal state-estimation POVM or its discretized version; see~\cite{optimal-state-estimation}.

Now we have the concentration theorem~\ref{thm:phase-concentration}.
In order to show that the reflection unitaries in both of our separating oracles can be simulated by a CPTP oracle that samples $(x,g(x))$ for uniformly random $x$, we also need to find a way to simulate the expectation.
Fortunately, this was well studied in literature.
We directly adapt the simulation process in \cite[Lemma 4.3]{CountCrypt}, which in turn follows from \cite[Lemma 4.4]{Goldin-Zhandry}.
We make minor modification to adapt our setup of $\pm1$ phases (instead of $\omega_{2^n}$ phases used in~\cite{CountCrypt}).

We consider the oracle $U_{g_n,R_n}$ in our separation between OWFs and EV-OWPuzzs as an example; the other oracle can be simulated similarly.
Recall Lemma~\ref{thm:JLS-simulation} that any oracle queries to $U_{g_n,R_n}$ can be simulated by querying the state
$$\ket{\phi_{g_n,R_n}^{-}}:=\frac{1}{\sqrt{2}}\ket{0}-\frac{1}{\sqrt{2^{n+1}}}\sum_{x\in \{0,1\}^n}(-1)^{R_n(x)}\ket{1}\ket{x}\ket{g_n(x)}\,.$$
We have the following theorem.

\begin{theorem}\label{thm:GMMY-simulation}
    Let $\mathsf{Samp}$ be a CPTP oracle that acts as follows: it samples $x\gets \{0,1\}^n$, and output $(x,g_n(x))$.
    Then, there exists an efficient oracle algorithm $\mathrm{Sim}^{(\cdot)}(1^t)$ such that for all $t$, we have
    $$\left\Vert\EE_{R_n}\left[\ketbra{\phi_{g_n,R_n}^{-}}{\phi_{g_n,R_n}^-}^{\otimes t}\right]-\mathrm{Sim}^{\mathsf{Samp}}(1^t)\right\Vert_1\leq O\left(\frac{t}{\sqrt{2^n}}\right)\,.$$
\end{theorem}

\begin{proof}
    The simulation algorithm $\mathrm{Sim}(1^t)$ will be defined as follows.
    \begin{enumerate}
        \item Sample $\eta\gets B(t,1/2)$ where $B(t,1/2)$ is the binomial distribution.
        That is,
        $$\Pr[B(t,1/2)=\eta]=\frac{1}{2^t}\binom{t}{\eta}\,.$$
        \item Query $\mathsf{Samp}$ a total of $\eta$ times, receiving outputs $(x_1,g_n(x_1)),(x_2,g_n(x_2)),\cdots,(x_\eta,g_n(x_\eta))$.
        \item Let $k_i=1||x_i$ for $i\leq \eta$, and $k_i=0$ for $i>\eta$.
        Define $$\ket{\phi_k}=\begin{cases}
            \ket{g_n(k_{[2:n+1]})} & k_1\neq 0\\
            \ket{0} & k_1=0
        \end{cases}\,.$$
        \item Generate the state
        $$\left(\frac{1}{\sqrt{t!}}\sum_{\pi\in S_t}\ket{\pi}_{\mathbf A}\right)\otimes \left(\bigotimes_{j=1}^t \ket{k_j}_{\mathbf{C_j}}\ket{\phi_{k_j}}_{\mathbf{D_j}}\right)\,.$$
        This can be done via an efficient QFT on $S_t$.
        \item Controlled by the register $\mathbf A$, apply the corresponding permutation on the registers $\mathbf{C_j}$ and $\mathbf{D_j}$.
        \item Uncompute the register $\mathbf A$.
        This can be done by comparing the values in registers $\mathbf {C_j}$ and $\vec k$.
        After that, the state on $\mathbf{C_j},\mathbf{D_j}$ registers is proportional to
        $$\sum_{\vec k':\type(\vec k')=\type(\vec k)}\bigotimes_{j=1}^t \ket{k_j'}\ket{\phi_{k_j'}}\,.$$
        The simulator $\mathrm{Sim}^{\mathsf{Samp}}(1^t)$ outputs the $\mathbf{C_j},\mathbf{D_j}$ registers.
    \end{enumerate}

    It is clear that the output of $\mathrm{Sim}^{\mathsf{Samp}}(1^t)$ is only related to the type of $\vec k$.
    For a type $T$, define
    $$|T|=\sum_{\vec k:\type(\vec k)=T}1\,.$$
    Define the distribution $\mathcal D$ over types as
    $$\Pr_{\mathcal D}[T]=\mu_T:=\frac{|T|}{2^{T_0}(2^{n+1})^{t-T_0}}\,,$$
    and
    $$\ket{\phi_T}=\frac{1}{\sqrt{|T|}}\sum_{\vec k:\type(\vec k)=T}\bigotimes_{j=1}^t \ket{k_j}\ket{\phi_{k_j}}\,.$$
    Then clearly
    $$\mathrm{Sim}^{\mathsf{Samp}}(1^t)=\EE_{T\gets \mathcal D}[\ketbra{\phi_T}{\phi_T}]\,.$$
    And the proof that $\ketbra{\phi_{g_n,R_n}^-}{\phi_{g_n,R_n}^-}^{\otimes t}$ is close to $\mathrm{Sim}^{\mathsf{Samp}}(1^t)$ in trace distance is just a recurrence of the proof in Theorem~\ref{thm:LOCC-decoupling}.
\end{proof}

\section{Proof of the Hilbert-valued Efron--Stein Inequality}
\label{sec:Appendix-efron-stein}
In this section, we present another technical contribution of this paper: the Hilbert-valued Efron--Stein inequality~\ref{thm:efron-stein}.
It is a straightforward generalization of the (real-valued) Efron--Stein inequality in Lemma~\ref{lem:Efron-stein-original}.

It is crucial that in Theorem~\ref{thm:efron-stein}, the range of $f$ is in a Hilbert space.
In fact, the proof uses the Parseval's indentity (and essentially, the Pythagorean theorem) in a Hilbert space, which cannot be generalized to an arbitrary Banach space.
Since the only $L^p$ space with an inner product is $L^2$, the Efron--Stein inequality~\ref{thm:efron-stein} yields an $L^2$ concentration, which is the fundamental reason that our concentration theorem does not contradict to the existence of EV-OWPuzz.

\begin{theorem}[Hilbert-valued Efron--Stein Inequality]\label{thm:efron-stein}
    Let $X_1,\cdots,X_n$ be mutually independent random variables taking values in the measurable space $\mathcal X$, and let $H$ be a real separable Hilbert space.
    Let $f:\mathcal X^n\to H$ be a measurable, square-integrable function.
    Define
    $$X=(X_1,\cdots,X_n)\,,\qquad Z=f(X)\,.$$
    For each $i$, let $X_i'$ be an independent copy of $X_i$, independent of all other random variables, and define
    $$X^{(i)}=(X_1,\cdots,X_{i-1},X_i',X_{i+1},\cdots,X_n)\,,
    \qquad
    Z^{(i)}=f(X^{(i)})\,.$$
    Then
    $$\Var(Z):=\EE||Z-\EE Z||_H^2\leq \frac{1}{2}\sum_{i=1}^n \EE||Z-Z^{(i)}||_H^2\,.$$
\end{theorem}

\begin{proof}
    Let $\{e_1,e_2,\cdots\}$ be a set of orthonormal basis of the Hilbert space $H$.
    Consider the function
    $$f_k(\cdot)=\left\langle f(\cdot),e_k \right\rangle_H$$
    By the vanilla Efron--Stein inequality~\ref{lem:Efron-stein-original},
    $$\Var(f_k(X))=\EE\left[|f_k(X)-\EE f_k(X)|^2\right]\leq \frac{1}{2}\sum_{i=1}^n \EE[(f_k(X)-f_k(X^{(i)}))^2]\,.$$
    By the linearity of inner product,
    $$\EE\left[\left\langle Z-\EE Z,e_k \right\rangle_H^2\right]\leq \frac{1}{2}\sum_{i=1}^n \EE\left[\langle Z-Z^{(i)},e_k\rangle^2\right]\,.$$
    Take the summation over $k$ on both sides.
    Note that both sides are nonnegative, so one can switch the order of expectation and summation:
    $$\EE\left[\sum_k\left\langle Z-\EE Z,e_k \right\rangle_H^2\right]\leq \frac{1}{2}\sum_{i=1}^n \EE\left[\sum_k\langle Z-Z^{(i)},e_k\rangle^2\right]\,.$$
    Parseval's identity yields
    $$\EE||Z-\EE Z||_H^2\leq \frac{1}{2}\sum_{i=1}^n \EE||Z-Z^{(i)}||_H^2\,.$$
\end{proof}

\end{document}